\documentclass[10pt,journal,compsoc,twoside]{IEEEtran}
\usepackage{graphicx}
\usepackage{balance}
\usepackage{comment}
\usepackage{url}
\usepackage{epsfig,tabularx,amssymb,amsmath}
\usepackage[noend]{algorithmic}
\usepackage{algorithm}
\usepackage{subfigure,multirow,graphicx}
\usepackage{setspace}
\usepackage{amsthm}
\usepackage{color} 
\definecolor{shadecolor}{gray}{0.85}
\usepackage{diagbox}
\usepackage{array}
\newcolumntype{L}[1]{>{\vspace{0.5em}\begin{minipage}{#1}\raggedright\let\newline\\
\arraybackslash\hspace{0pt}}m{#1}<{\end{minipage}\vspace{0.5em}}}
\newcolumntype{R}[1]{>{\vspace{0.5em}\begin{minipage}{#1}\raggedleft\let\newline\\
\arraybackslash\hspace{0pt}}m{#1}<{\end{minipage}\vspace{0.5em}}}
\newcolumntype{C}[1]{>{\vspace{0.5em}\begin{minipage}{#1}\centering\let\newline\\
\arraybackslash\hspace{0pt}}m{#1}<{\end{minipage}\vspace{0.5em}}}

\newtheorem{lemma}{Lemma}
\floatname{algorithm}{Procedure}

\newcommand{\algorithmicbreak}{\textbf{break}}
\newcommand{\algorithmiccontinue}{\textbf{continue}}

\newcommand{\BREAK}{\STATE{\algorithmicbreak}}
\newcommand{\CONTINUE}{\STATE{\algorithmiccontinue}}
\newtheorem{theorem}{Theorem}

\usepackage{caption}
\usepackage{fancyhdr}

\ifCLASSINFOpdf
\else
\fi

\ifCLASSOPTIONcompsoc
  \usepackage[nocompress]{cite}
\else
  \usepackage{cite}
\fi

\begin{document}
%
\title{Triangle Extension: Efficient Localizability Detection in Wireless Sensor Networks}

\author{
    Hejun~Wu~\IEEEmembership{Member~IEEE,} Ao~Ding, Weiwei~Liu, Lvzhou~Li
    and~Zheng~Yang~\IEEEmembership{Member~IEEE}

\IEEEcompsocitemizethanks{
\IEEEcompsocthanksitem Hejun~Wu (Corresponding author), Ao~Ding and Lvzhou~Li are with Guangdong Key Laboratory of Big Data Analysis and Processing, Department
of Computer Science, Sun Yat-sen University, Guangzhou, China.
~~E-mail: wuhejun@mail.sysu.edu.cn, dingao@mail2.sysu.edu.cn, lilvzh@mail.sysu.edu.cn \protect

\IEEEcompsocthanksitem Weiwei~Liu is with Horizon Robotics, Beijing, China.~~E-mail: weiwei.liu@hobot.cc \protect
\IEEEcompsocthanksitem Zheng~Yang is with the School of Software,
Tsinghua National Laboratory for Information Science and Technology,
Tsinghua University, Beijing, China.
~~E-mail: yangzheng@tsinghua.edu.cn \protect\\
}
\thanks{Manuscript received  February 23, 2017; revised August 28, 2015.}}

\pagestyle{fancy}
\fancyhead[L]{Wu \MakeLowercase{\textit{et al.}}: Triangle Extension: Efficient Localizability Detection in Wireless Sensor Networks}

\IEEEtitleabstractindextext{%
\begin{abstract}

Determining whether nodes can be localized, called localizability detection, is essential for { wireless sensor networks} (WSNs).  This step is required { for} localizing nodes, achieving low-cost deployments, and identifying prerequisites in location-based applications. Centralized graph algorithms are inapplicable to a resource-limited WSN { because} of their high computation and communication costs, whereas distributed approaches may miss a large number of theoretically localizable nodes in a resource-limited WSN. In this paper, we propose an efficient and effective distributed approach { in order} to address this problem. Furthermore, we prove the correctness of our algorithm and analyze the reasons our algorithm can find more localizable nodes while requiring fewer known location nodes than existing algorithms, under the same network configurations. The time complexity of our algorithm is linear with respect to the number of nodes { in a network}. We conduct both simulations and { real-world} WSN experiments to evaluate our algorithm under various network settings. The results show that our algorithm significantly outperforms the existing algorithms { in terms of} both the latency and the accuracy of localizability detection.
\end{abstract}

}

\maketitle

\IEEEdisplaynontitleabstractindextext

 \ifCLASSOPTIONpeerreview
 \begin{center} \bfseries EDICS Category: 3-BBND \end{center}
 \fi
%
\IEEEraisesectionheading{\section{Introduction}\label{sec:intro}}


\IEEEPARstart{I}{n}  { wireless sensor networks} (WSNs), { owing to} the high hardware and/or energy cost and { indoor} blindness of GPS components, localization algorithms are often required \cite{conf/vtc/AlmuzainiG11}, \cite{journals/jsac/LazosP06}, \cite{DVhop}, \cite{rangebased}, \cite{uwsn}. Such algorithms employ beacons, which are special nodes of known locations, to determine the unknown locations of the other nodes in a WSN \cite{DBLP:journals/tmc/ChenLLP13}, \cite{DBLP:conf/mobicom/PriyanthaCB00}, \cite{127405}.  {An essential problem in localization algorithms is to detect whether a WSN or a node in the WSN is localizable \cite{DBLP:journals/tmc/YangL12}.  For instance, most localization algorithms can only localize a relatively small percentage of nodes in a WSN, especially in sparsely deployed WSNs. A localizability detection algorithm can help such localization algorithms { avoid} non-stopping failures or { incorrect} localization answers. Localizability information also gives guidelines or { identifies} prerequisites to { location-related} applications, { e.g., tracking and event detection}. Finally, {  node localizability information} is helpful for many mechanisms of WSNs, such as topology control, sensing area adaptation, and geographic routing.}

{{ The determination of} whether a network is localizable is known as the \emph{network localizability} problem, and { the detection of} whether a node is localizable is called the \emph{node localizability} problem.} The network localizability problem is formulated as follows \cite{DBLP:journals/tmc/YangL12}.  A WSN is modeled as a connected graph: $G=(V,E)$, in which $i$ and $j$ {  are nodes and $(i, j)$ are the link between them} ($i$, $j \in V$, $(i, j)\in E$).   {In $G$, the distance between $i$ and $j$, denoted as $d_{(i, j)}$, is known or can be measured,  {for example by} using distance measuring methods \cite{Han-JS16}, \cite{Shangguan-17}, \cite{Yang-15} . The location of a node $i$, denoted as $p_i$, is known (beacon) or unknown (non-beacon). A graph $G$ with a constraint set $C$, e.g., a set that specifies the locations of the beacon nodes,  is \emph{localizable} when the following holds: For each node $i$ in $G$, { there is } a unique $p_i$, such that $d_{(i, j)} = distance(p_i, p_j)$ for all $(i,j) \in E$} and the constraint set $C$ is satisfied. Given such definitions, researchers have found a close relationship between the network localizability problem and the graph rigidity { problem}\cite{newtrilateration}, \cite{trilateration},  \cite{bruce}, \cite{matroid}, \cite{MR42:4430}, \cite{DBLP:journals/jct/JacksonJ05}.  A rigid graph has a finite number of frameworks  { with which it can be realized }\cite{MR42:4430}.

{ Although the graph rigidity testing algorithms are theoretically sound, such algorithms are not applicable to WSNs}, as they only output \emph{false} in most  { real-world} WSNs\cite{DBLP:journals/tmc/YangL12}. Considering the realistic issues, researchers explore the \emph{node localizability} problem instead of {  the network localizability problem}. The following  { is} an example of  { localizability testing, for} node 1 in Fig.\ref{realizations}. There are four nodes and nodes 2, 3, and 4 are beacons in Fig.\ref{realizations}. Note that in Fig.\ref{realizations} and other network figures throughout this paper, an edge between two nodes {  indicates that} the distance between the two nodes is known and fixed. As can be seen in Fig.\ref{realizations}, for node 1, there are two possible locations that  {  each satisfies} these constraints. Each possible location of node 1  { may }have a possible network framework. Fig.\ref{realizations} shows two frameworks  { for} the network. In conclusion, node 1 is not localizable.

\begin{figure}[ht]
\centering
\includegraphics [width=4cm]{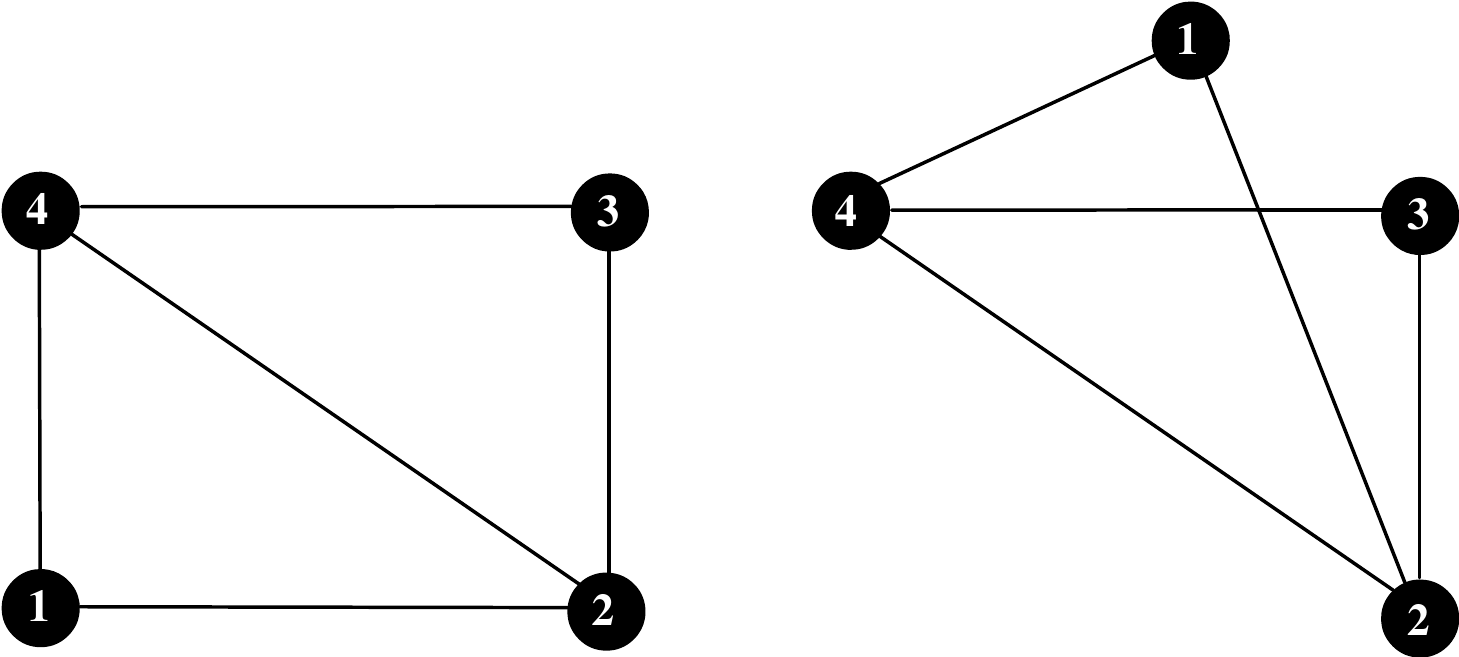}
\caption{Two different frameworks of a network} \label{realizations}
\end{figure}

 \emph{Node} localizability differs significantly from \emph{network} localizability. The  { approach} of simply partitioning a graph to a set of sub-graphs and detecting the localizability of the sub-graphs does not work for determining node localizability, as a partition may remove some constraints of the original network \cite{DBLP:journals/tmc/YangL12}. The problem is even more difficult when it is required to test the node localizability in a distributed way. A centralized approach may exhaust the energy of each node {early in the process}, as a packet of a source node usually needs to be forwarded many hops before it arrives at the center for processing. This  problem is so difficult that no distributed solution {has yet been found that can identify} all theoretically localizable nodes.

{In this paper, we propose an efficient and effective distributed algorithm to address the node localizability problem. Specifically, our algorithm employs a new method to  { extend} each graph. {  In a graph, the extension of our algorithm starts from  { a sub-graph having only }two beacon nodes and continues until the sub-graph reaches another beacon}. Then, our algorithm determines whether the extended graph can be localized, according to graph rigidity theory. In particular, this paper makes the following contributions: (1) The proposed graph extension method is theoretically proved to be able to find localizable nodes with  { fewer} beacons than the existing methods {  do}. (2) The complexity of the distributed algorithm is  { shown to be} linear to the number of nodes in a WSN. (3)  { Because of its low resource requirements}, the algorithm is { shown to be} applicable to a sparsely deployed WSN. To our knowledge, the algorithm proposed in this paper is currently the best solution to the node localizability problem.}

\section{Preliminaries}

\subsection{Graph Rigidity}
A rigid graph is a connected graph that has a finite number of frameworks \cite{MR42:4430}.  Furthermore, if a graph has a unique framework, it is called \emph{globally~rigid} \cite{bruce}. The graph shown in Fig.\ref{rigid-flexible-grigid}(a) is rigid, as the graph has two and only two different frameworks in a 2-dimension($2D$) plane. In contrast, the graph shown in Fig.\ref{rigid-flexible-grigid}(b) has an infinite number of frameworks;  { this} is called a flexible graph. The graph in Fig.\ref{rigid-flexible-grigid}(c) is globally rigid as it has only one framework. { In the graphs, an edge between two nodes indicates that the distance between the two nodes is known.}

\begin{figure}[ht]
\centering
\begin{minipage}{0.1 \textwidth}
\centering
\includegraphics [width=2 cm]{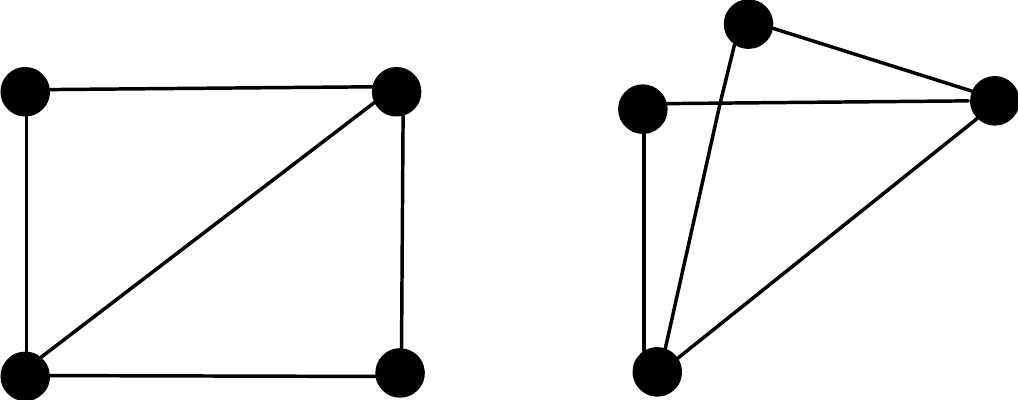}
\footnotesize\text{(a)}
\end{minipage}
\begin{minipage}{0.051 \textwidth}
\centering
\includegraphics [width=0.4 cm]{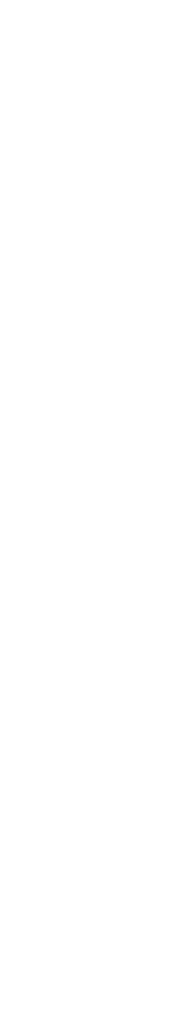}
\end{minipage}
\begin{minipage}{0.2\textwidth}
\centering
\includegraphics [width=3.5 cm]{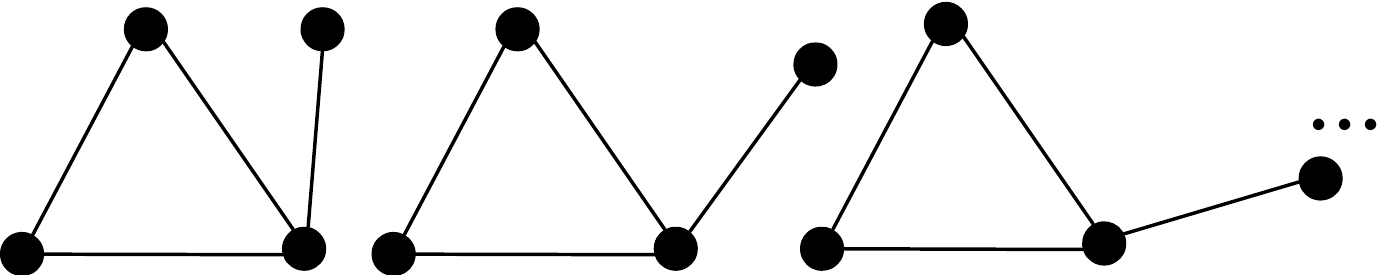}
\footnotesize\text{(b)}
\end{minipage}
\begin{minipage}{0.12\textwidth}
\centering
\includegraphics [width=0.65 cm, height = 0.65 cm]{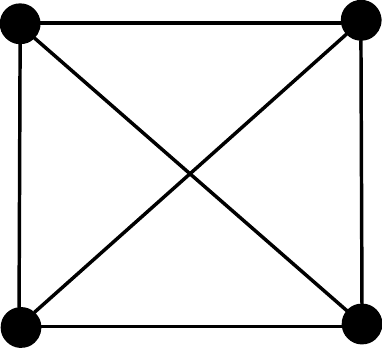}

\footnotesize\text{(c) }
\end{minipage}
\caption{Graphs: (a) rigid, (b) flexible, and (c) globally rigid} 
\label{rigid-flexible-grigid}
\end{figure}

Suppose a graph $G$ satisfies the
following two conditions: (1) { it is} globally rigid, (2)  {  it contains} three or more beacons that are non-collinear (i.e., they are not  { on} the same line). Then, all the nodes in $G$ can { be identified as} localizable. Hence, the key issue  { for} detecting localizability of a network is  { the detection of} global rigidity.

\subsection{Global Rigidity}

We { first} review Hendrickson's Lemma for the principle of detecting global rigidity. In Hendrickson's Lemma, a \emph{skeletal sub-graph} of $G$ contains all the vertexes of $G$. Then, we go { on to the} M-circuit  \cite{bruce}, to learn how to construct a globally rigid graph. The notation is listed in Table \ref{tablenotations}.

\noindent\textbf{Hendrickson's Lemma \cite{bruce}:} \textit{A graph
is globally rigid if it is 3-connected and has a skeletal
sub-graph that is a M-circuit.} 

\noindent\textbf{Definition of M-circuit \cite{bruce}:}
Graph $G$ is called an \emph{M-circuit} 
if it satisfies the following two conditions:\\
(1) $|E_{G}|=2|V_{G}|-2$;\\
(2) for each $X\subseteq V_{G}$ with $2\leq |X| \leq |V_{G}|-1$, $|E[X]|\leq 2|X|-3$.\\

\begin{table}[ht]\footnotesize
\caption{Notation in this paper}
\centering
\begin{spacing}{1.1}
\label{tablenotations}
\begin{tabular}{|c|c|}
\hline
Notation & Description\\
\hline
G & An undirected and connected graph \\
\hline
$V_{G}$ & Vertex set of graph $G$ \\
\hline
$E_{G}$ & Edge set of graph $G$\\
\hline
$|V| $ & The number of nodes in $V$\\
\hline
$G[X]$ & A sub-graph of $G$ induced by $X$; $X\subseteq V_{G}$ \\
\hline
$E[X]$ & The edge set of $G[X]$ \\
\hline
$|E[X]|$ & The number of edges in $G[X]$ \\
\hline
$L_x$ & The level of node $x$\\
\hline
\end{tabular}
\end{spacing}
\end{table}

According to Laman's Lemma \cite{MR42:4430}, a rigid graph can also be further broken down to a \emph{minimally rigid graph}. Laman proved that every rigid graph $G$ has a skeletal sub-graph $G'$ that is minimally rigid. Hence, each rigid graph can be reduced { to} a minimally rigid sub-graph by removing certain edges.

\noindent{Laman's Lemma\cite{MR42:4430}}: \textit{A graph
$G$ is minimally rigid { if and only if} $|E_{G}| = 2|V_{G}|-3$ and for each $X\subset V_{G}$ with $2\leq |X| \leq |V_{G}|-1$, $|E[X]|\leq 2 |X|-3$.}\\


\section{Related Work}

A series of methods have been proposed to determine the rigidity of a
graph \cite{Alex}, \cite{matroid}. However, as most of these methods are centralized, they are not suitable for a real world WSN. A centralized algorithm needs the global information of the network topology to construct the adjacency matrix. However, { because of the memory limit restriction} on a node, it is unrealistic for each node to maintain the global topology information in a WSN, especially { a large-scale} WSN.

Yang and Liu previously proposed a theoretical approach { for the detection of }localizable nodes, called RR3P \cite{DBLP:journals/tmc/YangL12}. According to RR3P, a node in the network is localizable if the following two conditions hold: (1) the node belongs to a redundantly rigid component, and (2) in this component there { exist} at least three vertex-disjoint paths connecting the node to three distinct beacons. RR3P is not { very} suitable for distributed algorithms, as finding { a} redundantly rigid component requires { near-global} topology information about a network.

The differences between RR3P and our approach can be summarized as follows: (1) We propose a realistic distributed algorithm to find the localizable nodes starting from pairs of beacons. (2) Our distributed algorithm is able to detect { most of the} localizable nodes without finding the redundantly rigid component which is the prerequisite in RR3P. (3) In our distributed algorithm, each node only uses the information from one-hop communication. In comparison, RR3P needs to detect the three vertex-disjoint paths, which usually span { several} hops.

Eren et al. proposed a distributed algorithm,
{ the} $trilateration~protocol$ (TP) \cite{trilateration}, to mark the localizable nodes and calculate their locations in a
network. The general idea of TP is that a node can be regarded as localizable, whenever the node { can identify} three or more localizable neighbors. This idea enables TP to be fully distributed. Hence, TP { has been} used by many applications \cite{
}, \cite{Savvides01a}.  Eren recently proposed the concepts of indices to quantitatively measure the network graph rigidity \cite{newtrilateration}. Such concepts are also helpful to these applications.

However, TP may miss localizable nodes { that} are on a $geographical ~ gap$ or in a graph with $border~nodes$. Fig.\ref{problem}(a) is a graph with a geographical gap.
The nodes without labels are detected {  to be} localizable. Nodes $A$, $B$, and $C$ are theoretically localizable. However, none of the nodes among $A,B,C$ will be determined {  to be localizable} by TP, because {  none of them has three or more localizable neighbors}. In Fig.\ref{problem}(b), nodes $A$ and $B$ are border nodes.
They can not be determined {  to be} localizable by TP, either.

\begin{figure}[ht]
\centering
\begin{minipage}{0.23\textwidth}
\centering
\includegraphics [width=3.38cm, height = 2.2cm]{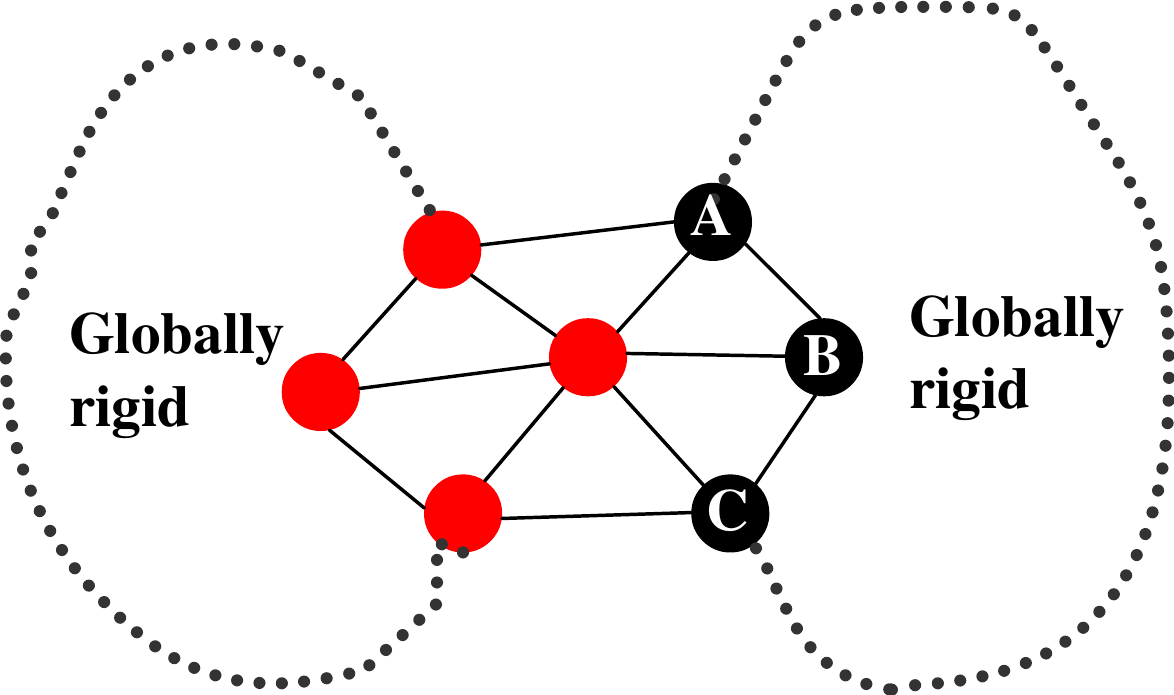}

\footnotesize\text{(a)Geographical gap}
\end{minipage}
\begin{minipage}{0.23\textwidth}
\centering
\includegraphics [width=3.35cm]{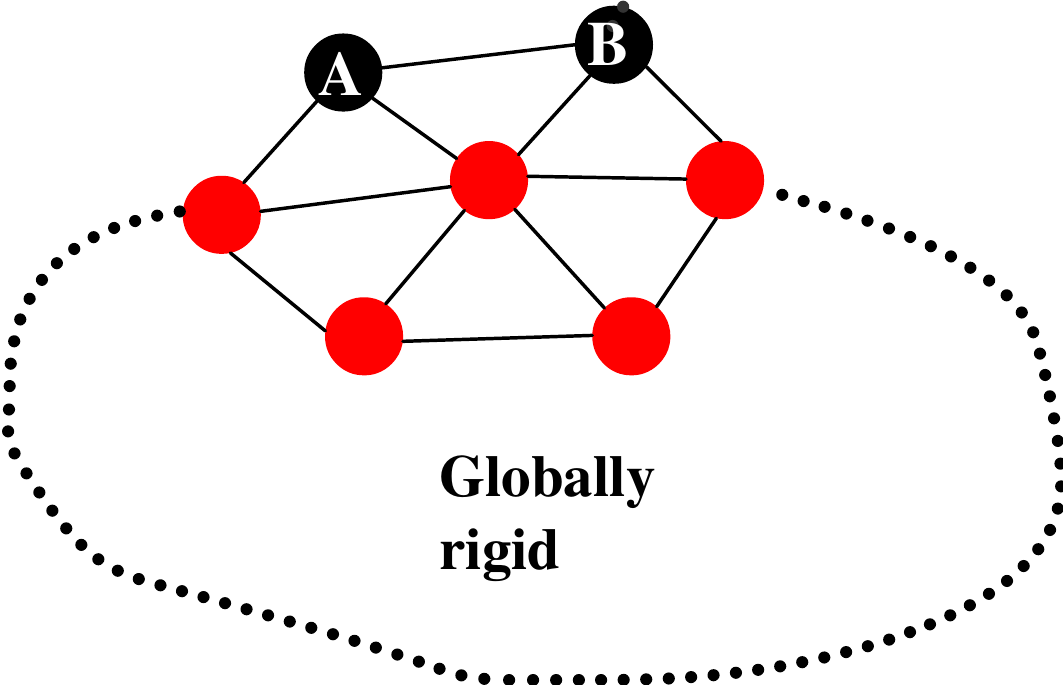}

\footnotesize\text{(b)Border nodes}
\end{minipage}
\caption{Deficiency of trilateration protocol} \label{problem}
\end{figure}

To solve the problems of { a} geographical gap and { of} border nodes, Yang et al. proposed a new fully distributed approach { for detecting} localizable nodes, called $wheel~ extension$ (WE) \cite{wheel}. In a WSN running WE, each node tries to construct wheel graphs within its neighbors. During the extension process, each node only needs to check the constructed wheel graphs. If there are three or more nodes that are localizable in a wheel graph, WE marks all the other nodes in the wheel graph as localizable.

Nonetheless, there are often scenarios { in which} WE may not work well, either. For example, { it may be the case that beacons are not present in} any wheel graphs in a network. {  Another scenario is one in which a} node cannot build a wheel graph,  {  because of the lack of adequate network information for the node}. As shown in Fig.\ref{circle},  {  none of} the nodes in graph $G_{c}$ are in a wheel graph and there is no node that has three or more localizable neighbors. Therefore, neither TP nor WE can find any localizable node in graph $G_{c}$. However, as we will show later, all the nodes in a graph such as $G_{c}$ are localizable.

\begin{figure}[ht]
\centering
\includegraphics [width=0.25\textwidth]{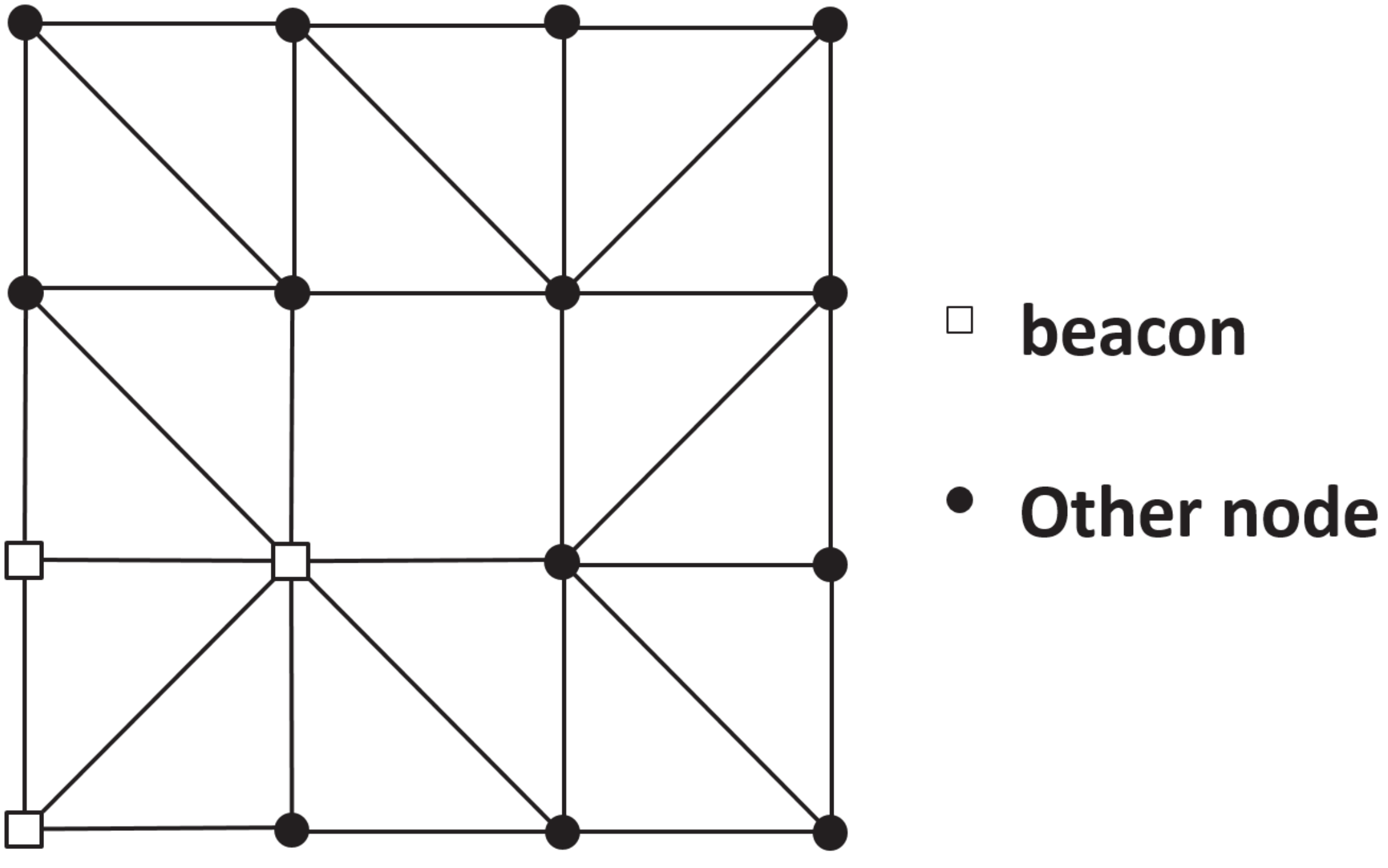}
\caption{$G_{c}$, A globally rigid graph.} \label{circle}
\end{figure}

 { In contrast with WE and TP}, our algorithm uses a new detection method, $triangle~extension$ (TE), to find the localizable nodes in a WSN. TE can start with  { far} fewer beacons than the existing algorithms. Moreover, TE not only works well { for a graph such as} $G_{c}$ in Fig.\ref{circle}, but also avoids the problems shown in Fig.\ref{problem}. 
 

\section{{Triangle Extension Theory}}

This section presents our theoretical approach to node localizability detection via global rigidity as follows: We define {  the concept} of a \emph{branch} and { propose related lemmas}. These concepts and lemmas enable us to construct a graph that is an M-circuit and globally rigid, { starting from a branch}.  

\subsection{Branch Concept}
  {We introduce three concepts { that will enable us to} define a branch. The first is an operation called \emph{extension}. The second is a special { kind of} extension: \emph{triangle extension}. \emph{Triangle block} is the third concept. A branch will then be constructed within a triangle block. }\\

\textbf{Extension}:  Given a graph $G$, an \emph{extension} operation on $G$ is { one that inserts} a new vertex $v$ into $V_{G}$ and two edges $(v,r_{1}),(v, r_{2})$ ($r_{1},r_{2}\in V_{G}$) into $E_{G}$. Here, nodes $r_{1}$ and $r_{2}$ are { said to be} \emph{extended} by node $v$. Nodes $r_{1} and r_{2}$ are called the \emph{parents} of $v$ and $v$ {  is called} a \emph{child} of $r_{1} and r_{2}$. The ancestors of node $v$ are recursively defined as $v$'s parents ($r_{1} and r_{2}$) and the \emph{ancestors} of $v$'s parents.

Lemma \ref{lextension} gives the property of extensions. Its proof is in Appendix A.  Following Lemma \ref{lextension}, when performing a series of extensions from a minimally rigid graph $K_{2}$, we can obtain a larger minimally rigid graph. $K_{2}$ is a complete graph of two nodes and is minimally rigid according to Laman's Lemma \cite{MR42:4430}. Fig.\ref{extension} { shows the} extensions from a $K_{2}$ graph(the leftmost graph with only $r_{1}$ and $r_{2}$). In this figure, the graph { is} extended by nodes $a$, $b$, $c$, $d$, and $e$, after five extension operations. The final { resulting} graph is still minimally rigid. 

\begin{lemma}
\label{lextension}
A graph $G_1$, generated from a series of extensions on a minimally rigid graph $G$, is also minimally rigid.
\end{lemma}

\begin{figure}[ht]
\centering
\includegraphics [width=0.49\textwidth]{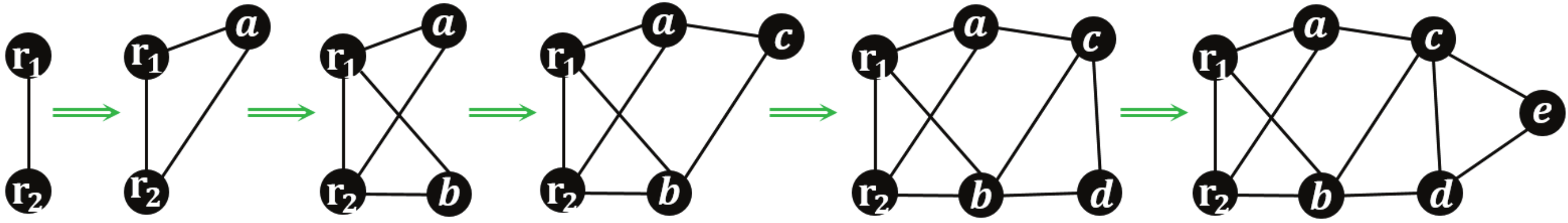}
\caption{A series of extensions} \label{extension}
\end{figure}

 {\textbf{Triangle Extension:}}  A triangle extension operation is a sequence of extensions, the first of which starts from a $K_{2}$. As shown in Fig.\ref{extension}, the extension sequence launched by nodes \{$a$, $b$, $c$\} {and that launched by nodes} \{$a$, $b$, $c$, $d$, $e$\} are both triangle extensions, but the extension sequence of \{$d$, $e$\} is not a triangle extension, as it does not start from a $K_{2}$.

  {\textbf{Triangle Block:} }A triangle block is a graph { constructed from a  $K_{2}$ graph by a \emph{triangle extension}}. The nodes in the original $K_{2}$ are called the \emph{roots} of this triangle block. In a triangle block, a node $x$ has a level denoted as $L_x$. If $x$ is a root, $L_x$ = 0; otherwise, $L_x$ = 1 + max(the levels of $x$'s parents). { The value of the level} indicates the temporal order of the extension. 

  {\textbf{Branch:}} In a triangle block $T$,  a \emph{branch} of $T$ { is a sub-graph of $T$} that is composed of a node $v$, the ancestors of $v$, and the edges between them. Node $v$ is called the \emph{leaf node} and the branch is denoted using $B(v)$ to indicate that the leaf node is the end of the triangle extension. The roots $r_{1}$ and $r_{2}$ of $T$ are also the roots of $B(v)$. Fig.\ref{branches} shows the four branches $B(c)$, $B(d)$, $B(e)$, and $B(f)$ of $T$.

\begin{figure}[ht]
\centering
\includegraphics [width=0.3\textwidth, height = 3cm]{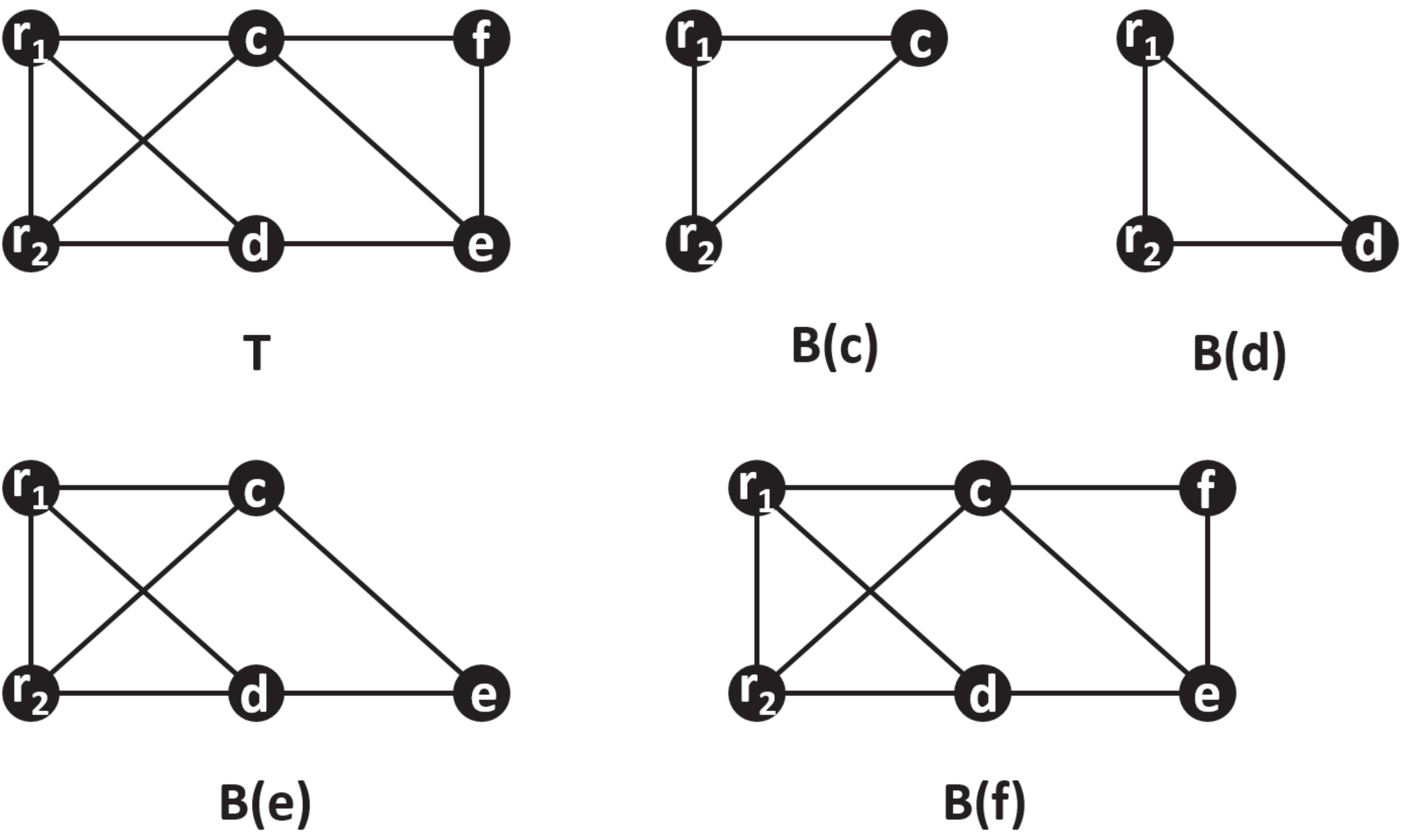}
\caption{Example triangle block and its branches} \label{branches}
\end{figure}

  {A branch has two following properties. \emph{Property-1}: A branch is minimally rigid. \emph{Property-2}: Removing any single node from a branch will not { cause the branch to become} disconnected. Property-1 can be derived from Lemma \ref{lextension} as a branch is extended from a $K_2$.  Property-2 can be proved as follows: { It is evident that} removing the leaf or one of the roots will not break a branch { into} two separate sub-graphs. Given a non-root and non-leaf node $x$ in a branch, the number of edges connected to $x$ is at least 3. The reason is that node $x$ must have two parents and at least one child connected to itself in a branch. Hence, removing $x$ will not {cause the branch to become} disconnected. Property-2 holds.}

\begin{lemma}
\label{mylemma1}
\text{{Branch Lemma}: }
In a branch $B(v)=(V_{B},E_{B})$, there is no set $X$ that satisfies all of the following three conditions.\\
(1) $X\subsetneqq V_{B}$, where $V_{B}$ is the vertices of $B(v)$.\\
(2) $X$ contains the two roots of $B(v)$ and $v$.\\
(3) $|E[X]|=2|X|-3$.\\
\end{lemma}
\vspace{-0.5cm}

\begin{proof} 
In the scenario of $L_v =1$, i.e., node $v$ is at level 1, $B(v)$ only contains $v$ and the two roots, which are also the parents of $v$. { It is evident that} conditions (1) and (2) cannot be satisfied at the same time.

In the scenarios of $L_v\geq 2$, we prove Lemma \ref{mylemma1} by contradiction: Suppose that there is a set $X$ satisfying all three conditions.  Let $Y= V_{B}-X$. 

  {When there is only a single node $y$ in $Y$, $y$ is { neither a root nor} the leaf node since they are in $X$. The number of edges connected to $y$ is at least three. Subsequently, as $|E[X]|=2|X|-3$, $|E[X \cup y]|\geq 2|X \cup y|-2$. On the other hand, since $B(v)$ is a minimally rigid graph (Branch Property-1), according to Laman's Lemma, for each $X'\subset V_{G}$ with $2\leq |X'| \leq |V_{G}|-1$, $|E[X']|\leq 2|X'|-3$. Let $X' = X \cup y$; then $|E[X \cup y]|\leq 2|X \cup y|-3$ contradicts the previous deduction: $E[X \cup y]\geq 2|X \cup y|-2$. }

  {When there { is} more than one node in $Y$, { we proceed as follows:} Let $ym$ be { the node in $Y$ having the smallest level}. As $ym$ has at least two parents in $X$, the number of edges connected to $ym$ is at least two. We move $ym$ from $Y$ to $X$ and { obtain} $|E[X]|\geq2|X|-3$. Similarly, we can move each node with the smallest level from the remainder nodes in $Y$ to $X$ until there is only a single node  in $Y$, denoted as $ym'$. Then there are at least three edges connected to $ym'$, since the parents of $ym'$ and {  the children of $ym'$} are all in $X$ now. After we move $ym'$ into $X$, we have $|E[X]|\geq 2|X|-2$. This is in conflict with Laman's lemma.}

In summary, {regardless of whether} $L_v=1$ or $L_v\geq 2$, there is no $X$ that satisfies all of the three conditions.
\end{proof}
\vspace{-0.5cm}
\subsection{Globally Rigid Graph Construction Theorem}

  {We now start to build a globally rigid graph from a branch. The steps are as follows: First, we deduce Lemma \ref{branchlemma2} to obtain a 3-connected graph from a branch. Then, we prove that the constructed graph is an M-circuit, according to Lemma \ref{mylemma1} of the last subsection. As a result, { the two necessary conditions for obtaining} a global rigidity graph are met.}

\begin{lemma}
\label{branchlemma2}
In a branch $B(v)$, the removal of two nodes $t_{1}$ and $t_{2}$, where $v\notin \{t_{1},t_{2}\}$ and there is at most one root in $\{t_{1},t_{2}\}$ divides $B(v)$ into at most two sub-graphs. If after the removal of $t_{1}$ and $t_{2}$, $B(v)$ is divided, the leaf node $v$ and the remaining root(s) are in different sub-graphs.
\end{lemma}

\begin{proof}
As a branch is a 2-connected graph, the removal of a single node will not divide the branch { into} two separate sub-graphs. Then removing another node may divide the graph  { into} at most two separate sub-graphs. Furthermore, the removal of two nodes  { cannot} divide a branch of three nodes into two sub-graphs. The { following} are the three possible scenarios for $t_{1}$ and $t_{2}$.

(I) {$\{t_{1},t_{2}\}$ does not contain any root.
We consider a node $t$, $t \in \{t_{1},t_{2}\}$. As $t$ has two different parents and each ancestor of $t$ has two different parents, there are two different routes with no intersection from $t$'s two parents to the roots, $r_{1}$ and $ r_{2}$. Similarly, the routes from $v$'s two parents to  $r_{1},r_{2}$ have no intersection either. Fig.\ref{branch_pics}(a) shows the { branch}. The three routes, { those} from $t$ to $v$, from $t$ to $r_{1}$, and from $t$ to $r_{2}$, only intersect at $t$. The { dash-dotted} lines in Fig.\ref{branch_pics} and other figures in this paper illustrate the edges between the roots or beacons (as their distances are known). Suppose $t_1 = t$; then $t_2$ should be { on} the other route from $v$ to the roots. Otherwise, the removal of $t_1$ and $t_2$ does not divide $B(v)$.  As such, if the removal of $t_1$ and $t_2$ { causes} $B(v)$ to be divided { into} two sub-graphs, the two disjoint routes from $v$ to $r_{1}, r_{2}$ must have been cut off. Therefore, $v$ is in a different graph from that of the two roots.} 

(II) $t_{1}$(or $t_{2}$) is a root. For simplicity, let $t_{1}=r_{1}$.
Then $t_2 \notin \{r_{1},r_{2},v\}$, { in which case} $t$ has two disjoint routes to $r_{2}$ and $v$. These two routes do not contain $r_{1}$, as $r_1$ is a root. Hence, the removal of $t_2$ will not break these two disjoint routes to $r_{2}$ and $v$. In summary, there is no node $t_{2}$ that together with $r_{1}$ { will} divide the branch while keeping $r_{2}$ and $v$ in the same component.

(III) $\{t_{1},t_{2}\} = \{r_{1},r_{2}\}$; i.e., $t_{1}$ and $t_{2}$ are the two roots. In this scenario, $v$ has a path to all its ancestors left in $B(v)$. $B(v)$ can not be divided by removing $\{r_{1},r_{2}\}$. Therefore, $\{t_{1},t_{2}\}$ contains at most one of the roots.

From the above, node $v$ and the remaining root(s) are in different sub-graphs if $B(v)$ is broken { into} two sub-graphs after the removal of two nodes.
\end{proof}
\vspace{-0.5cm}
\begin{figure}[ht]
\begin{minipage}{0.22\textwidth}
\centering
\includegraphics [height=2.5cm]{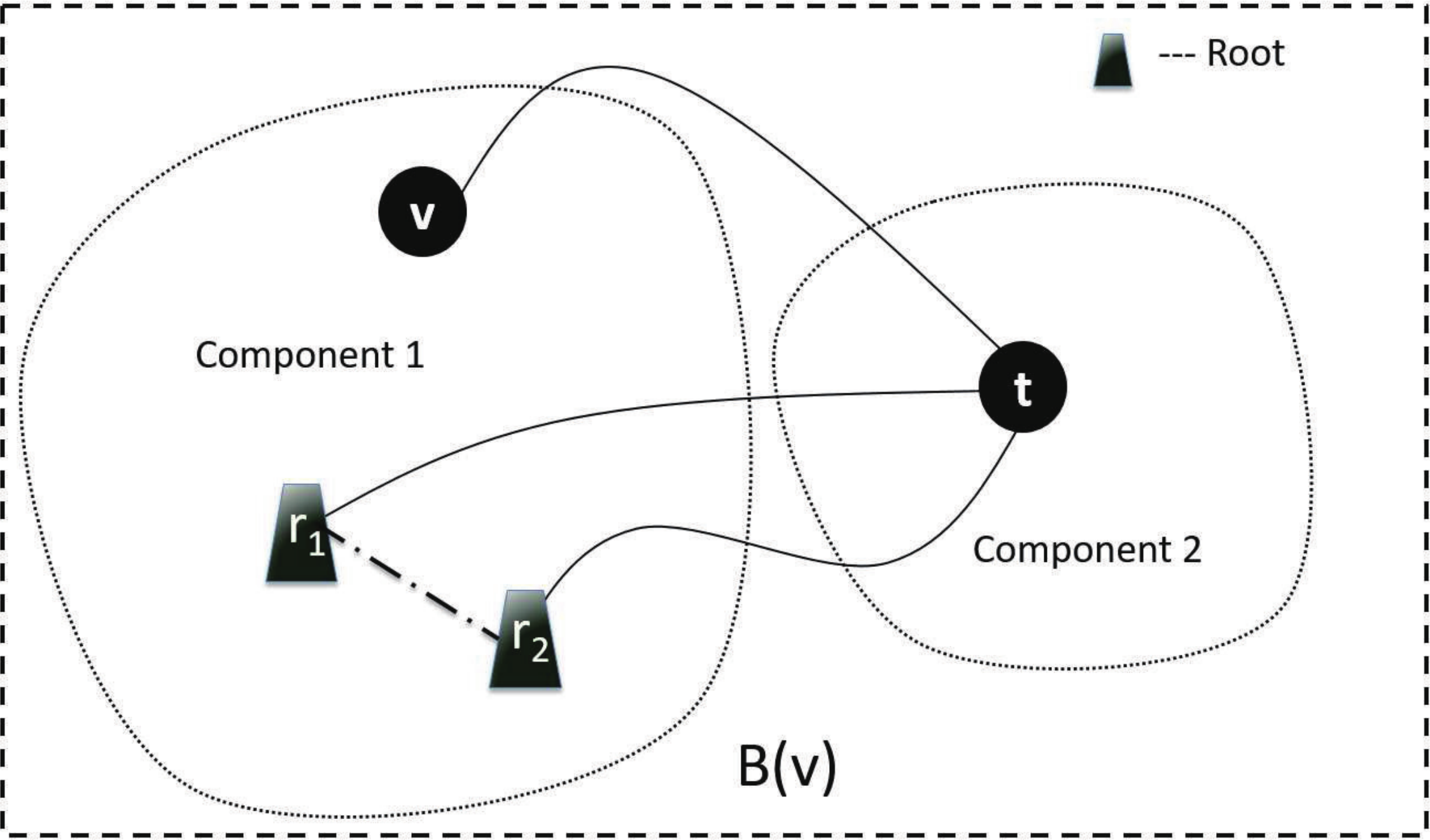}

\footnotesize\text{(a) Branch division}

\end{minipage}
\begin{minipage}{0.01\textwidth}
\end{minipage}
\begin{minipage}{0.24\textwidth}
\centering
\includegraphics [height=2.5cm]{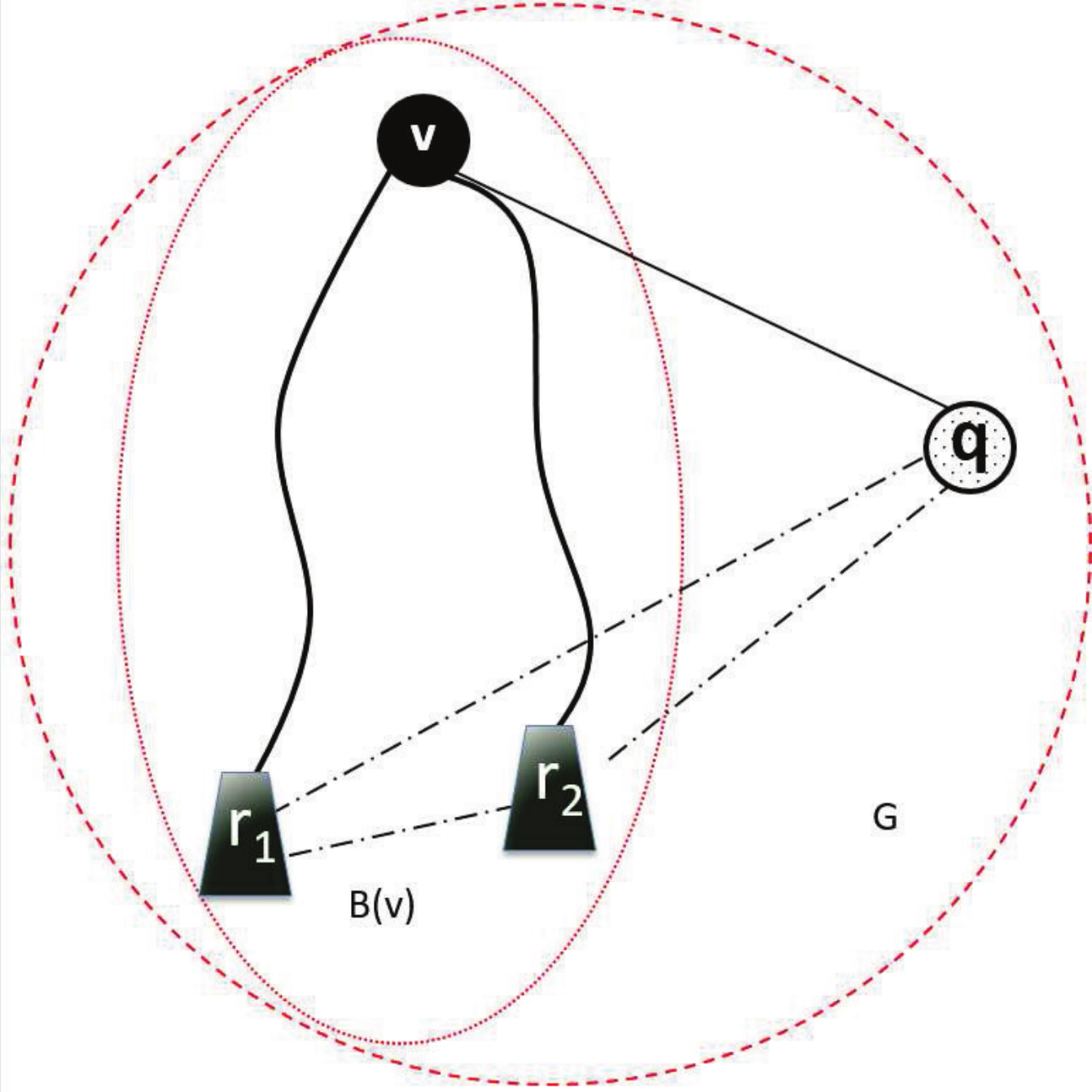}
\footnotesize\text{ (b) Globally rigid graph construction}
\end{minipage}
\caption{Branch and construction of globally rigid graph (dash-dotted lines: existing edges with fixed distances)}  \label{branch_pics}
\end{figure}

Using Lemma \ref{branchlemma2}, the graph $G$ composed of $B(v)$, $q$, $(q, v), (q, r_{1}), (q, r_{2})$ {  such as} that in Fig.\ref{branch_pics}(b) can be proved to be 3-connected:  First of all,  if $q$ in $G$ is removed, then $G$ becomes $B(v)$, which is still a connected graph. Then, removing any two nodes of $B(v)$ divides $B(v)$ into at most two disconnected sub-graphs, as a branch is 2-connected. From { L}emma \ref{branchlemma2}, after the removal of two more nodes,  if $B(v)$ is divided into two graphs, then $v$ and the roots must be in different graphs. Therefore, $G$ is 3-connected. If $G$ is further proved to be an M-circuit, then $G$ can be claimed as a globally rigid graph.  Theorem \ref{wldtheorem}  shows how to obtain an M-circuit.

\begin{theorem}
\label{wldtheorem}
Given a branch $B(v)$, where $B(v)=(V_{B},E_{B})$, if a graph $G$ is constructed by adding a new vertex $q$ into $V_{B}$ and three edges $(q, v), (q, r_{1}), (q, r_{2})$ into $E_{B}$, then $G$ is globally rigid.
\end{theorem}
\vspace{-0.1cm}
\begin{proof}

Consider a branch $B(v)$ and a graph $G$ { obtained} by adding node $q$, and { edges} $(q, v), (q, r_{1}), (q, r_{2}) $ to  $B(v)$.  Using Lemma \ref{branchlemma2}, as mentioned above, $G$ can be proved to be 3-connected.  This 3-connected property leads to the following equations:  $|E_{B}|=2|V_{B}|-3$, $|E_{G}|=|E_{B}|+3$ and $|V_{G}|= |V_{B}|+1$; and finally $|E_{G}|=2|V_{G}|-2$. 

Given a node set $X$, $X \subseteq V_{B}$, if node $q \notin X$,  then $|E[X]|\leq 2|X|-3 $ since $B(v)$ is minimally rigid and $X$ is a subset of { the} nodes in $B(v)$.  If $q \in X$, we prove by contradiction. We assume $|E[X]| \geq 2|X|-2$. Let $X' = X - q$. As node $q$ has at most three edges connected to the nodes in $X$, we have $|E[X']| \geq 2|X'|-3$ and $X' \subsetneqq V_{B}$ after the removal. Because $X' \subsetneqq V_{B}$ and $B(v)$ is minimally rigid, $|E[X']|>2|X'|-3$ does not hold. From Lemma \ref{mylemma1}, $|E[X']|=2|X'|-3$ is not possible, either. As a result,  the assumption $|E[X]| \geq 2|X|-2$ does not hold and thus only $|E[X]| \leq 2|X|-3$ holds. 
Now that for each $X\subsetneqq V_{2}$ with $|X|\geq2$, $|E[X]|\leq 2|X|-3$ holds, { the 3-connected graph $G$ can be concluded to be} an M-circuit. Therefore, $G$ is globally rigid.
\end{proof}
\vspace{-0.5cm}


\section{Distributed Localizability Detection Using Triangle Extension}    

  {In this section, we propose a distributed approach to localizability detection through  triangle extensions. Fig.\ref{grgc} provides an example to illustrate this idea. In Fig.\ref{grgc}, there are two nodes in the initial stage, $r_1$ and $r_2$, which constitute a $K_{2}$ graph. In the first step, by triangle extension, $v_1$ extends $r_1$ and $r_2$ to form {  the} branch $B(v_1)$. Then,  $v_2$ is added as the { leaf node} to form the branch $B(v_2)$. Finally, the node $q$, which has edges to $r_1$ and $r_2$, is found to be a neighbor of $v_2$. According to Theorem \ref{wldtheorem}, the graph $G$, which contains both $q$ and $B(v_2)$, is globally rigid. { Supposing} that the locations of nodes $q$, $r_1$ and $r_2$ are known and the three nodes are not { on} the same line, all of the nodes in Branch $B(q)$ can be determined  { to be} localizable. }

\begin{figure}[ht]
\centering
\includegraphics[width=0.4\textwidth]{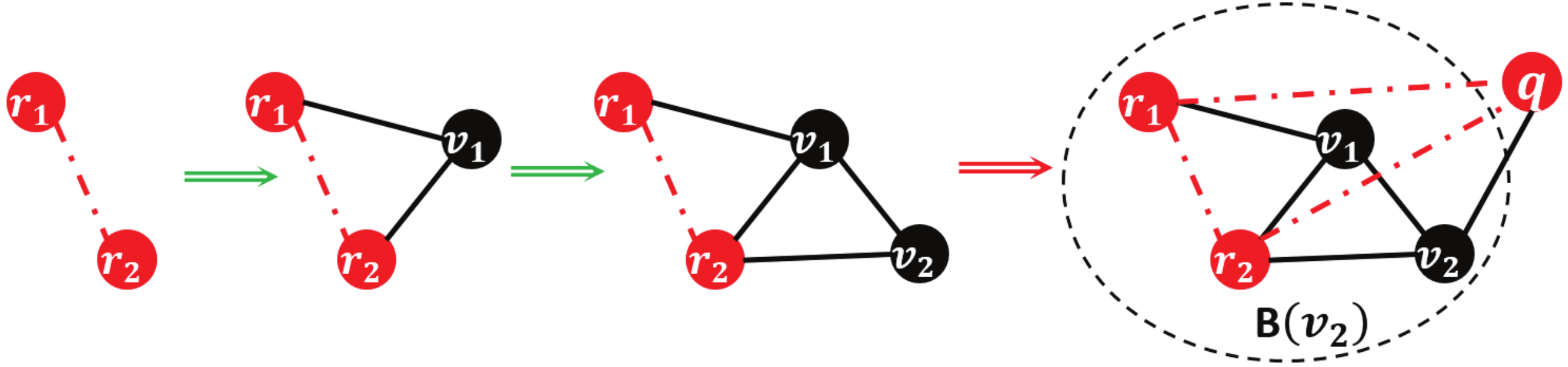}
\caption{A complete example of branch construction and global rigidity determination}
\label{grgc}
\end{figure}

  {Our approach proceeds in two phases: the extension phase and the detection phase. In the detection phase, a node determines its localizability $state$ according to the received messages. The nodes in a WSN may be { in any of the following} three states:  \emph{flexible}, \emph{rigid} and \emph{localizable}. The state of a beacon is initialized as localizable since its location is known and fixed.}

The  { details of the two phases} are as follows. In the extension phase, the states of beacons are first labeled as localizable, and those of the other nodes are labeled as flexible.  Then, a pair of beacons triggers the extension operations on their neighbor nodes. In turn, certain neighbor nodes will be added to form triangle blocks. These newly added nodes change their states to \emph{rigid} and inform their neighbors. In the detection phase, when a node $v$ changes its state from \emph{flexible} to \emph{rigid}, it checks whether the following two conditions are satisfied;  { if so, then all of the nodes in $B(v)$ can be determined to be localizable}.
\\
(1) $v$ has a localizable neighbor $q$ that is not an ancestor of $v$ in $B(v)$.\\
(2) Locations of $q$ and the roots of the branch are not collinear (not { on} the same line).\\

We now formulate the initial version of the distributed localizability detection algorithm using the triangle extension, denoted as ITE. ITE has two phases: \emph{extension} and \emph{detection} as shown in Procedure \ref{algote} and Procedure \ref{algoteD}, respectively.  In Procedure \ref{algote}, initially a node broadcasts messages that contain its unique ID, its current state and its location (if known). A node { in} the \emph{flexible} state performs extensions after measuring distances to its neighbors: (1) It chooses a pair of neighbors as its parents and updates its state to \emph{rigid} (Lines 11--20). (2) It broadcasts the state information to its neighbors (Lines 21--22).   In Procedure \ref{algote}, the sets $P$ and $B$ denote the parent set and branch set on the node.

\algsetup{indent=1.2em}
\begin{algorithm}
\caption{Extension Phase}
\label{algote}
\begin{algorithmic}[1]
\STATE type b: \{state, parents\{NULL, NULL\},roots\{NULL, NULL\} \} // b: branch tuple
\STATE type p: \{id, b\} // p: a candidate of parents for this node
\STATE $P \gets \phi$ //$P$: set of $p_i$, $i=1,2,3,...$
\STATE $B \gets \phi$ //$B$: set of $b_i$, $i=1,2,3,...$
\STATE Initialization($this.id, P, B$) //  the states of beacons are set as localizable
\IF{ReceivedMessagesFromNeighbors($p_i$)}
  	\STATE $P$.Add($p_i$)
	\FOR{$p_j$ in $P$}
  	\IF{$p_i=p_j$}
  		\CONTINUE
   	\ENDIF
   	\IF{$p_i$ is $beacon$}
    	\IF{$p_j$ is $beacon$}
     		\STATE $B$.Add(new b($state \gets rigid, parents \gets \{p_i,p_j\}, roots \gets \{p_i,p_j\}$))
    	\ELSIF{$p_i \in p_j.roots$} 
     		\STATE $B$.Add(new b($state=rigid, parents \gets \{p_i,p_j\}, roots \gets p_j.roots$))
    	\ENDIF
   	\ELSE
    	\IF{$p_j.state$ = $localizable$ and $p_j \in p_i.roots$}
     		\STATE $B$.Add(new b($state \gets rigid, parents \gets \{p_i,p_j\}, roots \gets p_i.roots$))
    	\ELSIF{$p_i.roots$ equals $p_j.roots$}
     		\STATE $B$.Add(new b($state \gets rigid, parents \gets \{p_i,p_j\}, roots \gets p_i.roots$))
    	\ENDIF
   	\ENDIF
  \ENDFOR
\ENDIF
\FOR{$b_i$ in $B$}
 	\STATE Broadcast(new p($this.id$, $b_i$))
 \ENDFOR
 \end{algorithmic}
\end{algorithm}

  {The nodes in Fig.\ref{grgc} are taken as an example to show how the triangle extension and detection are launched. Initially, every node runs Procedure \ref{algote}. Node $v_1$ finds that its two parents $r_1$ and $r_2$ are two beacons when it receives the messages from these two parents (Line 6). It then adds them to form its branch $B(v_1)$,  as shown in Lines 11--13. This step starts a \emph{triangle extension}. Next, $v_1$ updates its state to rigid and integrates this new state and branch information into a parent candidate, $p$.  Finally, $v_1$ broadcasts $p$ to its neighbors (Line 22).  $v_2$ receives $p$ and adds $v_1$ and $r_2$ to construct a new branch $B(v_2)$ (Lines 14--15). The triangle extension comes to $v_2$. Procedure \ref{algoteD} on a node detects whether there is an extra localizable neighbor that is non-collinear with the two roots of the branch of this node. If so, the node marks itself as localizable and { broadcasts} this news to its neighbors.}

\begin{algorithm}
\caption{Detection Phase}
\label{algoteD}
\begin{algorithmic}[1]
\IF{ReceivedMessagesFromNeighbors($p_n$)}
 	\IF{$p_n$ is a beacon}
  	\FOR{$b_i$ in $B$}
   		\IF{$p_n \notin b_i.roots$ and non-collinear($b_i.roots$, $p_n$)}
    		\STATE $b_i.state \gets localizable$
   		\ENDIF
  	\ENDFOR
 	\ELSIF{$p_n.state=localizable$ and ($this.id = p_n.parents[0].id$ or $this.id = p_n.parents[1].id$)}
  	\FOR{$b_i$ in $B$}
   		\IF{$b_i.roots$ equals $p_n.roots$}
    		\STATE $b_i.state \gets localizable$
   		\ENDIF
  	\ENDFOR
 	\ENDIF
\ENDIF
\FOR{$b_i$ in $B$}
	\IF{$b_i.state=localizable$}
	\STATE Broadcast((new p($this.id, b_i$))\ENDIF
\ENDFOR
\end{algorithmic}
\end{algorithm}

  {We use the previous figure, Fig.\ref{grgc}, to show how Procedure \ref{algoteD} works. As { given} in Procedure \ref{algote}, node $q$ finally receives the branch information from its parents. Since node $q$ is a beacon, it { notifies its localizability state back} to its neighbors. After $v_2$ receives this information from $q$, it proceeds Lines 2-5, and notifies all its neighbors that the nodes on the branch of $v_2$ are all localizable. }

In Algorithm ITE, a node finds two neighbor nodes and extends them if these two neighbors are in the same triangle block as itself. The node can find out whether it shares the same triangle block with the two neighbors, by comparing the roots of the neighbors with itself. The two neighbors are then marked as the possible parents of this node. Each node maintains a set of possible parents and a set of branches. The extension phase on each node covers all pairs of neighbors in the branch set. In the detection phase, a traversal of set $B$ is performed to inform all neighbors. Therefore, the time complexity of ITE is O($m^2$), where $m$ is the number of neighbors of a node has.

We also analyze the space complexity of the two sets $P$ and $B$ due to the space limit of a sensor node. The number of neighbors of a node in a sparse network is usually relatively small, { on} the order of tens. Therefore, a single resource-modest sensor node can maintain the two sets. Take node $d$ in the WSN of Fig.\ref{spacecomlexity}(a) as an example. Node $d$ only needs to put four nodes, $r_{2}$, $a$, $b$, and $c$, into its set $P$. In Fig.\ref{spacecomlexity}(a), the ancestors' topologies are transparent to node $d$, and thus there are only six different branches in its neighborhood. In general, within a branch any pair of neighbors of a node can be the node's parents and any pair of beacons can be the roots. Hence, there { are} at most $\binom{m}{2}*\binom{k}{2}$ different branches for a node, { supposing} there are $m$ neighbors of a node and $k$ beacons on average ($m \ge 2$, $k \ge 2$). Node $d$ can construct six different branches and { these} are all listed in Fig.\ref{spacecomlexity}(b) .

\begin{figure}[htbp]
\centering
\begin{minipage}{0.2\textwidth}
\includegraphics[width=2.5cm]{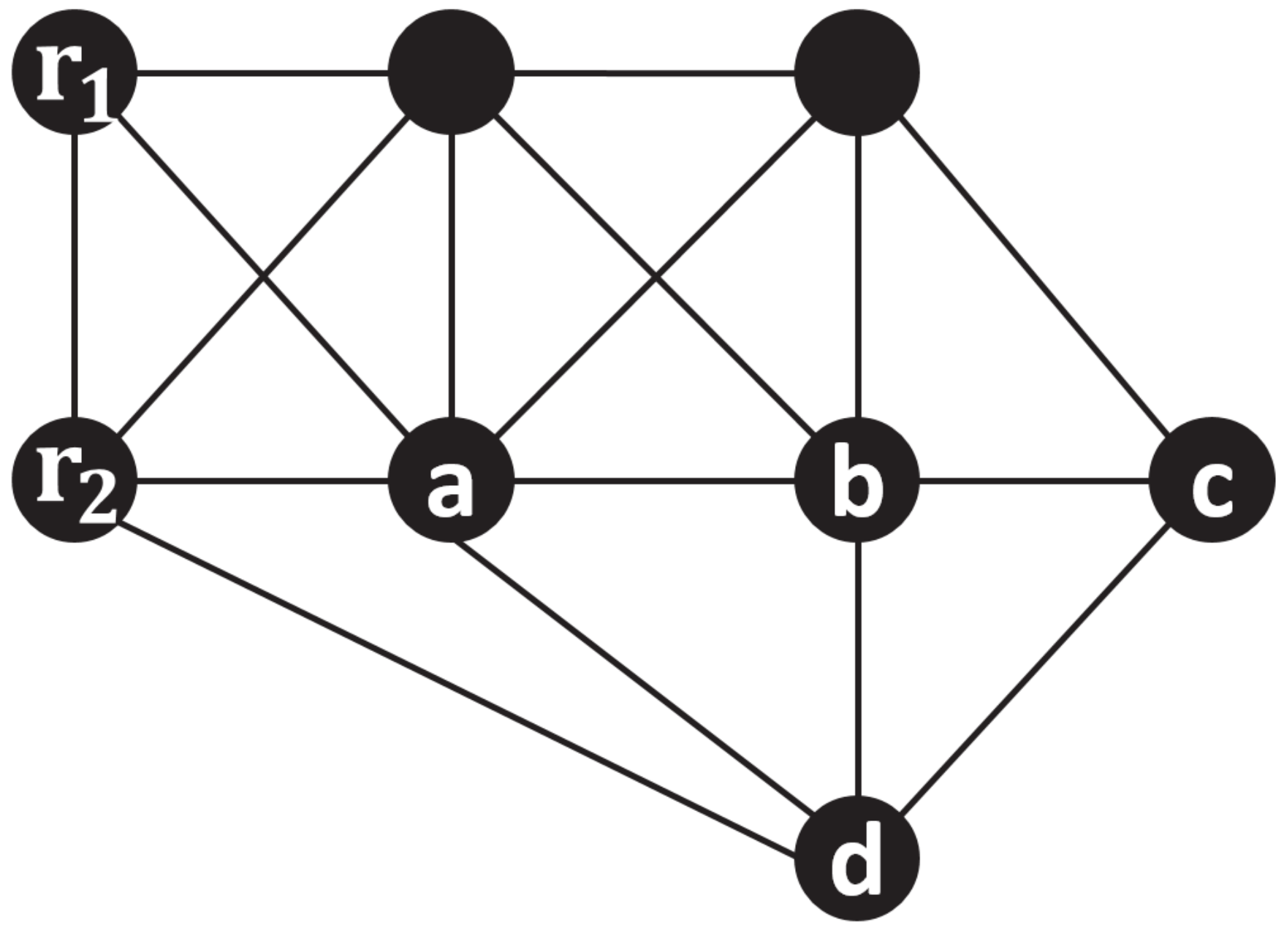}
\footnotesize\text{(a) { Working} scenario for TP and TE}
\end{minipage}

\begin{minipage}{0.45\textwidth}
\centering
\includegraphics[width=6cm, height = 4cm]{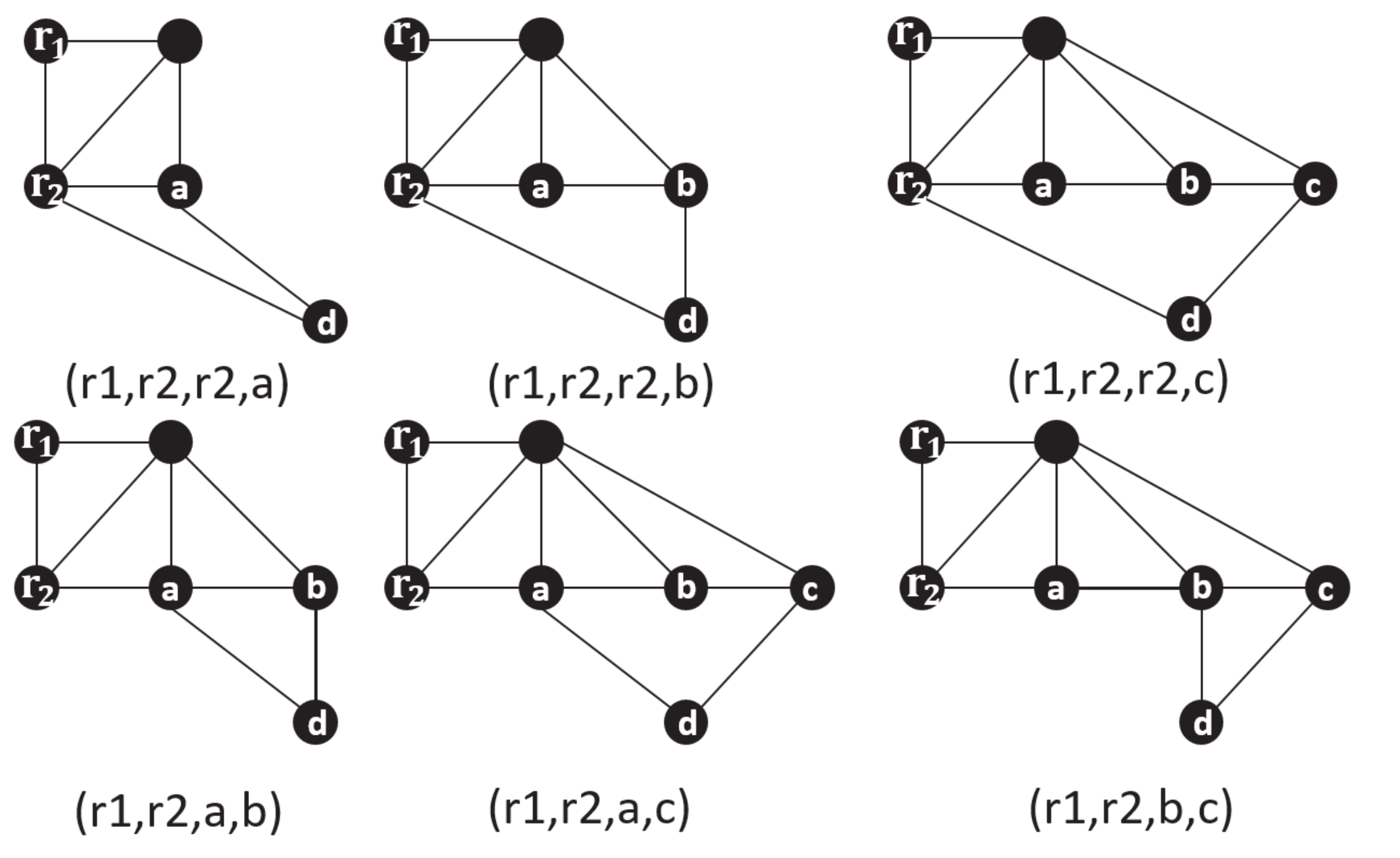}

\footnotesize{(b)  The six branches on node $d$}
\end{minipage}
\caption{Space complexity analysis example}
\label{spacecomlexity}
\end{figure}

As most sensor nodes are highly resource-limited, it is inefficient for a node to maintain a set of all its branches in practice. Hence, ITE is not applicable to large-scale networks. A naive method is to limit the size of the branch set, but it is difficult to locally choose the best branch for a node itself. The best branch can help the node to detect the most localizable nodes.

To address this resource { limitation} problem, we next propose an advanced extension operation, called \emph{directed}-\emph{extension}, as follows. The extension operation in ITE is not directed.  For instance, the undirected extension of ITE can add a new node $v$ in $G$ and two edges ($v$,$v_{1}$), ($v$,$v_{2}$) to form $B(v)$, as long as $v_{1}, v_{2}\in V_{G}$. In directed-extension, $v$'s parents $v_1$ and $v_{2}$ should conform to the following additional rule: $v_{1}$ should be a parent of $v_{2} $ or $v_{2}$ a parent of $v_{1}$ if they are not the two roots. Fig.\ref{newExtension} shows an example: when $b$ is to be added to form $B(b)$, its parent $a$ and another parent $r_2$ conform to the rule, as $r_2$ is a parent of $a$. Furthermore, in directed-extension, once a node changes its state to rigid, it will { no longer} accept any other nodes as its parents. This way, the extension will follow { only one} direction. 

Directed-extension { holds} the property of minimally rigidity, since the new branch set of node $v$ { is a subset of the branch set in the original (undirected version) extension}. Because the directed-extension approach limits the extension possibilities, the set $B$ can be reduced significantly. In Fig.\ref{newExtension}, a node sequence such as ($a, b, c, d, e, f$) is called a \emph{directed triangle extension path} and this extension path specifies a unique branch.

 \begin{figure}[ht]
\centering
\includegraphics [width=0.49\textwidth]{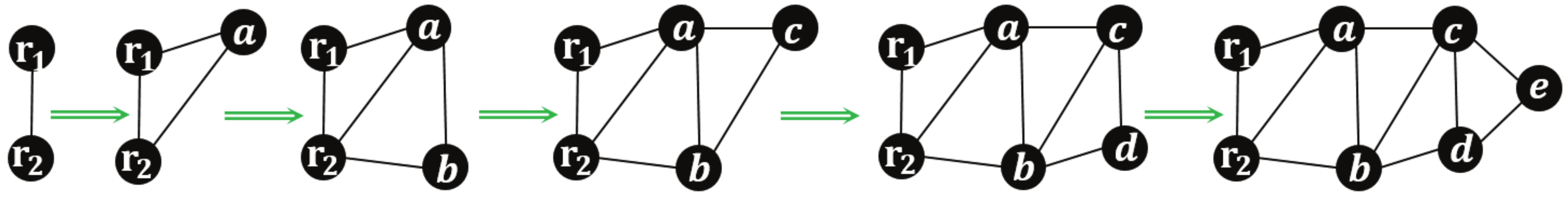}
\caption{Directed-extension} \label{newExtension}
\end{figure}
\begin{figure}[ht]
\centering
\includegraphics [width=0.3\textwidth]{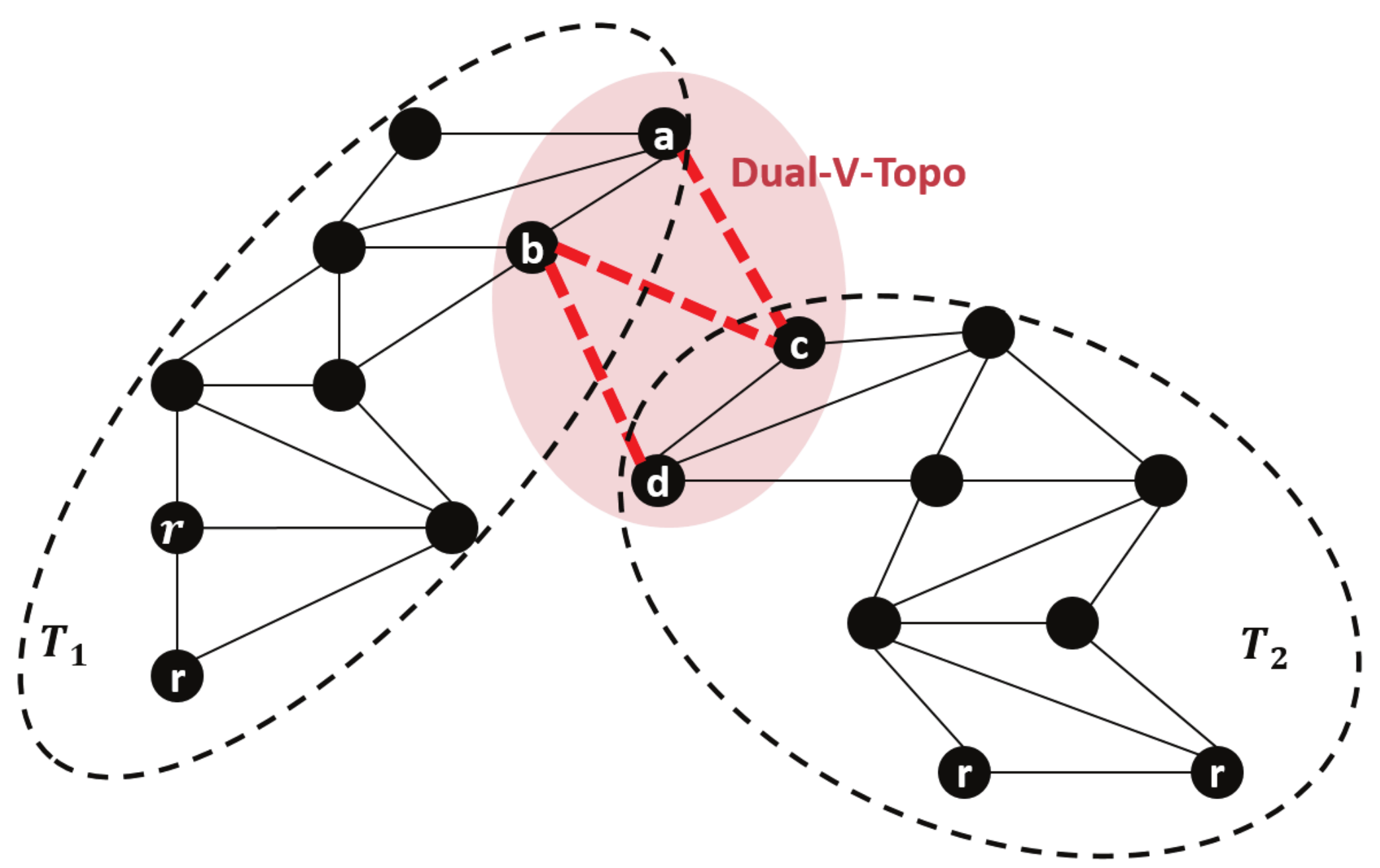}
\caption{Triangle blocks meeting} \label{Tmeets}
\end{figure}

Nonetheless, directed-extension has \emph{an early-stop problem}. Suppose that in a WSN as shown in Fig.\ref{Tmeets} there are two pairs of beacons ({ each root} is denoted as $r$ for simplicity). Each pair of beacons can create a triangle block, denoted as $T_{1}$ and $T_{2}$ respectively. After a series of extensions, $T_{1}$ might overlap $T_{2}$ as shown in the shaded area of Fig.\ref{Tmeets}. The special topology in the overlapping area of $T_{1}$ and $T_{2}$ is named as \emph{Dual-V-Topo}. The Dual-V-Topo is composed of four nodes and five edges, in which two of the nodes have three edges { each}. As the directed-extension on a node stops when the node changes its state to rigid, the nodes $a$, $b$, $c$ and $d$ { perform no further} extension operations. As a result, since there are no more than three non-collinear beacons in each block, neither graph $T_{1}$ nor graph $T_{2}$ can be determined { to be} localizable. However, the two blocks  can actually be determined { to be} localizable by the original undirected triangle extension approach, ITE.

To address this early-stop problem, we modify the detection phase to \emph{dual-v-detection}. A node in a branch may launch dual-v-detection when it finds rigid neighbors { in a triangle block other than its own}. For example, in the network of Fig.\ref{Tmeets}, the extension from $b$ to $a$ and the extension from $d$ to $c$ finish at the same time. Next, node $c$ learns that it can access rigid neighbors $a$ and $b$, both of which belong to a different triangle block. Node $c$ will launch the dual-v-detection procedure to test { whether} node $d$ can access $a$ or $b$. $T_1$ and $T_2$ can be connected by the three edges  $(a, c)$, $(b, c)$, and $(b, d)$ when node $c$ confirms that node $d$ can access $a$ or $b$. The two connected triangle blocks can be determined localizable. This procedure { for} determining the localizability of $T_1$ and $T_2$ is called \emph{dual-v-detection}. Combining the directed extension (Procedure \ref{algote2}), and the dual-v-detection { procedure} (Procedure \ref{algoteD2}), we propose the final version of { the} localizability detection algorithm via triangle extension, denoted as \emph{TE}.

\begin{algorithm}
\caption{TE-Extension Phase}
\label{algote2}
\begin{algorithmic}[1]
\STATE type $b$: \{state, parents\{NULL, NULL\},roots\{NULL, NULL\} \} // b: branch tuple
\STATE type $p$: \{$id$, $b$\} // p: a candidate of parents for this node
\STATE $P \gets \phi $ // $P$: set of $p_i$
\STATE Initialization($myp$) //$myp$, $p_i$ and $p_j$ represent instances of type p
\IF{ReceivedMessagesFromNeighbors($p_i$)}
  \STATE $P$.Add($p_i$) 
  \FOR{$p_j$ in $P$}
   \IF{$p_i.id=p_j.id$}
    \CONTINUE
   \ENDIF
   \IF{$p_i.state$ = $beacon$}
    \IF{$p_j.state$ = $beacon$}
     \STATE $myp.state \gets rigid$; $myp.parents \gets \{p_i,p_j\}$; $myp.roots \gets \{p_i,p_j\}$
     \BREAK
    \ELSIF{$p_i \in p_j.parents$}
     \STATE $myp.state \gets rigid$; $myp.parents \gets \{p_i,p_j\}$; $myp.roots \gets p_j.roots$
    \ENDIF
   \ELSE
    \IF{$p_j.state$ = $localizable$ and $p_j \in p_i.parents$}
    \STATE $myp.state \gets rigid$; $myp.parents \gets \{p_i,p_j\}$; $myp.roots \gets p_i.roots$
     \BREAK
    \ELSIF{$p_i \in p_j.parents$ or $p_j \in p_i.parents$}
     \STATE $myp.state \gets rigid$; $myp.parents \gets \{p_i,p_j\}$; $myp.roots \gets p_i.roots$
     \BREAK
    \ENDIF
   \ENDIF
  \ENDFOR
  \IF{$myp.state \ne flexible$}
 	\STATE Broadcast($myp$)
 	\BREAK
  \ENDIF
\ENDIF
\end{algorithmic}
\end{algorithm}

\begin{algorithm}
\caption{TE-Detection Phase}
\label{algoteD2}
\begin{algorithmic}[1]
 \IF{ReceivedMessagesFromNeighbors($p_n$)}
 	\STATE $P$.Add($p_n$) //$p_n,p_i$: type p
 	\IF{$p_n.state=beacon$}
  	\IF{$p_n \notin myp.roots$ and non-collinear($myp.roots$, $n$)}
   		\STATE $myp.state \gets localizable$
  	\ENDIF
 	\ELSIF{$p_n.state=localizable$ and $myp \in p_n.parents$}
   	\STATE $myp.state \gets localizable$
 	\ELSIF{$p_n.roots \ne myp.roots$ and non-collinear($n.roots$, $myp.roots$)}
  	\STATE //Dual-V-Topo detection
  	\FOR{$p_i$ in $P$}
  		\STATE //find the child and parent in $P$, check parents' neighbors
   		\IF{$p_n \in p_i.parents$ or $p_i \in p_n.parents$}
    			\IF{$p_n \in parents.P$ or $p_i \in parents.P$}
     			\STATE $myp.state \gets localizable$
     			\BREAK
    		\ENDIF
   		\ENDIF
  	\ENDFOR
 	\ENDIF
  \ENDIF
  \IF{$myp.state = localizable$}
  		\STATE Broadcast($myp$)
  \ENDIF
\end{algorithmic}
\end{algorithm}

In TE, a node { in a WSN} broadcasts at most twice, once for the state transition from the flexible state to the rigid state and the other time for the notification of the success of localizability detection. In a WSN, the first traversal launched by two beacons creates several minimally rigid triangle blocks. The second traversal from the third beacon to the roots of a certain branch $B(v)$ in these blocks { informs} all the nodes in $B(v)$ that they are localizable. Hence, the time complexity of TE is $O(n)$, where $n$ is the number of nodes in the WSN. In TE, each node keeps the information about the neighboring beacon { and} rigid nodes. Such information takes $O(m)$ space, where $m$ is the number of neighbors. To find the pair of connected neighbors, TE searches every pair of neighbors on each node. Hence, TE has a time complexity of $O(m^2)$, which is acceptable since $m$ is usually small, especially in a sparse network. 

  {The following discusses the advantages of TE over the other two previous approaches, TP and WE.  The network scenarios that TP, WE, and TE work with are shown in Fig.\ref{comparison-analysis}(a),(b). { TE additionally} works in the network scenario of Fig.\ref{comparison-analysis}(c), where TP and WE cannot. In Fig.\ref{comparison-analysis}, $r_1$, $r_2$, and $q$ are beacons and their distances are known before the algorithms run. Their distances are denoted using the { dash-dotted} lines. As shown in Fig.\ref{comparison-analysis}(a), TP requires the three beacons to be the neighbors of a single node so that TP can measure the distances between node $v$ and the beacons. TE works in this scenario by a triangle extension. As shown in Fig.\ref{comparison-analysis}(b), WE requires three beacons within a $wheel$. The dashed arrows in Fig.\ref{comparison-analysis}(b) show the sequence of extensions by TE. As the extensions from $r_1$ and $r_2$ finally come to  $q$, the nodes can be determined to be localizable by TE. }

  {Fig.\ref{comparison-analysis}(c) presents a common network { in which} many nodes such as $v_3$--$v_6$ are not  neighbors of beacons. Nonetheless, they can extend from beacons by a sequence of extensions, as shown by the numbers 1 to 7 in Fig.\ref{comparison-analysis}(c).  In this figure, when $v_7$ performs the last extension(No. 7), $v_7$ determines that $B(v_7)$ is localizable, since its neighbor $q$ is a beacon and the three beacons, $q$, $r_1$, and $r_2$ { constitute} a triangle. $v_7$ then broadcasts back to its parents, $v_5$ and $v_6$, that the branch $B(v_7)$ is localizable. Similarly,  $v_4$, $v_3$, and $v_1$ are notified by their children that the branch is localizable. { The localizability of node $v_2$ cannot be determined}, however, because it is not in $B(v_7)$ but in $B(v_2)$. }

\begin{figure*}[!b!t]
\centering
\begin{minipage}{0.19\textwidth}
\centering
\vspace{0.6cm}
\includegraphics[width=2.5cm]{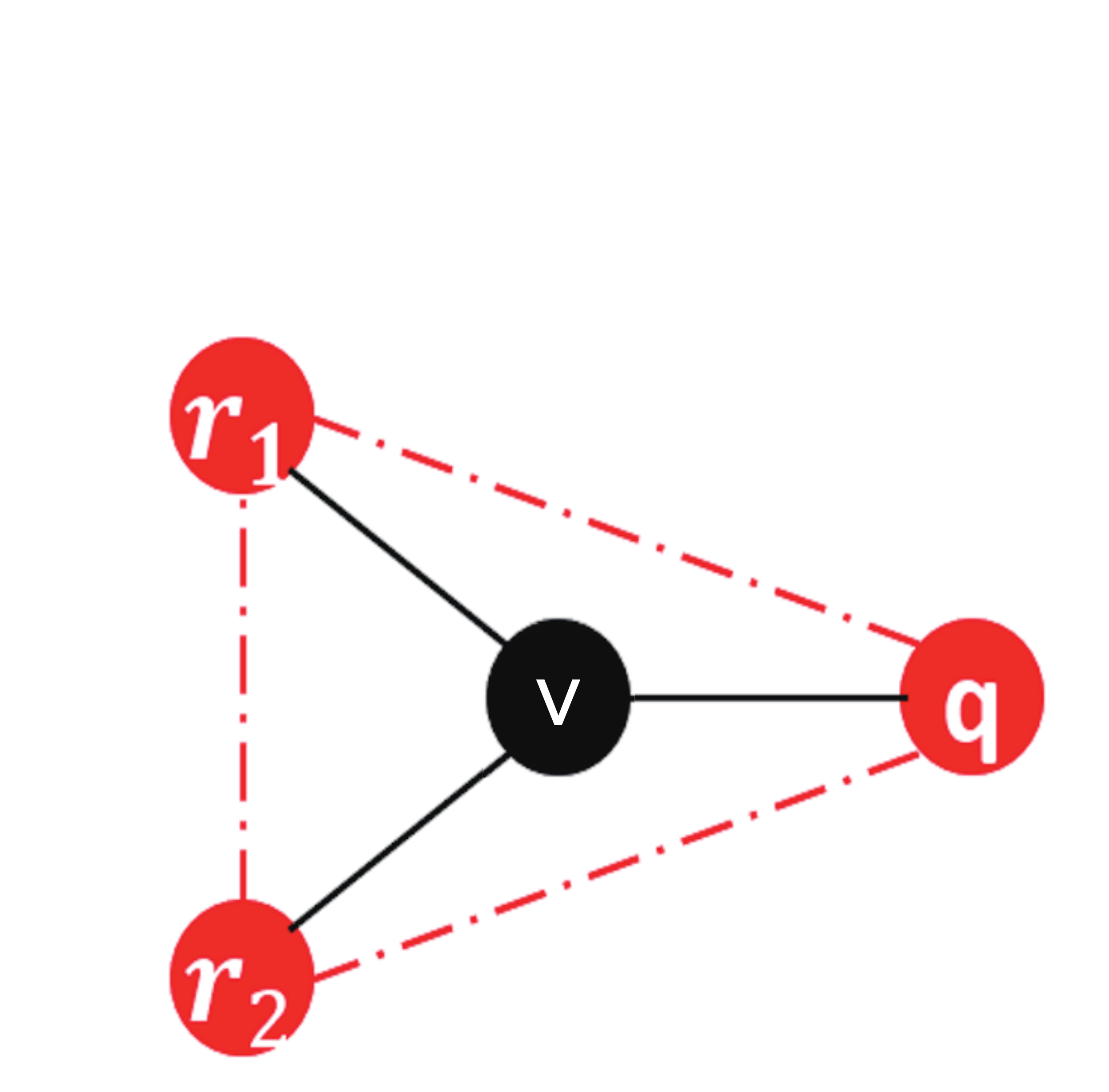}
\vspace{0.8cm}

\footnotesize\text{(a) Scenario of TP and TE}
\end{minipage}
\begin{minipage}{0.07\textwidth}
\centering
\includegraphics[width=1cm ]{blank.pdf}

\end{minipage}
\begin{minipage}{0.32\textwidth}
\centering
\vspace{0.68cm}
\includegraphics[width=5cm, height = 2.5cm ]{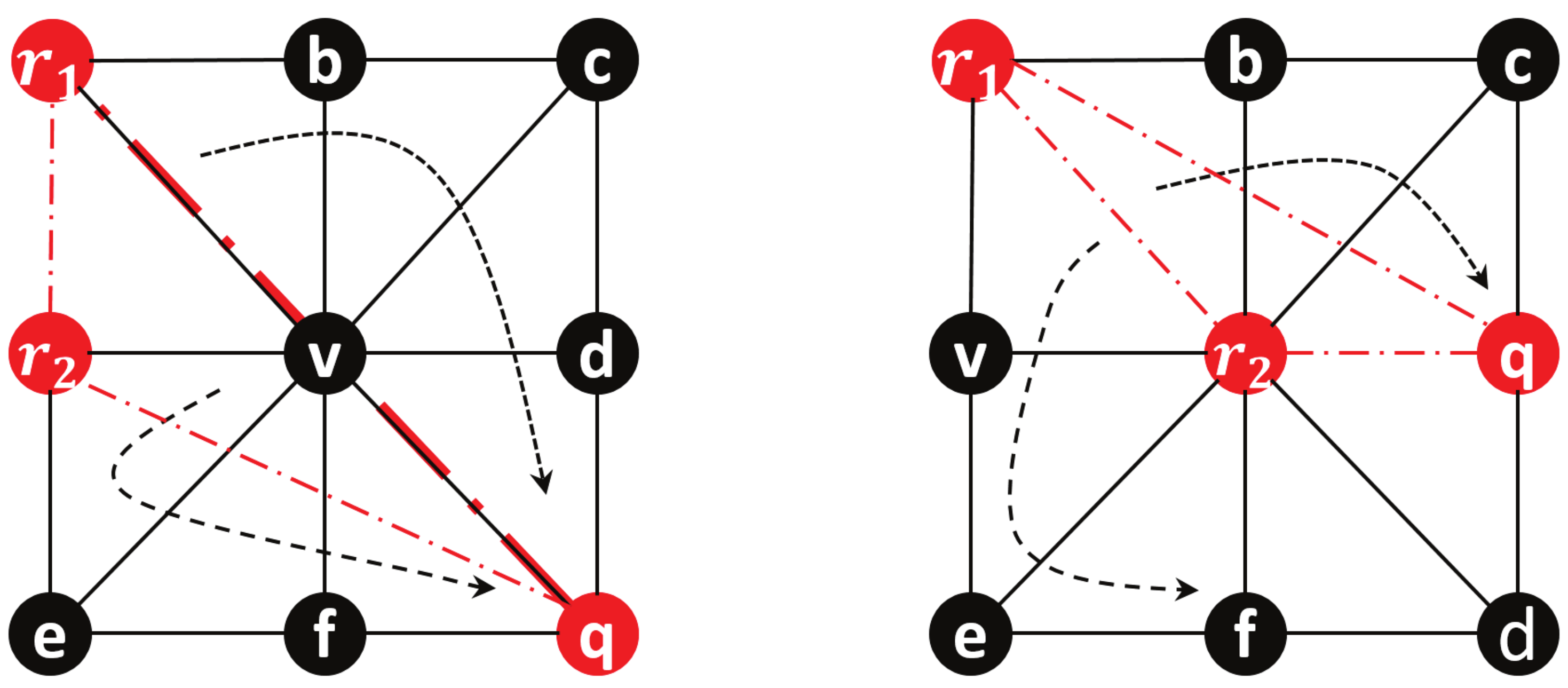}
\vspace{0.68cm}

\footnotesize{(b) Scenario of WE and TE}
\end{minipage}
\begin{minipage}{0.32\textwidth}
\centering
\includegraphics[width=3.5cm, height = 3.2cm]{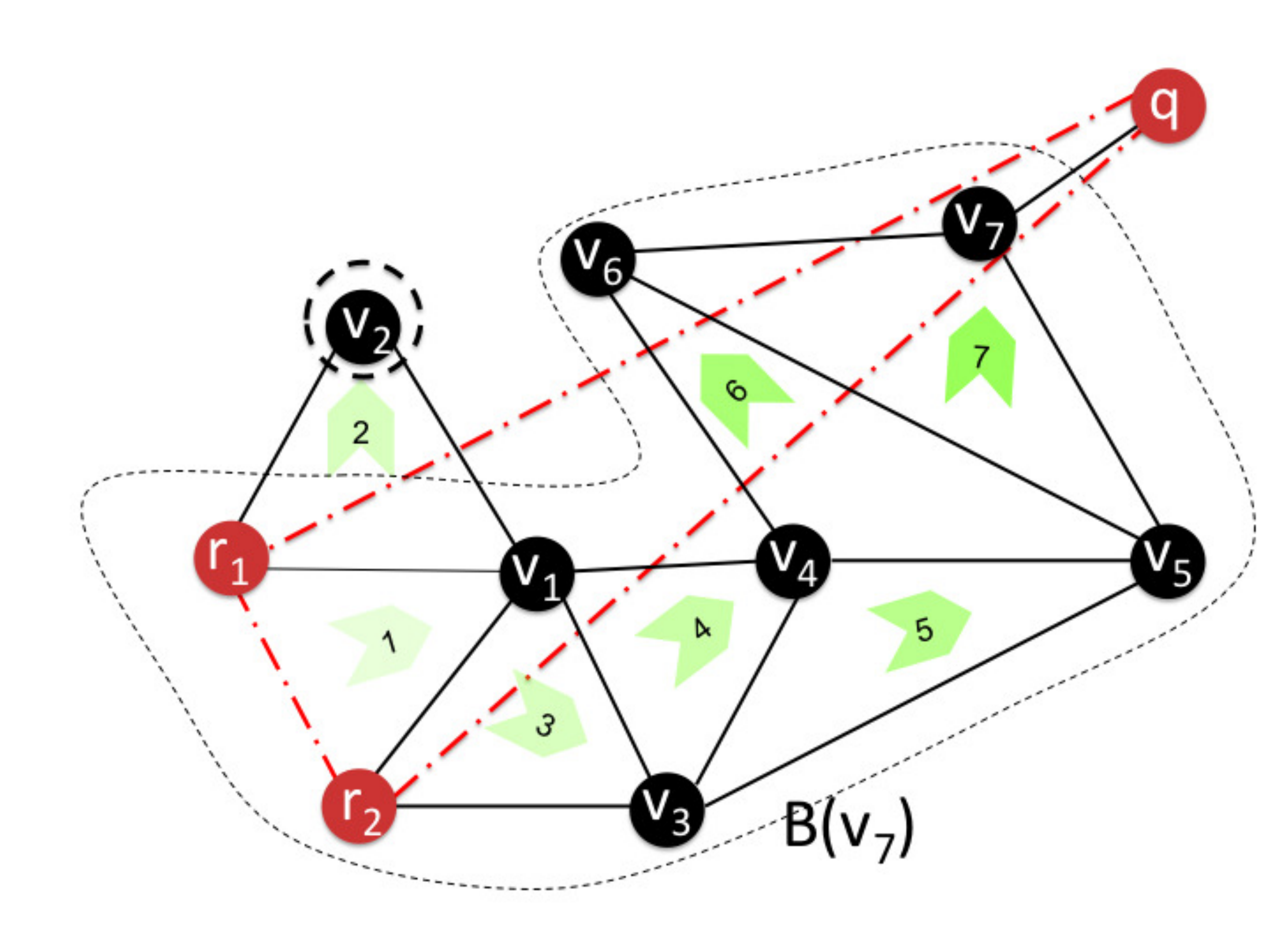}
\vspace{0.6cm}

\footnotesize{(c) Additional scenario of TE}
\end{minipage}
\caption{Comparison of TP, WE and TE}
\label{comparison-analysis}
\end{figure*}

\section{Evaluation}

\subsection{Simulation}
\label{simulations}

  {We first used TOSSIM \cite{tossim} to run the three localizability detection algorithms:  { TE (ours),  WE \cite{wheel}} and TP \cite{trilateration}. TOSSIM is a scalable WSN simulator that can simulate the behaviors of application programs in a WSN with hundreds of nodes\cite{tossim}. The application programs tested on TOSSIM can be directly executed on  { actual} sensor nodes. } 
Table \ref{topoPara}  { shows} the simulation setup, including the network and beacon deployment parameters. The deployment area was partitioned evenly into 20 $\times$ 20 cells and each node was { placed into a random cell}. The cell has a side length of $D_0 \times N$, where $D_0$ is set as the standard unit for node distance and $N$ the network density parameter. The average range for reliable communication of simulated nodes is within $6D_0$. Therefore, we set the maximum value for $N$ to 5.8, to allow the neighboring sensor nodes to communicate. A smaller $N$ value leads to a higher network density in terms of the number of nodes deployed in a unit area.  

\begin{table}[ht]
\caption{Simulation setup}
\centering
\resizebox{0.49\textwidth}{!}{
\begin{tabular}{|C{0.1\textwidth}|C{0.37\textwidth} |C{0.12\textwidth}|}
\hline
\textbf{Parameter} & \textbf{Definition} & \textbf{Value}\\
\hline
$S$ & Scale: number of nodes in a WSN  & $S = 400$\\
\hline
$C$ & Num of correctly~detected~localizable~nodes  & $0 \le C \le S$\\
\hline
$D_0$ & Unit of distance between nodes & 10 m\\
\hline
$B$ & Beacon density  & $0.01 \le B \le 0.2$\\
\hline
$N$ & Network density  & $2.0 \le N \le 5.8$\\
\hline
$L$ & Localizability detection accuracy:  $\frac{C}{S}$   & $0.0 \le L \le 1.0$ \\
\hline
\end{tabular}
}
\label{topoPara}
\end{table}


{ The simulations were run under different beacon and node densities to} reveal  how these factors affect the node localizability detection in WSNs. A beacon { continued} broadcasting its location periodically in each simulation.  The assumption is that the locations of beacons in the simulations and experiments are accurately set { and thus} will not introduce location biases.   

To { obtain} a comprehensive view of the performance of the algorithms under evaluation, we performed a series of simulations with two parameter sets with (1) $B = 0.04$ and (2) $B = 0.1$.  The network density parameters in the two sets are $N =2.4$, $N = 3.2$, and $N = 4.2$. These two parameter sets specify six WSNs, from a low beacon density and network density to a high beacon density and network density.  We ran each algorithm 30 times on TOSSIM under the two parameter sets.  Fig.\ref{composition} shows $L$, which is the average percentage of detected localizable nodes over the total number of nodes in the WSN, for the three algorithms.  The range line on the bar for an algorithm shows the upper and lower bounds of $L$ { for that} algorithm.  The results show that TE performed the best on average. Furthermore, TE is stable as the upper and lower bounds of $L$ for TE { are close together}. The bounds actually show the localizability detection probability distribution of an algorithm.

TE performed especially well when the network and beacons { were} sparse. Sparse beacons are { typical} scenarios in WSN applications, since it is often infeasible to deploy a network with up to 10\% beacons and short node distances. In addition, a high network density is not realistic, either. The parameter { value} of $N = 2.4$ might result in about 50 neighbors for an internal node. This large number of neighbors is good for localizability detection but it is { not efficient}, as there may be heavy radio interference and the node energy will be exhausted { early on}.

\begin{figure}[htbp]
\centering
\begin{minipage}{0.24\textwidth}
\centering
\includegraphics[width=4.8cm]{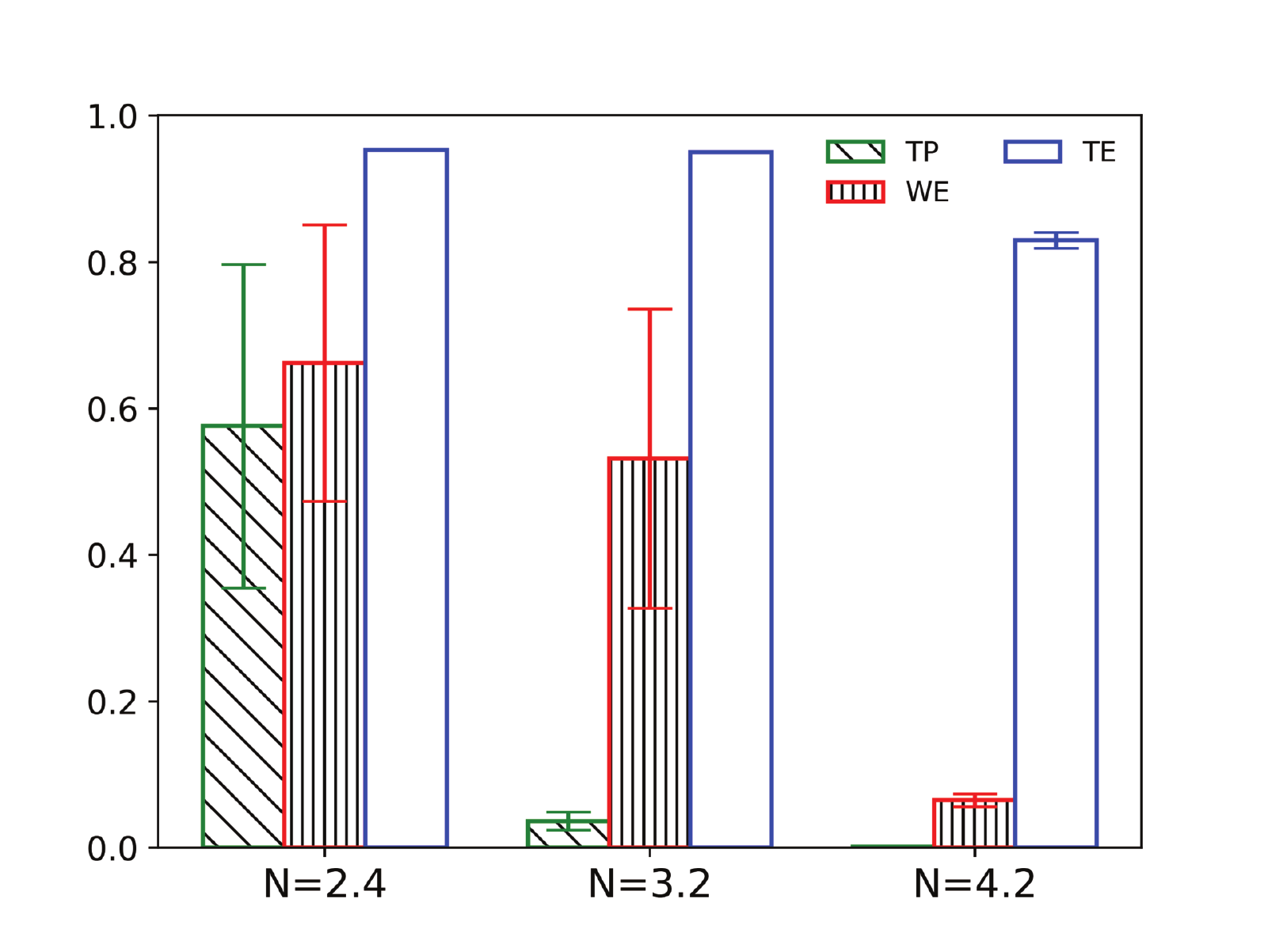}

\footnotesize{(a) B = 0.04: low beacon density}
\end{minipage}
\begin{minipage}{0.24\textwidth}
\centering
\includegraphics[width=4.8cm]{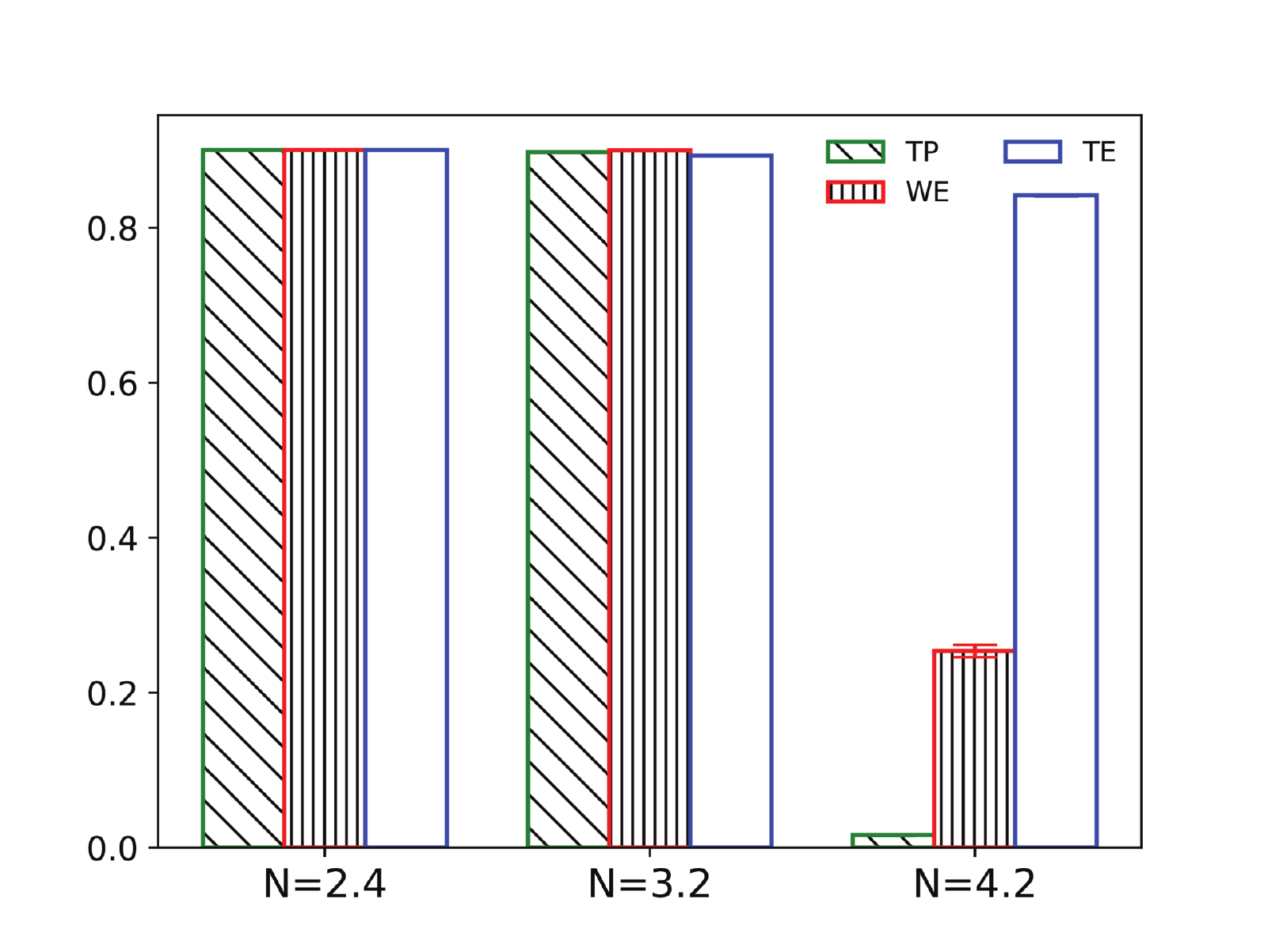}

\footnotesize{(b) B =  0.1: high beacon density}
\end{minipage}
\caption{Percentages of localizable nodes under different beacon and network densities}
\label{composition}
\end{figure}

Fig.\ref{bnl} records the values of $L$ { while $B$ and $N$ are gradually changed}. As some algorithms may not function under { very low} beacon or network densities, the results of $L$ in Fig.\ref{bnl} were recorded by running the three algorithms { until} timeout. The timeout length was set to 15 seconds, after which { none} of the three algorithms could find more than 1\% { additional} localizable nodes. The detected localizable nodes are theoretically correct. Therefore, the scheme that finds the most localizable nodes is the best one.  It can be seen that TE significantly outperforms TP and WE when $N$ is 4.2. On average, TE's $L$ is at least 50\% higher than { that for} TP and WE.  The { numeric} values of $L$ for TE, TP, and WE { are} listed in Tables \ref{teresult}, \ref{tpresult}, and \ref{weresult}) in Appendix B, respectively. The value of $L$ is calculated to an accuracy of 0.001, as our simulated network scalability is within a thousand.

\begin{figure*}[htbp]
\centering
\includegraphics[width=0.92\textwidth, height = 3cm]{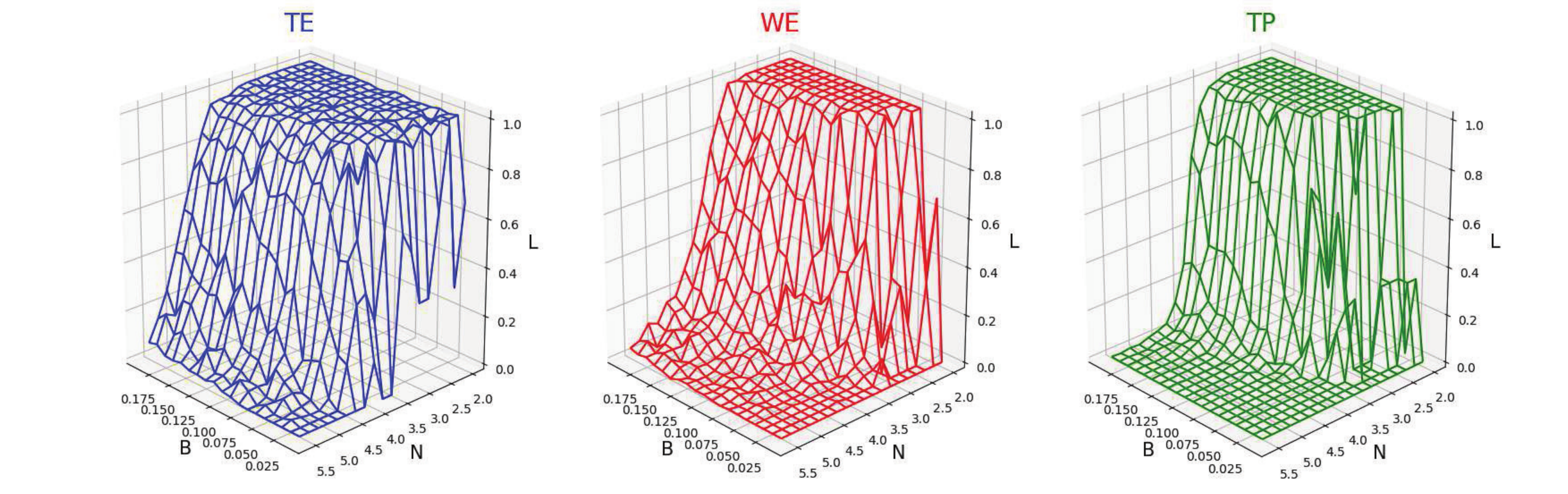}
\caption{The effect of beacon deployment and node density}\label{bnl}
\end{figure*}

We can draw the following conclusions from Fig.\ref{bnl}: (1) Node density is the dominant performance factor, especially for TP and WE. (2) A higher beacon density in a network helps finding more localizable nodes. (3) The performance of TP and WE drops sharply when $N$ grows larger than $3.5$ (Fig.\ref{bnl}). In { contrast}, the performance of TE is more stable.

During the above simulations, beacons were { placed randomly} in the WSNs.  We found that the beacon placement could affect the detection when $B \approx 0.01$, as can be seen in the first row of each of the three tables (Tables \ref{teresult}--\ref{weresult}). Fig.\ref{specialBD} shows the localizable nodes of a sparse beacon deployment network, { a kind of network for which} it is difficult to detect the localizable nodes. The network of Fig.\ref{specialBD} has only four beacons placed in a $400\times400$ square ($N=2.0$, $B=0.01$).

\begin{figure}[htbp]
\centering
\begin{minipage}{0.24\textwidth}
\centering
\includegraphics[width=4.8cm]{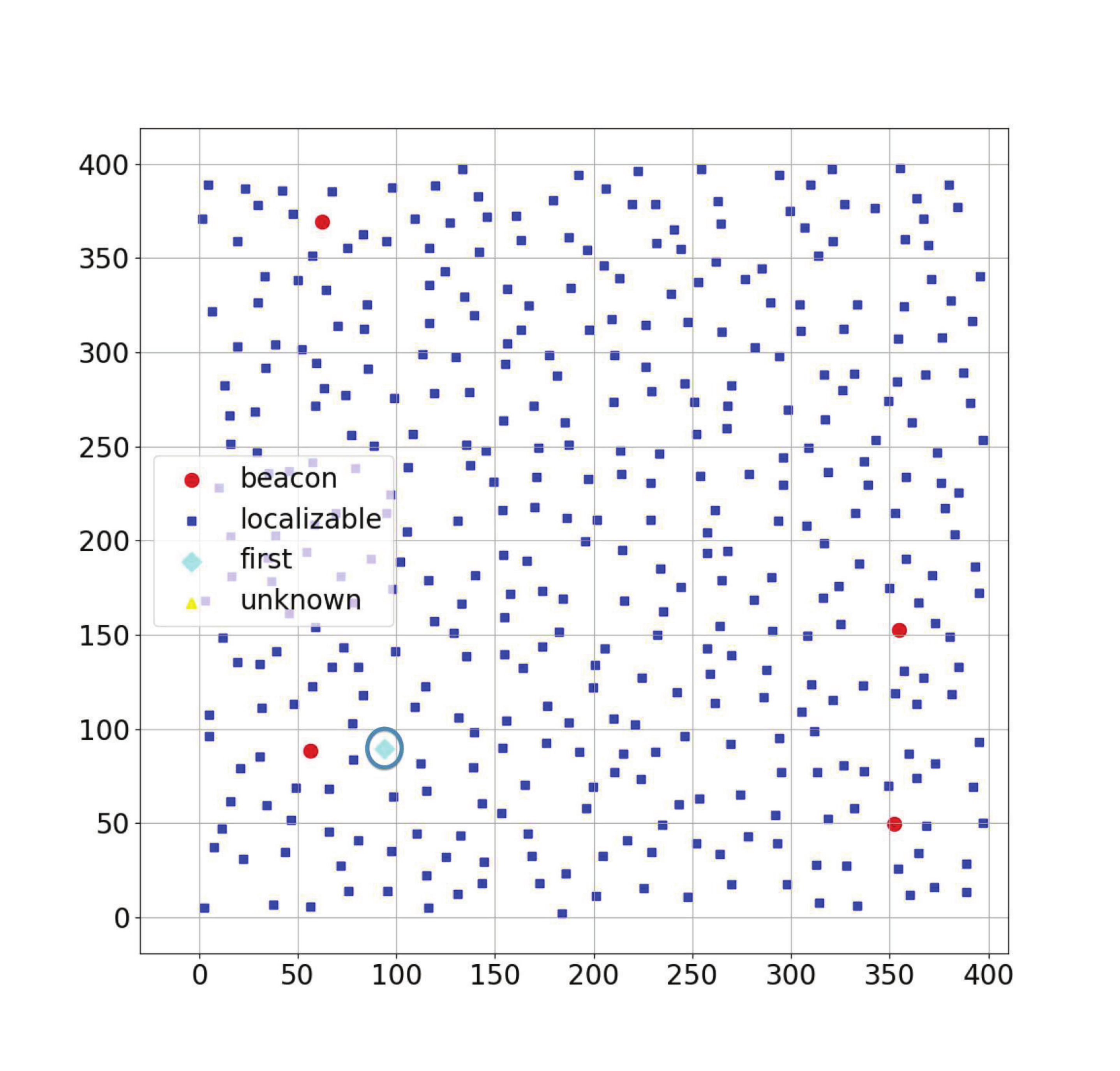}

\footnotesize{(a) [TE]}
\end{minipage}
\begin{minipage}{0.24\textwidth}
\centering
\includegraphics[width=4.8cm]{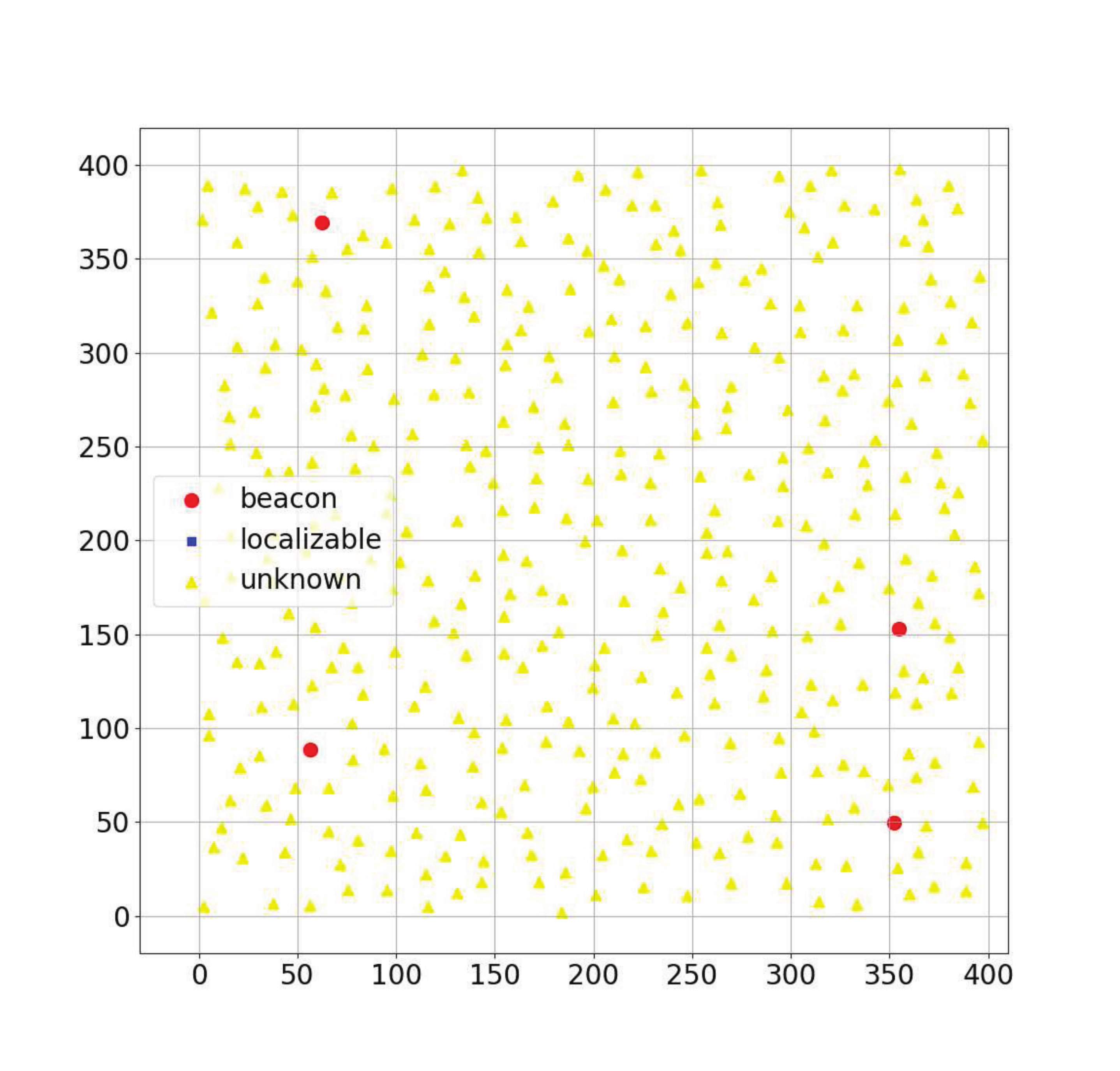}

\footnotesize{(b) [WE and TP]}
\end{minipage}
\caption{Sparse beacon placement}
\label{specialBD}
\end{figure}

WE and TP could not find any localizable nodes when the beacons  { were} placed much more sparsely, whereas TE still worked well even when only two beacons were closely placed  { near} a flexible node. The first localizable node found by TE was {close} to the third beacon, as marked in Fig.\ref{specialBD}.

The simulations demonstrate that the triangle extension (TE) method improves the localizability detection efficiency, especially when sensor nodes are sparsely deployed. Moreover, for each flexible node $v$, TE only requires two non-flexible neighbors for state transitions. In contrast, WE needs at least six edges to set up a wheel graph, and TP needs three edges connected to three localizable neighbors. As a result, { not only did} TE perform the best among the three algorithms, but it also required the least information about the network.

 \begin{table}[ht]\footnotesize
\caption{Hole simulation results}
\centering
\begin{tabular}{|c|c|c|c|c|c|}
\hline
& 1 & 2 & 3 & 4 & 5\\
\hline
TE & 0.740 & 0.855 & 0.835 & 0.789 & 0.797 \\
\hline
WE & 0.363 & 0.430 &0.499 & 0.560 & 0.545 \\
\hline
TP & 0.148 & 0.217 & 0.173 & 0.247 & 0.236 \\
\hline
\end{tabular}
\label{holeT}
\end{table}

We next evaluated the algorithms on networks with holes. A hole { of a network} is defined as an empty area within the network that has a minimum diameter greater than the transmission range of the nodes.  Hence, any two nodes on opposite borders of a hole are not neighbors.  Five simulations were performed under the same network deployment with $B=0.1$ and $N=3.2$.  To facilitate { hole generation}, there was no cell partitioning in the networks of these five simulations.

Table \ref{holeT} lists the average results of the five simulations. Fig.\ref{Holelocation} shows the nodes in one of the simulated networks classified by the algorithms. Since the other simulation results are similar, we omit them in the figure. In all simulations, TE detected more localizable nodes than the other two algorithms.  To explore the effect of network density $N$ on different algorithms, we further performed  simulations on a network with a hole under $N=2.4$ with $B = 0.1$. The three algorithms could detect all the localizable nodes in the four corners, but TE was the fastest one to finish when $N=2.4$ and $B = 0.1$.   
\begin{figure*}[!h!t!bp]
\centering
\subfigure[TE]{
\includegraphics[width=0.322\textwidth, height = 4cm]{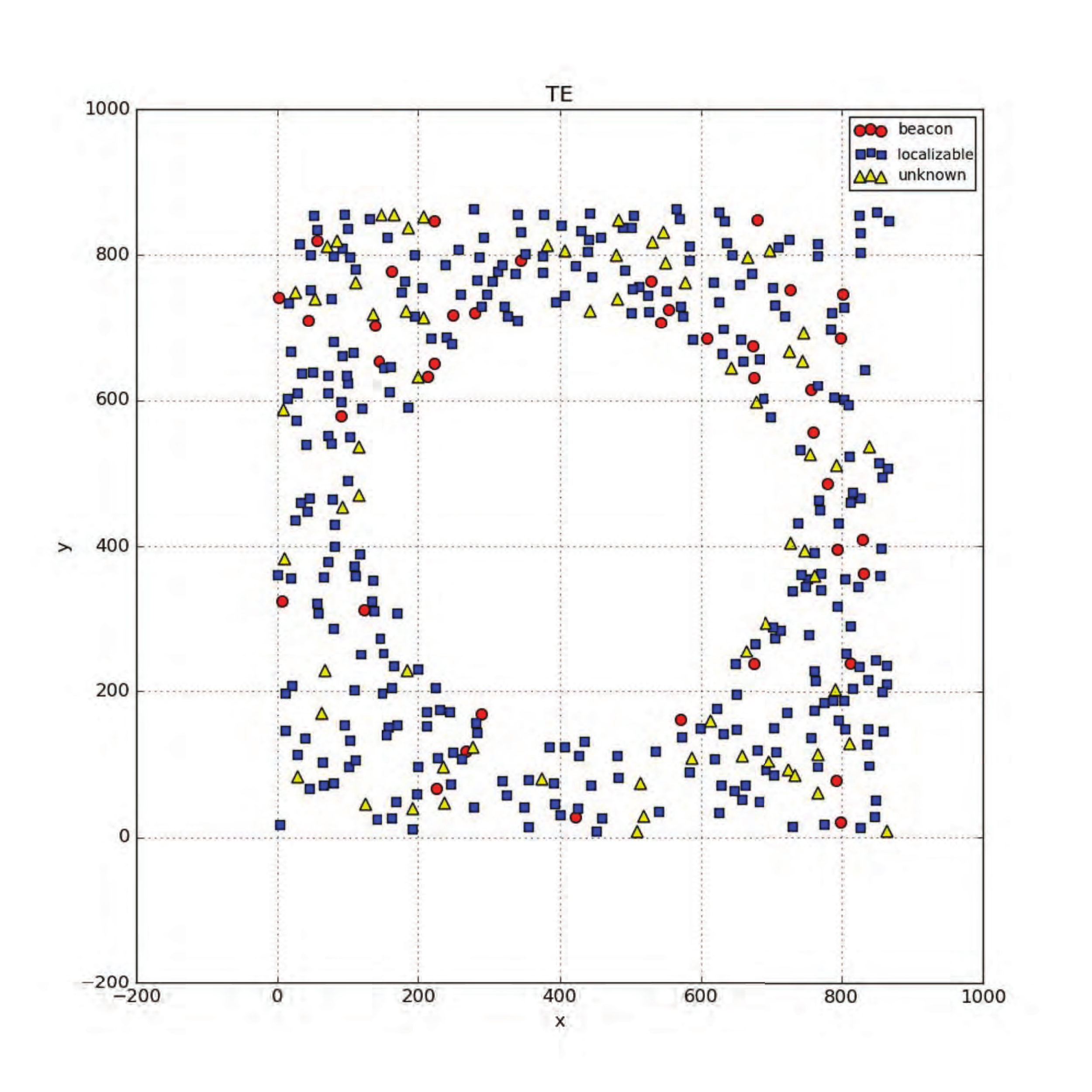}}
\subfigure[WE]{
\includegraphics[width=0.322\textwidth, height = 4cm]{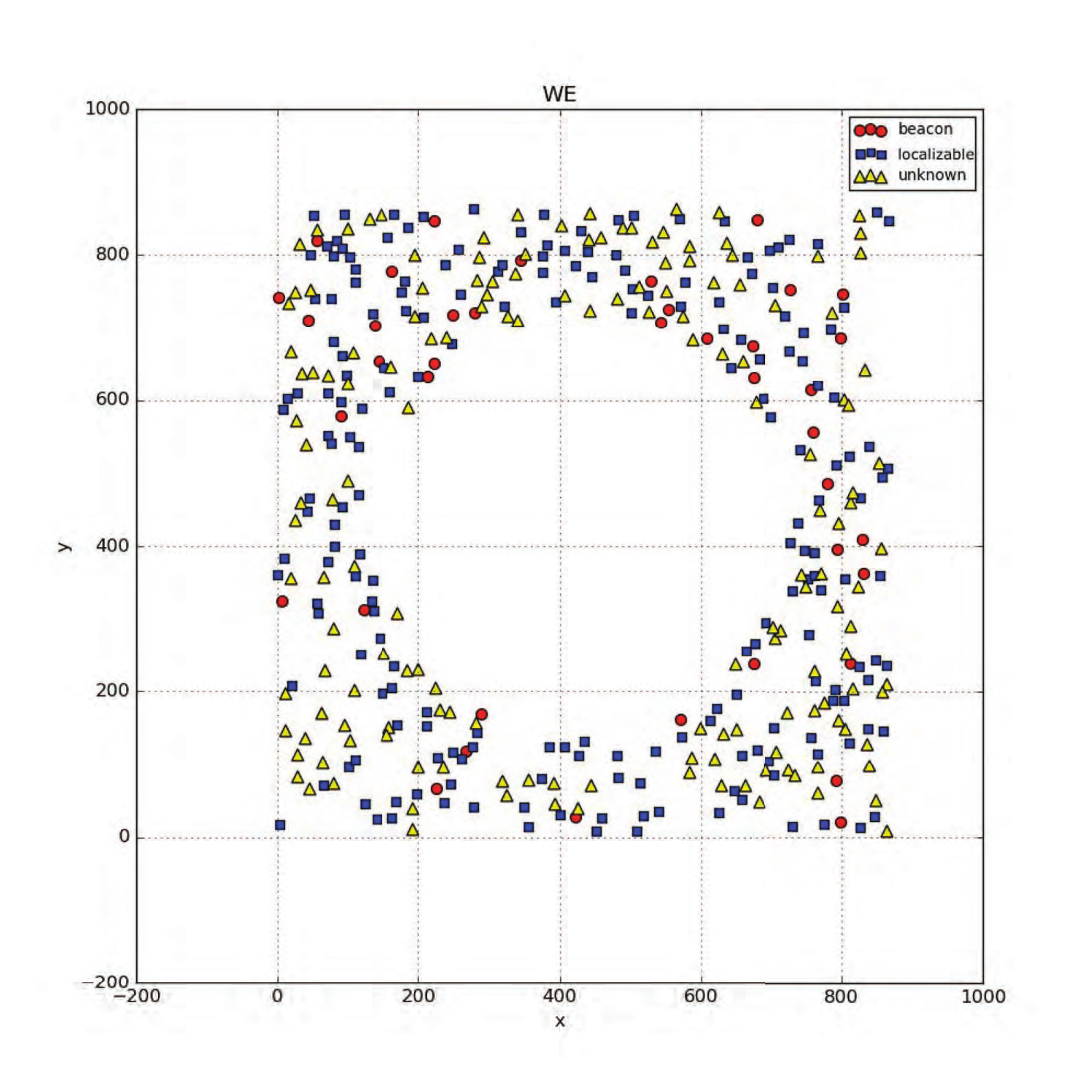}}
\subfigure[TP]{
\includegraphics[width=0.322\textwidth, height = 4cm]{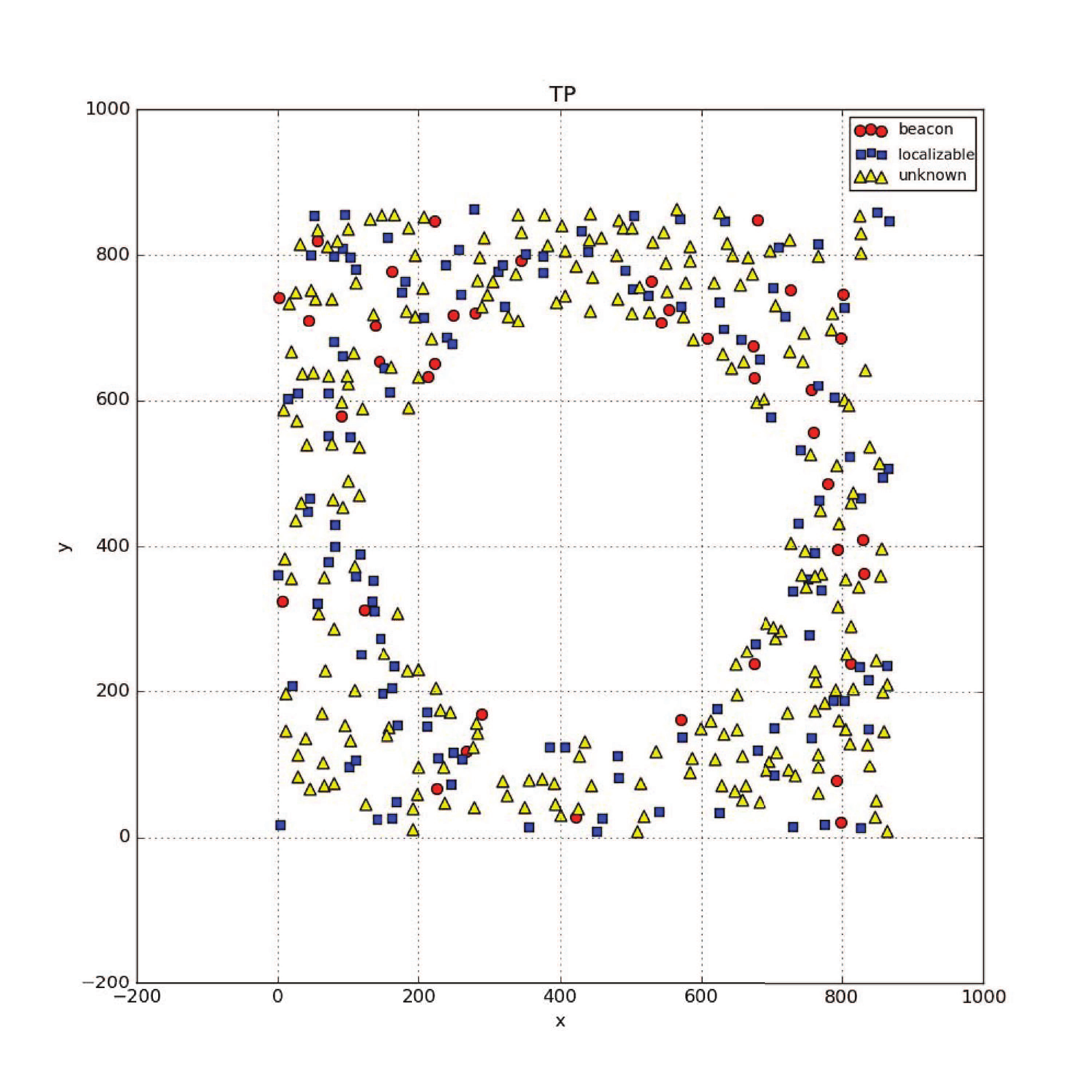}
}
\caption{Hole simulation (sparse network) T0: N = 3.2, B =0.1}
\label{Holelocation}
\end{figure*}

We then adjusted $B$ to 0.05 and carried out two different beacon deployment distributions, \emph{random} and \emph{skewed}, as shown in Fig.\ref{Holesimulation1} and Fig.\ref{Holesimulation2}, respectively.  The random beacon deployment is shown in Fig.\ref{Holesimulation1}. Fig.\ref{Holesimulation2} shows that, even with the same beacon density, when beacons were densely deployed in some corners, both TP and WE could work in these corners.  In contrast, TE worked in both deployments. It was not affected by the sparse beacon deployment, { whether} skewed or not.
\begin{figure*}[!htbp]
\centering
\subfigure[TE]{
\includegraphics[width=0.32\textwidth, height = 4cm]{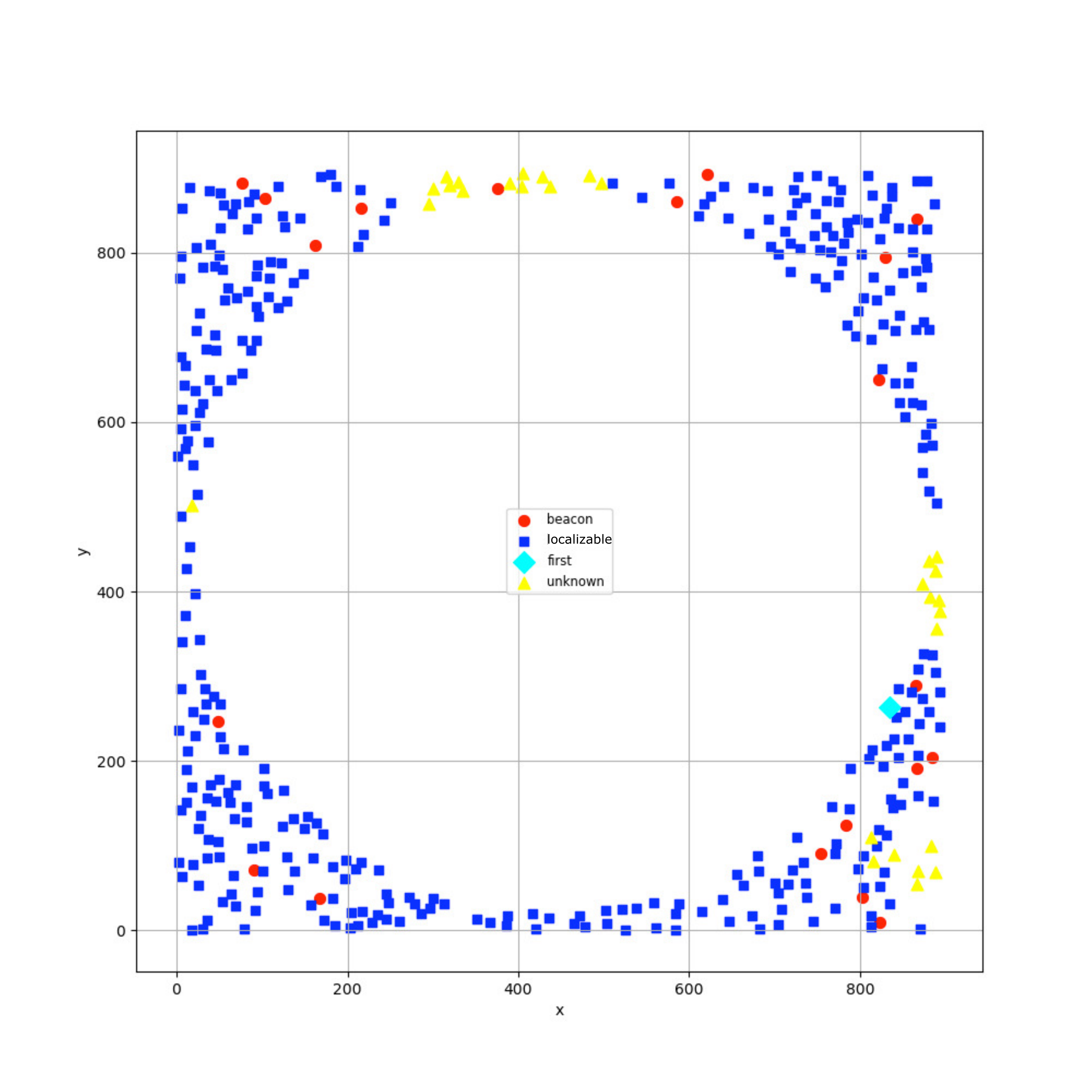}}
\subfigure[WE]{
\includegraphics[width=0.32\textwidth, height = 4cm]{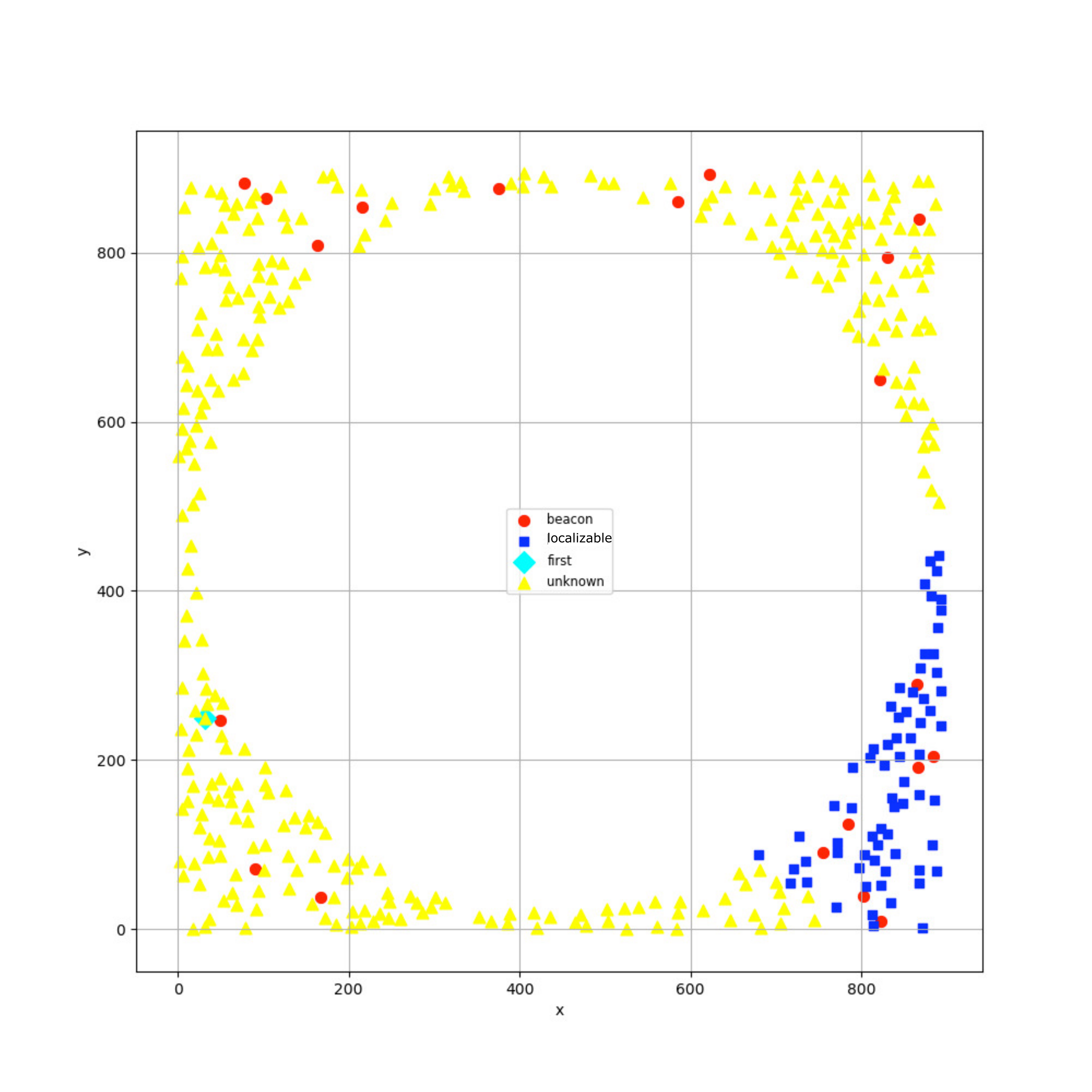}}
\subfigure[TP]{
\includegraphics[width=0.32\textwidth, height = 4cm]{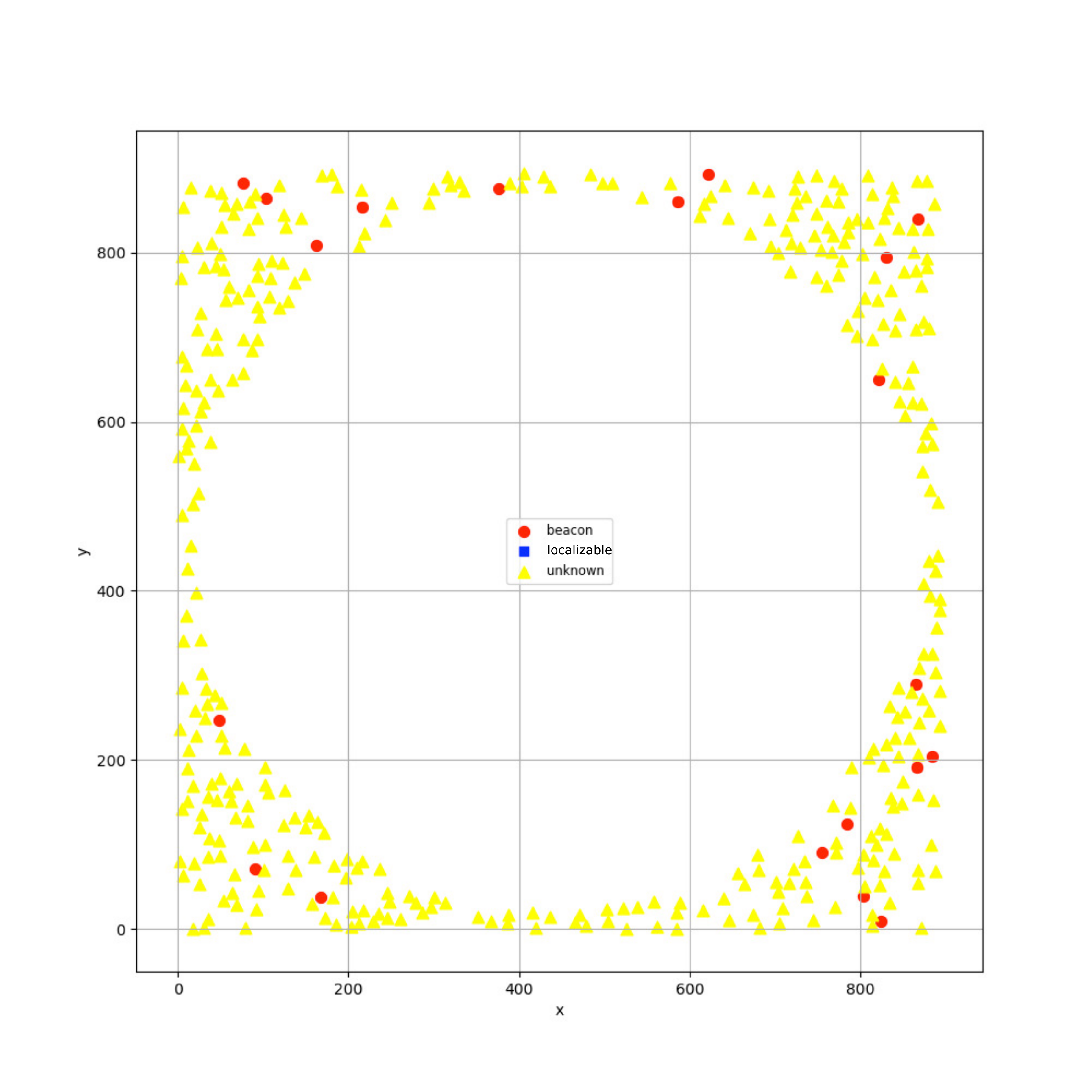}
}
\caption{Hole simulation (dense network, random beacon deployment) T1: N = 2.4, B =0.05}
\label{Holesimulation1}
\end{figure*}
\begin{figure*}[!htbp]
\centering
\subfigure[TE]{
\includegraphics[width=0.32\textwidth, height = 4cm]{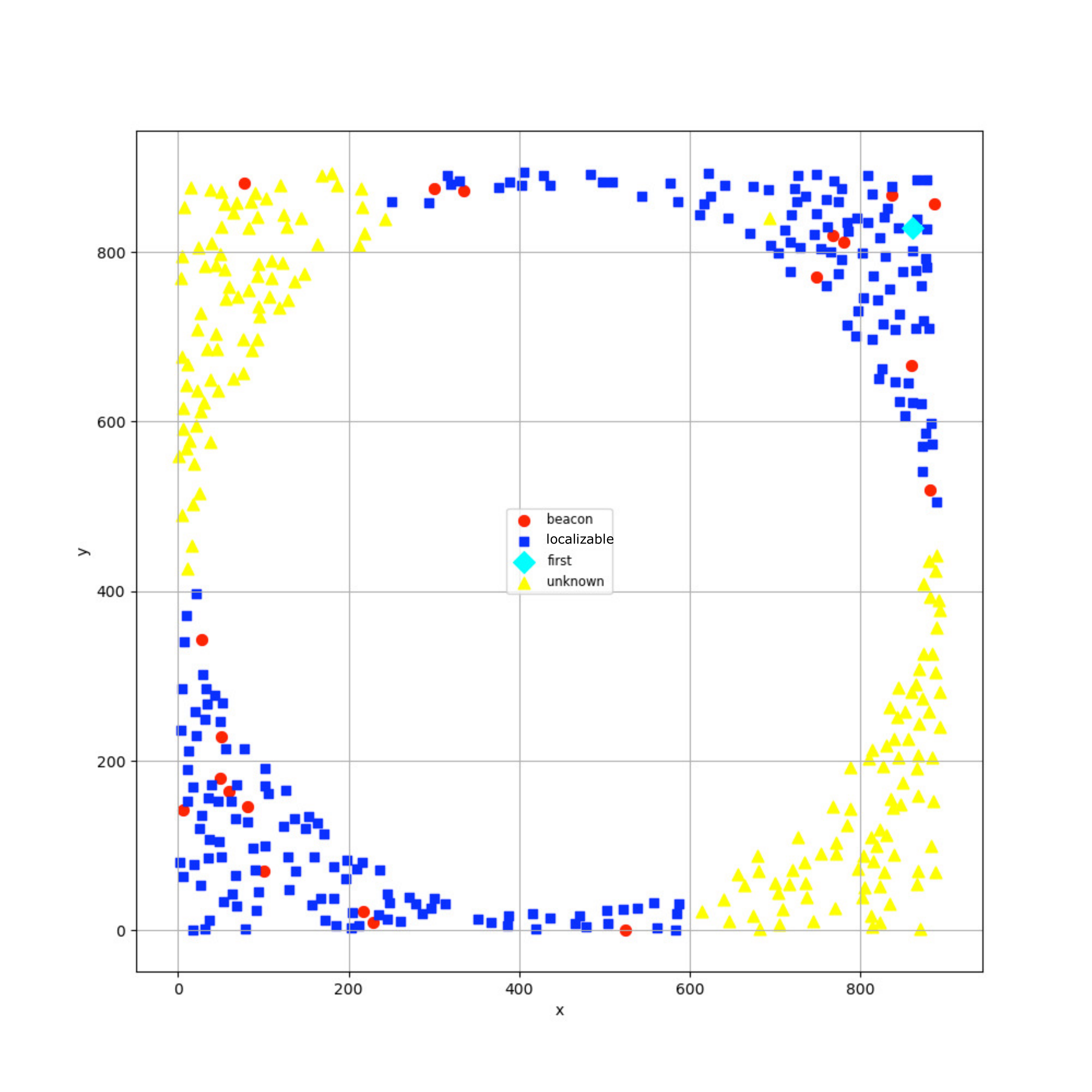}}
\subfigure[WE]{
\includegraphics[width=0.32\textwidth, height = 4cm]{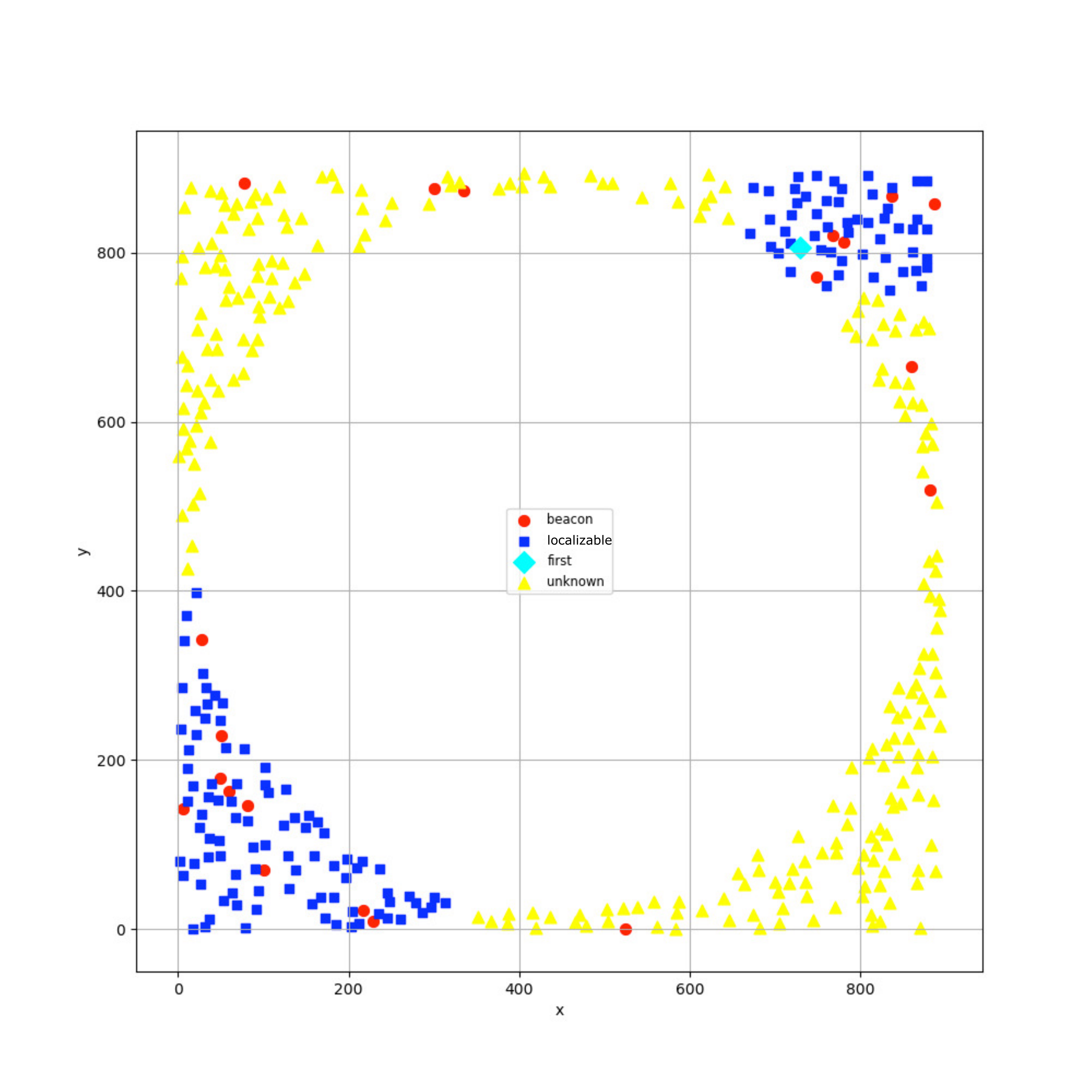}}
\subfigure[TP]{
\includegraphics[width=0.32\textwidth, height = 4cm]{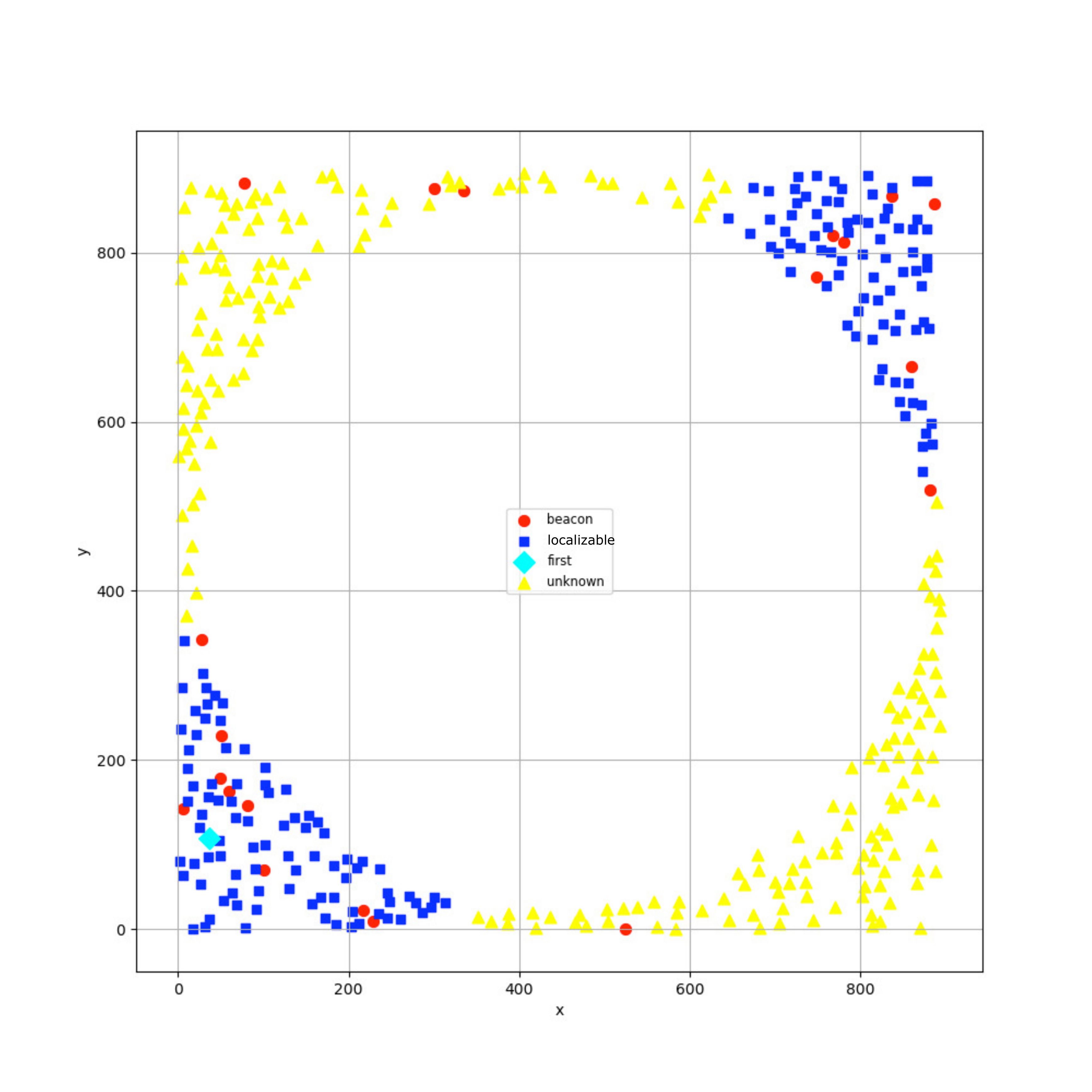}
}
\caption{Hole simulation (dense network, skewed beacon deployment) T2: N = 2.4, B =0.05}
\label{Holesimulation2}
\end{figure*}

Finally, we estimated the energy consumption of the algorithms using the execution time and the node electric current. In the energy consumption measurement in VMNet \cite{vmnet}, the average electric current { of} a working node is about 20 mA. { The execution time of each algorithm} can thus indicate its energy consumption on a sensor node. The simulations were performed under different $N$ and $B$ parameter { values}. Fig.\ref{energy} shows the simulation results.  The values on the Y-axis are the numbers of run cycles. Suppose the time length of each run cycle is $T$ and an algorithm runs $P$ cycles: The energy consumption of the algorithm is estimated as $E= 20~(mA)/1000\times 3~(V)\times P \times T~(s) = 0.06PT~(J)$, where $3~(V)$ is the voltage of the batteries of a sensor node.  The algorithms were driven to detect from 25\% to 50\% of the localizable nodes in a WSN. However, some algorithms failed to reach the required localizability detection percentages.  We set a timeout to stop the algorithms as can be seen by the columns that reach the { maximum} Y (350) in  Fig.\ref{energy}.  The results show that TE consumed the { least} energy to find the same number of localizable nodes.

\begin{figure}[]
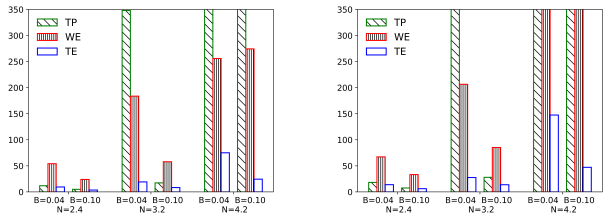

\centering
\subfigure[25\% localizable nodes detected]{
\includegraphics[width=0.23\textwidth]{bar3NBtl100.pdf}}
\subfigure[50\% localizable nodes detected]{
\includegraphics[width=0.23\textwidth]{bar3NBtl200.pdf}}
\caption{Energy consumption(Unit: 0.06T Joule)}
\label{energy}
\end{figure}

\vspace{-0.3cm}

\subsection{Experiments}

We performed two series of experiments to evaluate the performance of our TE { algorithm} with up to 14 TelosB wireless sensor nodes. The same detection program was installed on the sensor nodes, each of which was assigned a unique ID. The hardware configuration of a sensor node is listed in Table \ref{mote}.

\begin{table}[ht]\footnotesize
\caption{Experiment setup} 
\centering 
\label{mote}
\begin{tabular}{|c|c|} 
\hline 
Operating System: & TinyOS\\
\hline
Processor: & 16-bit RISC \\ 
\hline
Memory: & 48 kB Flash and 10KB RAM\\
\hline
Bandwidth: & 250 kbps\\
\hline 
\end{tabular}
\end{table}

In order to take photos with all sensor nodes in our experiments, we reduced the node radio power to the minimal level to limit the network deployment area. The photos can thus indicate the status of all the nodes in each network through the LED lights of the sensor nodes. The red, yellow, and blue  LED lights are used to represent the three states, flexible, rigid, and localizable, respectively.
\vspace{-0.3cm}
\subsubsection{Experiment 1}
The deployment of the sensor nodes of a WSN in our experiments is shown in Fig.\ref{distribution1}. The nodes with red light turned on in the figure functioned as beacons. Our manual deployment ensured that certain nodes have to communicate with the beacons via multi-hop. This setup enables nodes to start the direction-extension to construct a triangle block. Otherwise, the nodes in the WSN might either determine their localizability without extension or be theoretically non-localizable.

\begin{figure}[ht]
 \centering
 \includegraphics[width=0.15\textwidth, height = 0.8cm]{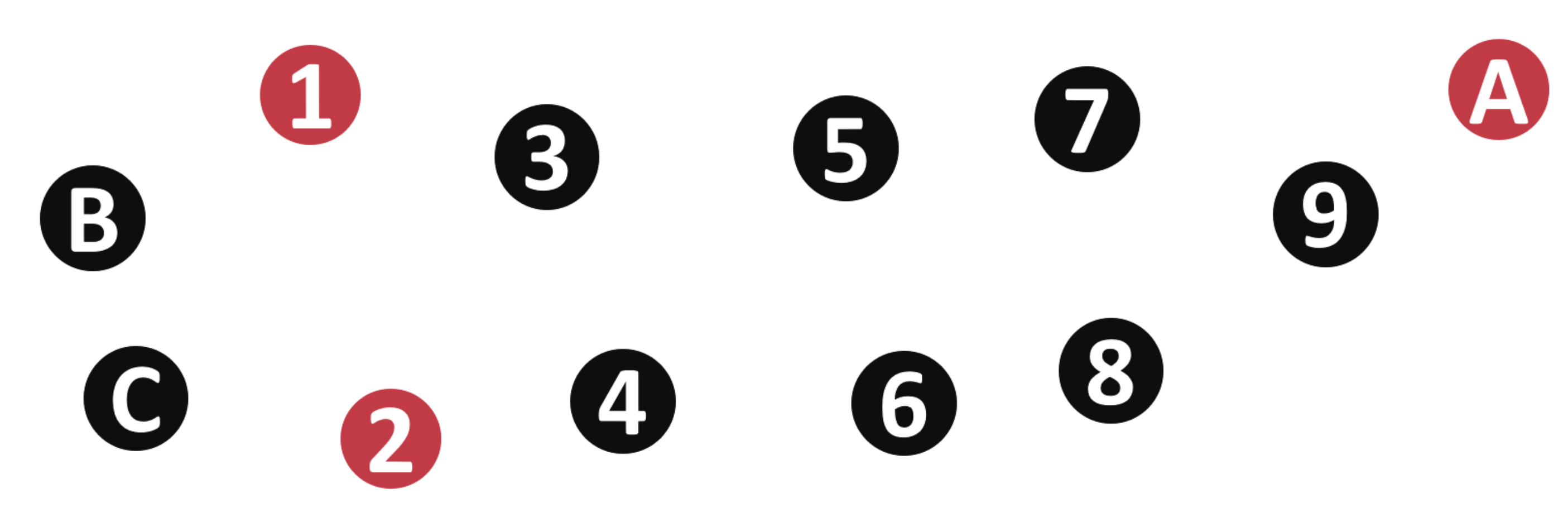}\\
 \caption{Node deployment in Experiment 1}\label{distribution1}
\end{figure}

In the WSNs { for the} experiments, the three beacons were sequentially set up within three seconds. Then, the pair of beacon nodes 1 and 2 launched the extensions. The LED light switching sequences were recorded manually. The LED light switching sequence corresponds to the temporal order of the state transition events. More specifically, the sequence of switching from red to yellow LED, named as the yellow sequence, indicates the transition from the flexible state to the rigid state; similarly, the blue sequence refers to state transition from the rigid state to the localizable state. The two recorded sequences are listed in Table \ref{e1Sequence}. In { the} table, { nodes whose LEDs switched simultaneously are combined as a single element}. The yellow sequence took about five seconds, and the blue sequence, three seconds.

\begin{table}[ht]\footnotesize
\caption{Sequences in experiment 1} 
\centering 
\label{e1Sequence}
\begin{tabular}{|c|c|} 
\hline
Yellow sequence & 3 and C, 4 and B, 5, 6, 7, 8, 9 \\ [0.5ex]
\hline 
Blue sequence & 9, 8 and 7, 6 and 5, 3 and 4 \\
\hline
\end{tabular}
\end{table}

The blue sequence finished more quickly as a localizable node also informs a pair of its parents { thus spreading} the information faster. Fig.\ref{e11} shows the final status of the whole network. Nodes $B$ and $C$ were not included in the blue sequence because they stayed rigid. The two branches grew in the triangle block, and only one branch could find the third beacon, $A$. The branch containing $B$ and $C$ could not find beacon $A$. In fact, the branch containing $B$ and $C$ is theoretically non-localizable.

\begin{figure}[ht]
 \centering
 \includegraphics[width=0.23\textwidth, height = 1.5cm]{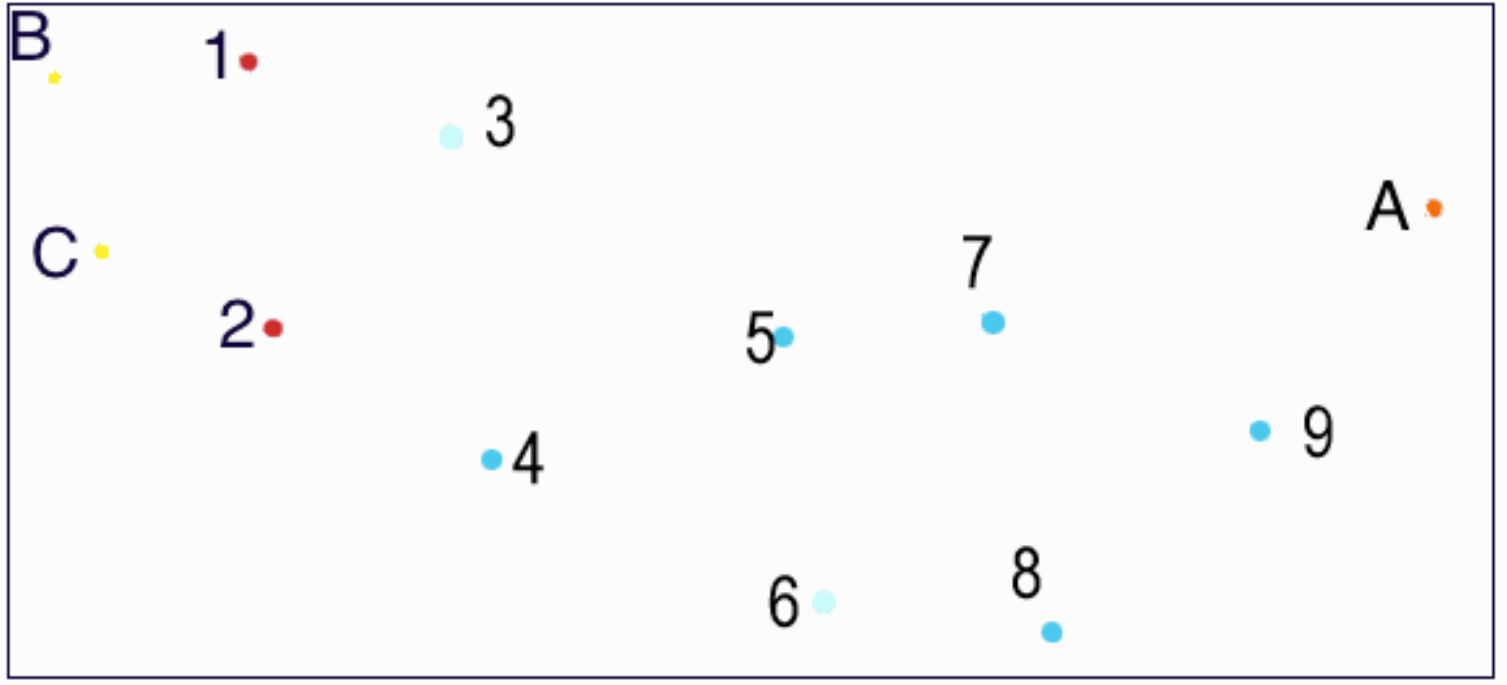}\\
 \caption{Snapshot of Experiment1 }\label{e11}
\end{figure}

\vspace{-0.3cm}

\subsubsection{Experiment 2}

In this experiment, the sensor nodes were divided into two groups. The IDs and approximate locations of the sensor nodes are shown in Fig.\ref{distribution3}. The distance between groups was changed to test the Dual-V-Topo scenario and { to verify that} our dual-v-detection method works. In each group, two nodes were chosen as a pair of beacons to launch extensions. To reproduce the meeting of two triangle blocks, two nodes in different groups could not receive messages from each other initially. After each group finished constructing its own triangle block, some border nodes in one group would be moved closer to the other group. Consequently, the two groups could form a Dual-V-Topo.

\begin{figure}[ht]
 \centering
 \includegraphics[width=0.15\textwidth, height = 1.5 cm]{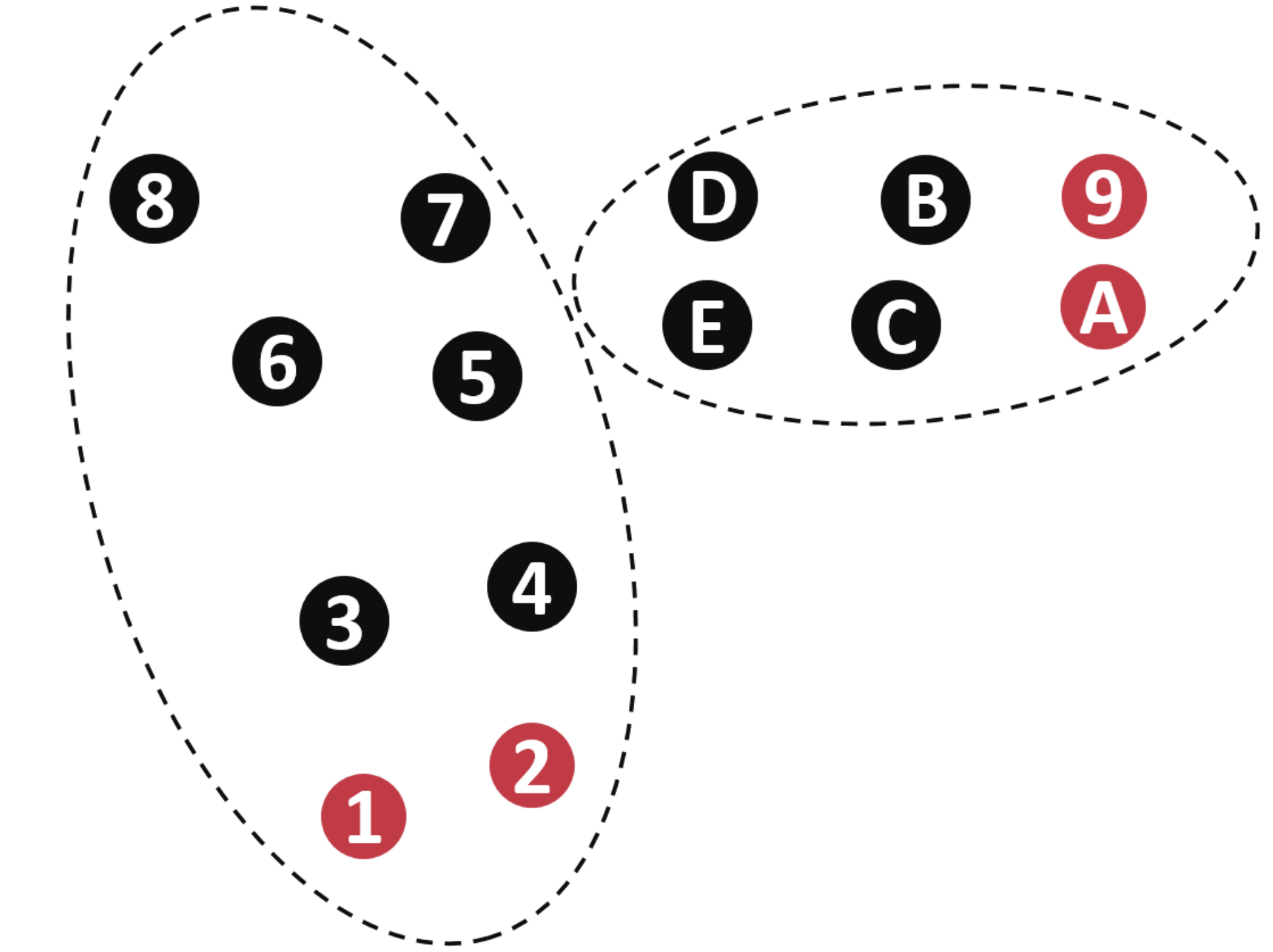}\\
 \caption{Node deployment in Experiment2}\label{distribution3}
\end{figure}

In the WSN of Fig.\ref{distribution3}, we moved the rigid nodes D and E closer to rigid nodes 5 and 7. Table \ref{e2Sequence} lists the blue sequences of the two groups. It can be seen from the table that dual-v-detection is able to detect the localizability of nodes in the meeting area. {  Each node in the meeting area collected the messages from two nodes in the different triangle blocks and then sent a query message to its neighbors}. In this experiment, node E communicated with nodes 7 and 5. Node E inquired of its parent D whether node D could also communicate with node 7 or 5. Node E informed node D of the success of the dual-v-detection process when node D replied with a confirmation message to node E. Then, the rigid nodes in this group sequentially turned on blue LED lights, indicating that these rigid nodes { were} all localizable.

\begin{table}[ht]\footnotesize
\caption{Blue sequences in Experiment 2} 
\centering 
\label{e2Sequence}
{
\begin{tabular}{|c|c|} 
\hline
Left group & 7 and 5, 6 and 4, 3 \\
\hline 
Right group & E and D, C and B \\
\hline
\end{tabular}
}
\end{table}

Fig.\ref{e31} shows the final LED light status of the WSN in Experiment 2. In this figure, the LED light of node 8 is yellow although that of its parent, node 7, is blue. The reason is that the localizability of a node does not help its child to detect its localizability but { instead} helps its ancestors to detect their localizability. Consequently, as node 7 is the parent of node 8, node 8 could not be determined { as} localizable even though node 7 is localizable.

\begin{figure}[ht]
 \centering
 \includegraphics[width=0.22\textwidth, height = 1.5cm]{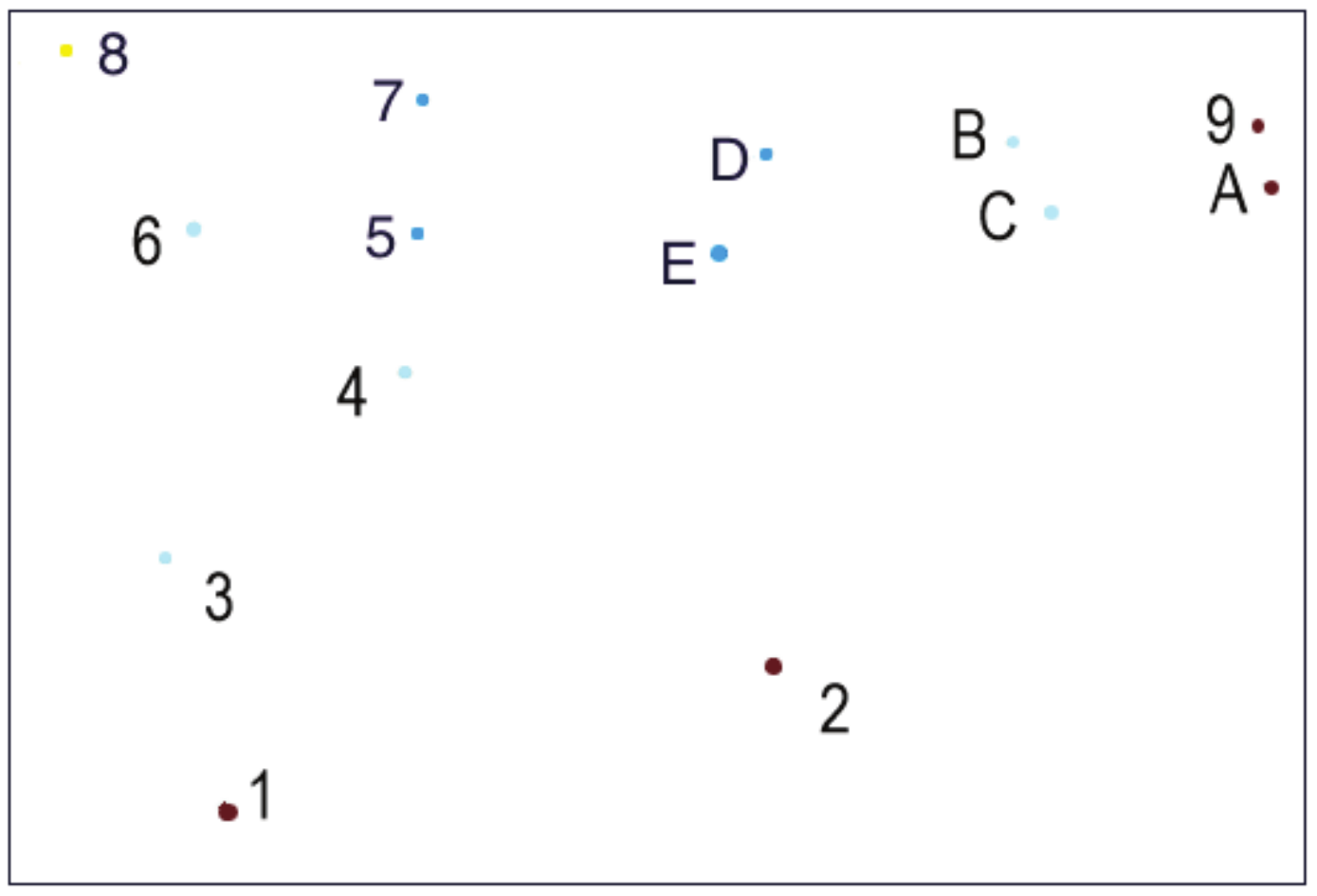}\\
 \caption{Snapshot of Experiment2 }\label{e31}
\end{figure}

\vspace{-0.2cm}
\section{Conclusion and Future Work}

Determining theoretically localizable and non-localizable nodes in a
WSN is important for most localization algorithms and applications.
In this paper, we propose a distributed algorithm, TE, to determine the localizable nodes in a network based on graph rigidity theory. TE uses an efficient approach of triangle extension to construct a rigid graph to detect the localizable nodes and needs less information than the existing algorithms. We theoretically analyzed the efficiency of TE { and compared it to that of} the existing algorithms. Simulations and experiments also demonstrated that TE is applicable to { real-world} WSNs. A promising direction is to integrate TE with localization algorithms.

\appendices

\section{Proof of Lemma \ref{lextension}}

\begin{proof} During an extension operation process, suppose a node, denoted as $v$, $v\notin V_{G}$, is added to $V_{G}$ and edges $(v,v_{1}),(v,v_{2}),(v_{1},v_{2})\in V_{G}$ are added to $E_{G}$. Now a new graph $G_{1}$ is created after the extension. As $G$ is minimally rigid, it has $|E_{G}| = 2|V_{G}|-3$; since $|V_{G_{1}}| =| V_{G}|+1$ and $|E_{G_{1}}| = |E_{G}| +2$, Equation (1) can be { derived}:

{
\centering{~~~~~~~~~~~~~~~~~~~~~~ $|E_{G_{1}}| = 2|V_{G_{1}}|-3~~~~~~~~~~~~~~~~~~~~~~~~~~~~~~~~~~~~~$(1)} \\
}

With equation (1), we now only need to prove the following condition according to Laman's Lemma to prove $G_{1}$ is minimally rigid:
$|E[X]|\leq 2|X|-3$ for each $X\subset V_{G_{1}}$ with $2 \leq |X| \leq |V_{G_{1}}|-1$.

If $v\notin X$, $X \subset V_{G}$, then according to Laman's Lemma, $|E[X]|\leq 2|X|-3$, since $G$ is minimally rigid.

If $v\in X$, we prove by contradiction. Suppose $|E[X]|> 2|X|-3$. We first remove node $i$ from $X$ and up to two relevant edges from $E[X]$. Then, $|E[X]|> 2|X|-3$ is still true. However, this contradicts the given condition that $G$ is minimally rigid and thus $|E[X]|\leq 2|X|-3$, since $X \in V_{G}$ after the removal. Therefore, $G_{1}$ is minimally rigid.  \end{proof}

\vspace{-0.6cm}
\section{Numeric results in simulations}
This section gives the numeric values in our simulation, as mentioned in Section\ref{simulations}.
\begin{table}[!htb]
\caption{Proportion of localizable nodes found by TE} \label{teresult}
\centering 
\resizebox{0.515\textwidth}{!}{
\begin{tabular}{|*{21}{c|}}
\hline
\diagbox{B}{N} & 2.0 & 2.2 & 2.4 & 2.6 & 2.8 & 3.0 & 3.2 & 3.4 & 3.6 & 3.8 & 4.0 & 4.2 & 4.4 & 4.6 & 4.8 & 5.0 & 5.2 & 5.4 & 5.6 & 5.8 \\
\hline
0.01 & 0.666 & 0.332 & 0.994 & 0.665 & 0.332 & 0.329 & 0.549 & 0.636 & 0.000 & 0.000 & 0.168 & 0.000 & 0.000 & 0.000 & 0.000 & 0.000 & 0.000 & 0.000 & 0.000 & 0.000 \\
\hline
0.02 & 0.999 & 0.996 & 0.997 & 0.956 & 0.999 & 0.931 & 0.969 & 0.991 & 0.644 & 0.300 & 0.429 & 0.395 & 0.171 & 0.000 & 0.000 & 0.000 & 0.000 & 0.000 & 0.000 & 0.000 \\
\hline
0.03 & 0.991 & 0.996 & 0.993 & 0.973 & 0.997 & 0.978 & 0.957 & 0.973 & 0.897 & 0.954 & 0.627 & 0.552 & 0.193 & 0.033 & 0.000 & 0.000 & 0.007 & 0.014 & 0.000 & 0.000 \\
\hline
0.04 & 0.999 & 0.949 & 0.997 & 0.996 & 0.953 & 0.993 & 0.973 & 0.976 & 0.957 & 0.818 & 0.910 & 0.543 & 0.296 & 0.166 & 0.025 & 0.016 & 0.008 & 0.000 & 0.000 & 0.000 \\
\hline
0.05 & 0.992 & 0.989 & 0.995 & 0.984 & 0.988 & 0.996 & 0.992 & 0.980 & 0.968 & 0.958 & 0.822 & 0.715 & 0.423 & 0.154 & 0.037 & 0.054 & 0.047 & 0.000 & 0.000 & 0.001 \\
\hline
0.06 & 0.992 & 0.994 & 0.996 & 0.999 & 0.976 & 0.989 & 0.991 & 0.996 & 0.973 & 0.954 & 0.941 & 0.821 & 0.537 & 0.137 & 0.095 & 0.082 & 0.016 & 0.020 & 0.002 & 0.000 \\
\hline
0.07 & 0.988 & 0.992 & 0.994 & 0.990 & 0.996 & 0.969 & 0.993 & 0.987 & 0.963 & 0.945 & 0.884 & 0.800 & 0.686 & 0.415 & 0.165 & 0.116 & 0.018 & 0.000 & 0.000 & 0.014 \\
\hline
0.08 & 0.998 & 0.998 & 0.989 & 0.989 & 0.992 & 0.973 & 0.986 & 0.979 & 0.988 & 0.964 & 0.940 & 0.880 & 0.717 & 0.375 & 0.254 & 0.037 & 0.043 & 0.039 & 0.004 & 0.005 \\
\hline
0.09 & 1.000 & 0.994 & 0.997 & 0.988 & 0.993 & 0.989 & 0.986 & 0.978 & 0.989 & 0.984 & 0.908 & 0.853 & 0.750 & 0.483 & 0.210 & 0.117 & 0.046 & 0.070 & 0.016 & 0.005 \\
\hline
0.1 & 0.988 & 0.994 & 0.990 & 0.996 & 0.994 & 0.996 & 0.992 & 0.996 & 0.970 & 0.972 & 0.944 & 0.854 & 0.722 & 0.523 & 0.318 & 0.125 & 0.074 & 0.028 & 0.020 & 0.004 \\
\hline
0.11 & 0.987 & 0.992 & 0.993 & 0.993 & 0.998 & 0.975 & 0.996 & 0.985 & 0.988 & 0.978 & 0.917 & 0.874 & 0.707 & 0.538 & 0.266 & 0.141 & 0.075 & 0.052 & 0.043 & 0.024 \\
\hline
0.12 & 0.995 & 0.997 & 0.995 & 0.998 & 0.993 & 0.994 & 0.988 & 0.986 & 0.980 & 0.973 & 0.964 & 0.906 & 0.826 & 0.630 & 0.371 & 0.296 & 0.132 & 0.044 & 0.031 & 0.006 \\
\hline
0.13 & 0.995 & 0.995 & 0.997 & 0.997 & 0.995 & 0.988 & 0.989 & 0.969 & 0.988 & 0.969 & 0.972 & 0.856 & 0.732 & 0.712 & 0.314 & 0.223 & 0.106 & 0.123 & 0.037 & 0.023 \\
\hline
0.14 & 0.999 & 0.996 & 0.996 & 0.999 & 0.987 & 0.997 & 0.995 & 0.996 & 0.990 & 0.997 & 0.962 & 0.912 & 0.830 & 0.621 & 0.459 & 0.303 & 0.102 & 0.058 & 0.038 & 0.033 \\
\hline
0.15 & 0.992 & 0.994 & 0.995 & 0.998 & 0.998 & 0.990 & 0.993 & 0.997 & 0.988 & 0.978 & 0.950 & 0.919 & 0.855 & 0.742 & 0.383 & 0.289 & 0.165 & 0.090 & 0.052 & 0.052 \\
\hline
0.16 & 0.994 & 0.998 & 1.000 & 0.996 & 0.999 & 0.985 & 0.998 & 0.985 & 0.979 & 0.974 & 0.968 & 0.922 & 0.819 & 0.699 & 0.434 & 0.311 & 0.247 & 0.150 & 0.045 & 0.051 \\
\hline
0.17 & 0.994 & 1.000 & 0.996 & 0.995 & 0.995 & 0.998 & 0.995 & 0.989 & 0.981 & 0.990 & 0.971 & 0.925 & 0.856 & 0.690 & 0.453 & 0.339 & 0.247 & 0.119 & 0.047 & 0.034 \\
\hline
0.18 & 0.998 & 0.998 & 0.997 & 0.997 & 0.990 & 0.998 & 0.981 & 0.996 & 0.990 & 0.977 & 0.958 & 0.918 & 0.831 & 0.768 & 0.557 & 0.409 & 0.169 & 0.189 & 0.078 & 0.042 \\
\hline
0.19 & 0.999 & 0.999 & 1.000 & 1.000 & 0.999 & 0.999 & 0.996 & 0.995 & 0.989 & 0.965 & 0.962 & 0.954 & 0.853 & 0.711 & 0.565 & 0.403 & 0.259 & 0.154 & 0.069 & 0.100 \\
\hline
0.2 & 0.998 & 0.995 & 0.998 & 0.999 & 0.999 & 0.996 & 0.997 & 0.995 & 0.993 & 0.980 & 0.976 & 0.940 & 0.876 & 0.748 & 0.599 & 0.424 & 0.256 & 0.147 & 0.118 & 0.059 \\
\hline
\end{tabular}
}
\end{table}
\vspace{-0.4cm}
\begin{table}[!htbp]
\caption{Proportion of localizable nodes found by TP} 
\centering 
\label{tpresult}
\resizebox{0.515\textwidth}{!}{
\begin{tabular}{|*{21}{c|}}
\hline
\diagbox{B}{N} & 2.0 & 2.2 & 2.4 & 2.6 & 2.8 & 3.0 & 3.2 & 3.4 & 3.6 & 3.8 & 4.0 & 4.2 & 4.4 & 4.6 & 4.8 & 5.0 & 5.2 & 5.4 & 5.6 & 5.8 \\
\hline
0.01 & 0.000 & 0.000 & 0.000 & 0.000 & 0.000 & 0.000 & 0.000 & 0.000 & 0.000 & 0.000 & 0.000 & 0.000 & 0.000 & 0.000 & 0.000 & 0.000 & 0.000 & 0.000 & 0.000 & 0.000 \\
\hline
0.02 & 0.333 & 0.333 & 0.000 & 0.000 & 0.000 & 0.000 & 0.000 & 0.000 & 0.000 & 0.000 & 0.000 & 0.000 & 0.000 & 0.000 & 0.000 & 0.000 & 0.000 & 0.000 & 0.000 & 0.000 \\
\hline
0.03 & 0.000 & 0.333 & 0.332 & 0.333 & 0.000 & 0.000 & 0.329 & 0.000 & 0.000 & 0.000 & 0.000 & 0.000 & 0.000 & 0.000 & 0.000 & 0.000 & 0.000 & 0.000 & 0.000 & 0.000 \\
\hline
0.04 & 1.000 & 1.000 & 0.333 & 0.000 & 0.000 & 0.000 & 0.272 & 0.194 & 0.000 & 0.000 & 0.000 & 0.000 & 0.000 & 0.000 & 0.000 & 0.000 & 0.000 & 0.000 & 0.000 & 0.000 \\
\hline
0.05 & 1.000 & 1.000 & 1.000 & 1.000 & 1.000 & 0.332 & 0.650 & 0.061 & 0.000 & 0.001 & 0.000 & 0.000 & 0.000 & 0.000 & 0.000 & 0.000 & 0.000 & 0.000 & 0.000 & 0.000 \\
\hline
0.06 & 1.000 & 1.000 & 1.000 & 0.667 & 0.998 & 0.998 & 0.326 & 0.515 & 0.161 & 0.000 & 0.000 & 0.004 & 0.000 & 0.000 & 0.000 & 0.000 & 0.000 & 0.000 & 0.000 & 0.000 \\
\hline
0.07 & 1.000 & 1.000 & 1.000 & 0.999 & 1.000 & 0.329 & 0.658 & 0.731 & 0.050 & 0.030 & 0.012 & 0.002 & 0.005 & 0.000 & 0.000 & 0.000 & 0.000 & 0.000 & 0.000 & 0.000 \\
\hline
0.08 & 1.000 & 1.000 & 0.999 & 1.000 & 0.999 & 0.998 & 0.989 & 0.289 & 0.189 & 0.021 & 0.000 & 0.000 & 0.000 & 0.000 & 0.000 & 0.000 & 0.003 & 0.000 & 0.000 & 0.000 \\
\hline
0.09 & 1.000 & 1.000 & 1.000 & 1.000 & 1.000 & 0.998 & 0.990 & 0.781 & 0.279 & 0.000 & 0.025 & 0.002 & 0.013 & 0.000 & 0.000 & 0.000 & 0.000 & 0.000 & 0.000 & 0.000 \\
\hline
0.1 & 1.000 & 1.000 & 1.000 & 1.000 & 0.999 & 0.995 & 0.986 & 0.866 & 0.495 & 0.130 & 0.042 & 0.014 & 0.002 & 0.003 & 0.001 & 0.001 & 0.000 & 0.000 & 0.000 & 0.000 \\
\hline
0.11 & 1.000 & 1.000 & 1.000 & 1.000 & 1.000 & 0.999 & 0.991 & 0.978 & 0.568 & 0.111 & 0.021 & 0.016 & 0.011 & 0.004 & 0.002 & 0.000 & 0.000 & 0.000 & 0.000 & 0.000 \\
\hline
0.12 & 1.000 & 1.000 & 1.000 & 1.000 & 1.000 & 0.996 & 0.989 & 0.906 & 0.697 & 0.137 & 0.053 & 0.020 & 0.004 & 0.001 & 0.001 & 0.001 & 0.000 & 0.004 & 0.001 & 0.000 \\
\hline
0.13 & 1.000 & 1.000 & 1.000 & 0.998 & 0.999 & 0.997 & 0.994 & 0.953 & 0.678 & 0.261 & 0.102 & 0.043 & 0.010 & 0.005 & 0.004 & 0.003 & 0.003 & 0.000 & 0.001 & 0.000 \\
\hline
0.14 & 1.000 & 1.000 & 1.000 & 1.000 & 0.999 & 0.997 & 0.996 & 0.979 & 0.801 & 0.329 & 0.130 & 0.047 & 0.008 & 0.007 & 0.003 & 0.000 & 0.002 & 0.002 & 0.000 & 0.000 \\
\hline
0.15 & 1.000 & 1.000 & 1.000 & 1.000 & 0.999 & 0.996 & 0.987 & 0.963 & 0.825 & 0.370 & 0.125 & 0.070 & 0.040 & 0.023 & 0.004 & 0.002 & 0.002 & 0.002 & 0.002 & 0.003 \\
\hline
0.16 & 1.000 & 1.000 & 1.000 & 1.000 & 0.999 & 0.999 & 0.993 & 0.954 & 0.827 & 0.365 & 0.162 & 0.123 & 0.029 & 0.025 & 0.010 & 0.002 & 0.001 & 0.003 & 0.001 & 0.000 \\
\hline
0.17 & 1.000 & 1.000 & 1.000 & 1.000 & 0.998 & 0.999 & 0.989 & 0.945 & 0.769 & 0.367 & 0.252 & 0.101 & 0.049 & 0.017 & 0.016 & 0.013 & 0.006 & 0.003 & 0.003 & 0.000 \\
\hline
0.18 & 1.000 & 1.000 & 1.000 & 0.999 & 1.000 & 0.994 & 0.989 & 0.968 & 0.771 & 0.646 & 0.292 & 0.081 & 0.030 & 0.019 & 0.014 & 0.001 & 0.001 & 0.004 & 0.002 & 0.002 \\
\hline
0.19 & 1.000 & 1.000 & 0.999 & 1.000 & 0.999 & 1.000 & 0.985 & 0.969 & 0.903 & 0.714 & 0.271 & 0.119 & 0.044 & 0.029 & 0.017 & 0.004 & 0.007 & 0.001 & 0.004 & 0.001 \\
\hline
0.2 & 1.000 & 1.000 & 1.000 & 1.000 & 0.999 & 0.997 & 0.986 & 0.978 & 0.893 & 0.655 & 0.308 & 0.108 & 0.062 & 0.028 & 0.009 & 0.009 & 0.006 & 0.002 & 0.002 & 0.005 \\
\hline
\end{tabular}
}
\end{table}

\vspace{-0.4cm}
\begin{table}[!bp]
\caption{Proportion of localizable nodes found by WE} 
\centering 
\label{weresult}
\resizebox{0.515\textwidth}{!}{
\vspace*{-\baselineskip}
\begin{tabular}{|*{21}{c|}}
\hline
\diagbox{B}{N} & 2.0 & 2.2 & 2.4 & 2.6 & 2.8 & 3.0 & 3.2 & 3.4 & 3.6 & 3.8 & 4.0 & 4.2 & 4.4 & 4.6 & 4.8 & 5.0 & 5.2 & 5.4 & 5.6 & 5.8 \\
\hline
0.01 & 0.000 & 0.000 & 0.000 & 0.000 & 0.000 & 0.000 & 0.000 & 0.000 & 0.000 & 0.000 & 0.000 & 0.000 & 0.000 & 0.000 & 0.000 & 0.000 & 0.000 & 0.000 & 0.000 & 0.000 \\
\hline
0.02 & 0.667 & 0.333 & 0.000 & 0.000 & 0.000 & 0.000 & 0.097 & 0.000 & 0.000 & 0.000 & 0.000 & 0.008 & 0.000 & 0.000 & 0.000 & 0.000 & 0.000 & 0.000 & 0.000 & 0.000 \\
\hline
0.03 & 0.336 & 0.667 & 0.155 & 0.333 & 0.199 & 0.000 & 0.033 & 0.025 & 0.000 & 0.000 & 0.000 & 0.000 & 0.000 & 0.000 & 0.000 & 0.000 & 0.000 & 0.000 & 0.000 & 0.000 \\
\hline
0.04 & 1.000 & 1.000 & 0.661 & 0.353 & 0.077 & 0.284 & 0.070 & 0.048 & 0.027 & 0.000 & 0.000 & 0.000 & 0.000 & 0.000 & 0.005 & 0.000 & 0.000 & 0.000 & 0.000 & 0.000 \\
\hline
0.05 & 1.000 & 1.000 & 0.931 & 0.997 & 0.877 & 0.527 & 0.220 & 0.138 & 0.000 & 0.000 & 0.016 & 0.023 & 0.000 & 0.000 & 0.000 & 0.000 & 0.012 & 0.002 & 0.000 & 0.000 \\
\hline
0.06 & 1.000 & 1.000 & 1.000 & 0.995 & 0.952 & 0.443 & 0.390 & 0.086 & 0.069 & 0.064 & 0.031 & 0.023 & 0.000 & 0.009 & 0.004 & 0.000 & 0.008 & 0.003 & 0.003 & 0.000 \\
\hline
0.07 & 1.000 & 1.000 & 1.000 & 1.000 & 0.997 & 0.646 & 0.377 & 0.294 & 0.177 & 0.081 & 0.062 & 0.074 & 0.011 & 0.014 & 0.014 & 0.025 & 0.000 & 0.000 & 0.004 & 0.004 \\
\hline
0.08 & 1.000 & 1.000 & 1.000 & 1.000 & 0.990 & 0.987 & 0.571 & 0.239 & 0.267 & 0.139 & 0.038 & 0.050 & 0.025 & 0.044 & 0.014 & 0.002 & 0.013 & 0.000 & 0.000 & 0.000 \\
\hline
0.09 & 1.000 & 1.000 & 1.000 & 1.000 & 0.986 & 0.931 & 0.435 & 0.416 & 0.285 & 0.100 & 0.122 & 0.043 & 0.055 & 0.028 & 0.026 & 0.002 & 0.019 & 0.009 & 0.004 & 0.005 \\
\hline
0.1 & 1.000 & 1.000 & 1.000 & 0.999 & 0.999 & 0.959 & 0.744 & 0.637 & 0.247 & 0.319 & 0.169 & 0.055 & 0.095 & 0.073 & 0.014 & 0.024 & 0.009 & 0.005 & 0.010 & 0.004 \\
\hline
0.11 & 1.000 & 1.000 & 1.000 & 1.000 & 0.999 & 0.981 & 0.884 & 0.630 & 0.400 & 0.170 & 0.145 & 0.085 & 0.083 & 0.072 & 0.018 & 0.021 & 0.003 & 0.010 & 0.008 & 0.004 \\
\hline
0.12 & 1.000 & 1.000 & 1.000 & 1.000 & 0.996 & 0.975 & 0.826 & 0.680 & 0.420 & 0.285 & 0.262 & 0.113 & 0.131 & 0.048 & 0.051 & 0.022 & 0.013 & 0.015 & 0.018 & 0.015 \\
\hline
0.13 & 1.000 & 1.000 & 1.000 & 1.000 & 1.000 & 0.989 & 0.842 & 0.716 & 0.637 & 0.488 & 0.258 & 0.187 & 0.119 & 0.055 & 0.039 & 0.032 & 0.035 & 0.007 & 0.017 & 0.007 \\
\hline
0.14 & 1.000 & 1.000 & 0.999 & 1.000 & 0.999 & 0.963 & 0.939 & 0.818 & 0.676 & 0.513 & 0.321 & 0.221 & 0.141 & 0.089 & 0.081 & 0.045 & 0.037 & 0.032 & 0.004 & 0.018 \\
\hline
0.15 & 1.000 & 1.000 & 1.000 & 1.000 & 1.000 & 0.995 & 0.946 & 0.844 & 0.757 & 0.537 & 0.396 & 0.291 & 0.169 & 0.121 & 0.109 & 0.111 & 0.023 & 0.046 & 0.017 & 0.013 \\
\hline
0.16 & 1.000 & 1.000 & 0.999 & 1.000 & 1.000 & 0.998 & 0.947 & 0.785 & 0.694 & 0.541 & 0.453 & 0.266 & 0.145 & 0.133 & 0.100 & 0.054 & 0.057 & 0.058 & 0.024 & 0.022 \\
\hline
0.17 & 1.000 & 1.000 & 1.000 & 0.999 & 0.999 & 0.989 & 0.954 & 0.854 & 0.738 & 0.658 & 0.442 & 0.418 & 0.222 & 0.117 & 0.090 & 0.104 & 0.062 & 0.040 & 0.039 & 0.015 \\
\hline
0.18 & 1.000 & 1.000 & 1.000 & 1.000 & 1.000 & 0.984 & 0.994 & 0.892 & 0.777 & 0.616 & 0.474 & 0.348 & 0.222 & 0.150 & 0.146 & 0.083 & 0.032 & 0.073 & 0.033 & 0.032 \\
\hline
0.19 & 1.000 & 1.000 & 1.000 & 0.999 & 0.999 & 0.999 & 0.984 & 0.987 & 0.870 & 0.746 & 0.540 & 0.381 & 0.281 & 0.182 & 0.172 & 0.098 & 0.114 & 0.057 & 0.043 & 0.048 \\
\hline
0.2 & 1.000 & 1.000 & 1.000 & 1.000 & 1.000 & 0.998 & 0.994 & 0.952 & 0.710 & 0.795 & 0.527 & 0.426 & 0.311 & 0.279 & 0.166 & 0.152 & 0.083 & 0.065 & 0.053 & 0.041 \\
\hline
\end{tabular}
}
\end{table}

\ifCLASSOPTIONcaptionsoff
 \newpage
\fi



%
\ifCLASSOPTIONcompsoc
  \section*{Acknowledgments}
\else
  \section*{Acknowledgment}
\fi
This work was supported primarily by the National Natural Science
Foundation of China (NSFC) (Grant No. 61672552). This work was also partially supported by the following NSFC grants: Grant No. 61472452, 61772565, 61472453, 61472453, and the Science and Technology Program of Guangzhou City of China (No. 201707010194).

\vspace*{-40pt}
\begin{IEEEbiography}[{\includegraphics[width=1in,height=1.25in,clip,keepaspectratio]{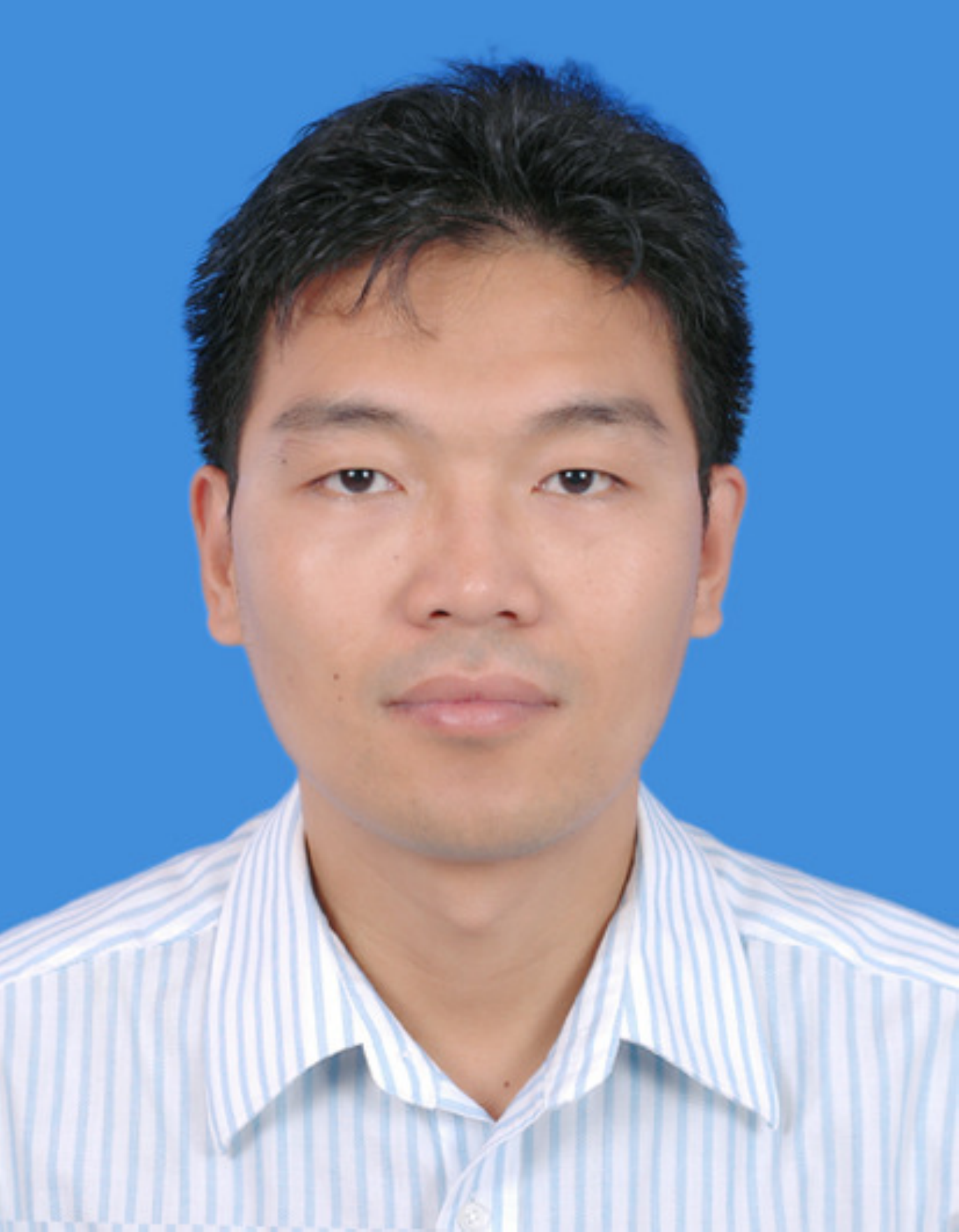}}]{Hejun Wu} received his Ph.D. degree in Computer Science and Engineering from Hong Kong University of Science and Technology in 2008. He is currently an Associate Professor with the School of Data and Computer Science, Sun Yat-sen University, Guangzhou, China. His main research interests include Wireless Sensor Networks and Distributed Computing. He is a member of the IEEE and ACM. 
\end{IEEEbiography}
\vspace*{-40pt}
\begin{IEEEbiography}[{\includegraphics[width=1in,height=1.25in,clip,keepaspectratio]{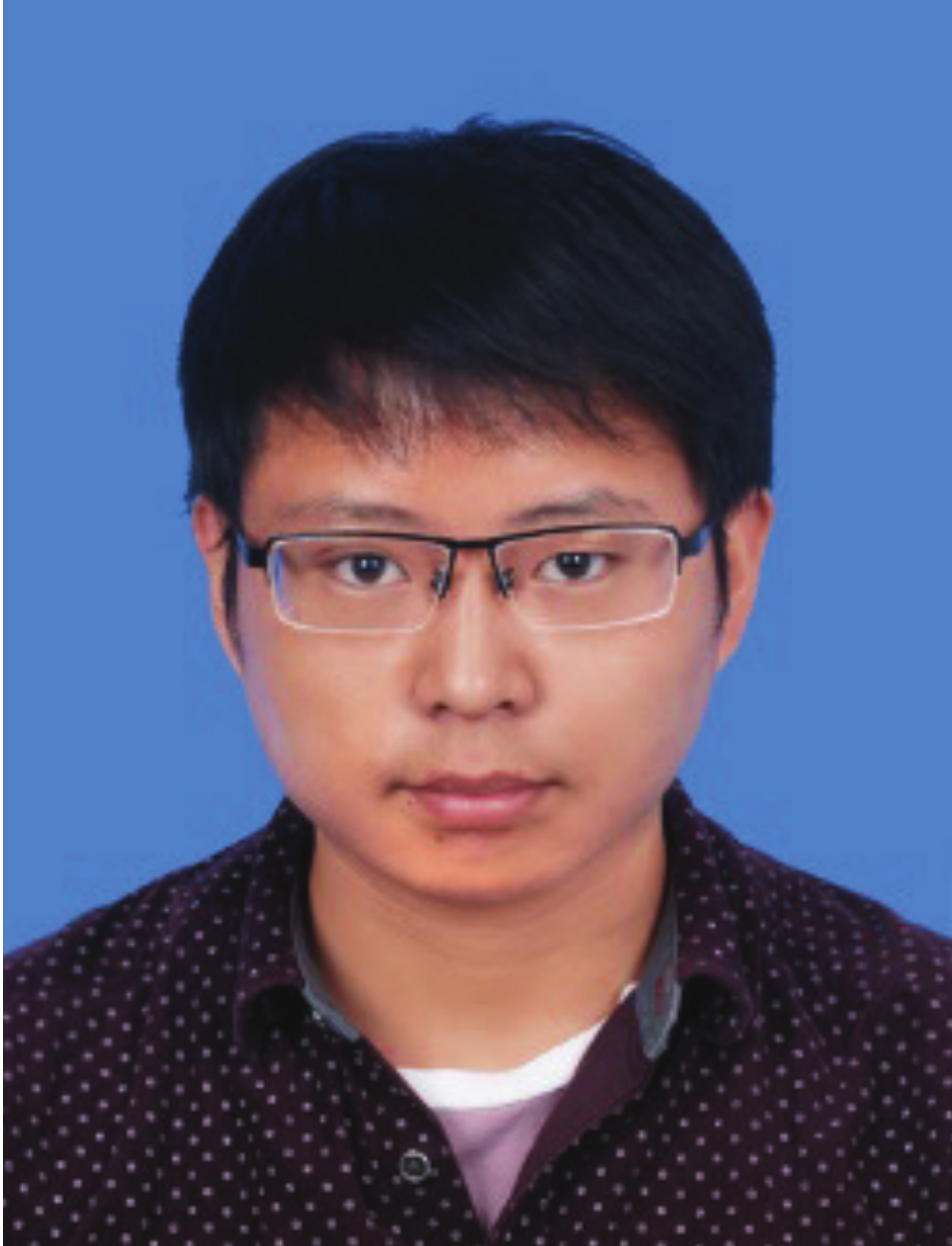}}]
{Ao Ding} received the B.S degree in computer science from Anhui University in 2014 and the M.Phil degree in computer science from Sun Yat-Sen University in 2017. His research interests include wireless sensor networks, distributed computing, and machine learning. 
\end{IEEEbiography}
\vspace*{-40pt}
\begin{IEEEbiography}[{\includegraphics[width=1in,height=1.25in,clip,keepaspectratio]{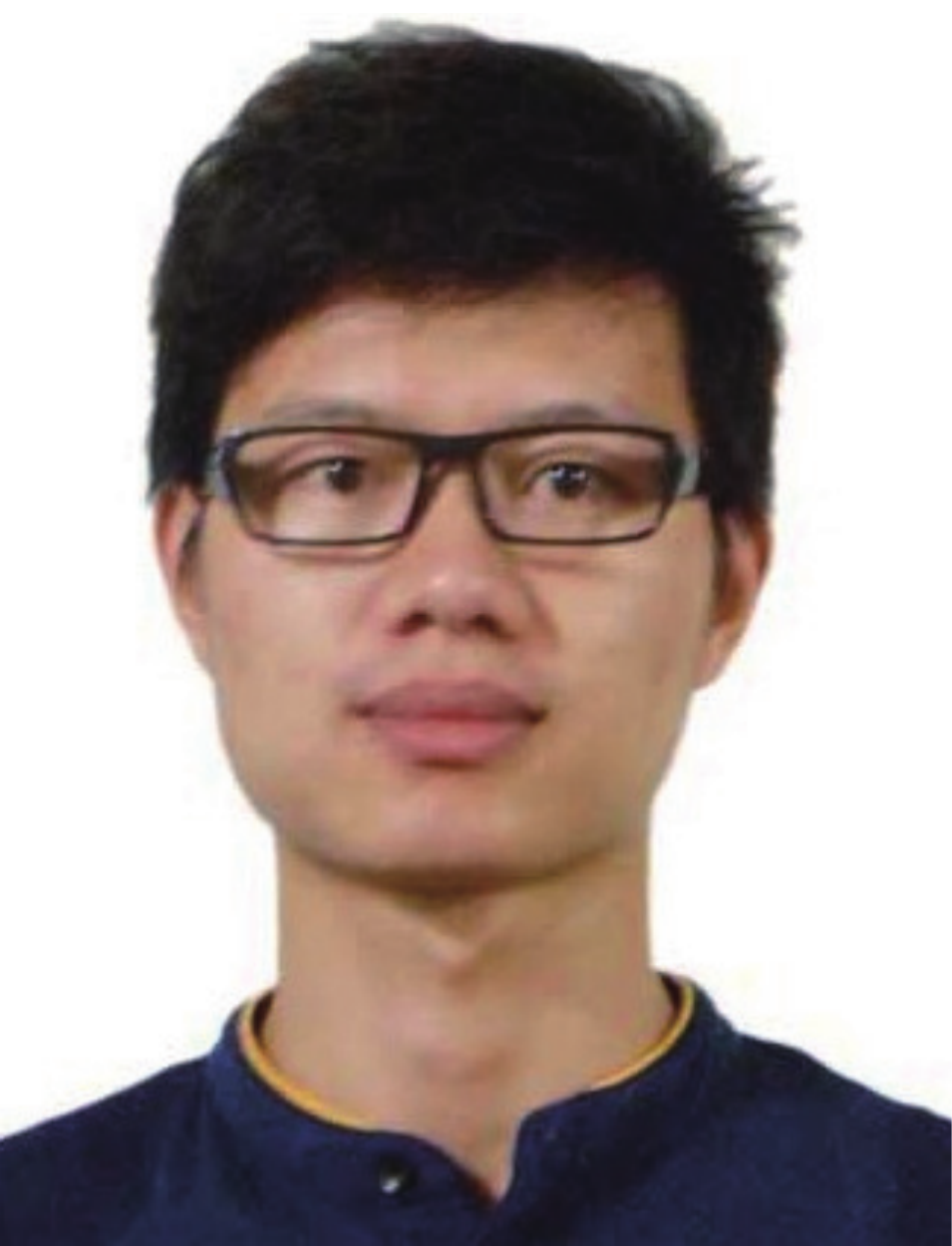}}]
{Weiwei Liu} received the B.S degree in computer science from Southeast University in 2011 and the M.Phil degree from Sun Yat-Sen University in 2014. He is currently with Horizon Robotics as a system engineer and major in automatic driving system development, including system design, HD map generating and vehicle trajectory planning.  His research interests includes nature language processing, trajectory planning and ad-hoc network localization.
 \end{IEEEbiography}
 \vspace*{-40pt}
\begin{IEEEbiography}[{\includegraphics[width=1in,height=1.25in,clip,keepaspectratio]{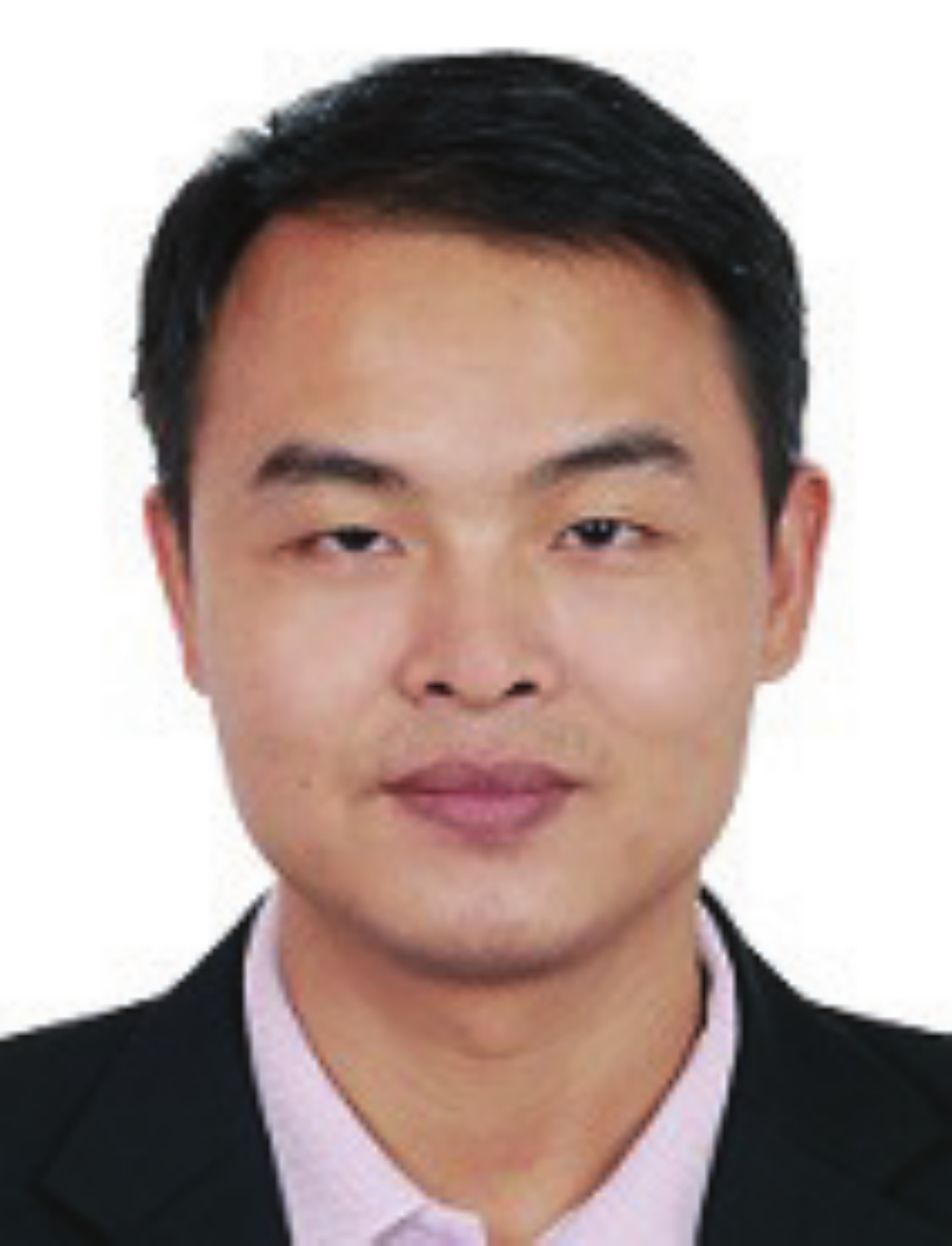}}] {Lvzhou Li} Lvzhou Li received the Ph.D. degree in computer science from Sun Yat-Sen University, Guangzhou, China, in 2009.\\
He is currently an Associate Professor with the School of Data and Computer Science, Sun Yat-sen University, Guangzhou, China. His current research interests include theoretical computer science and quantum computing.
\end{IEEEbiography}
\vspace*{-40pt}
\begin{IEEEbiography}[{\includegraphics[width=1in,height=1.25in,clip,keepaspectratio]{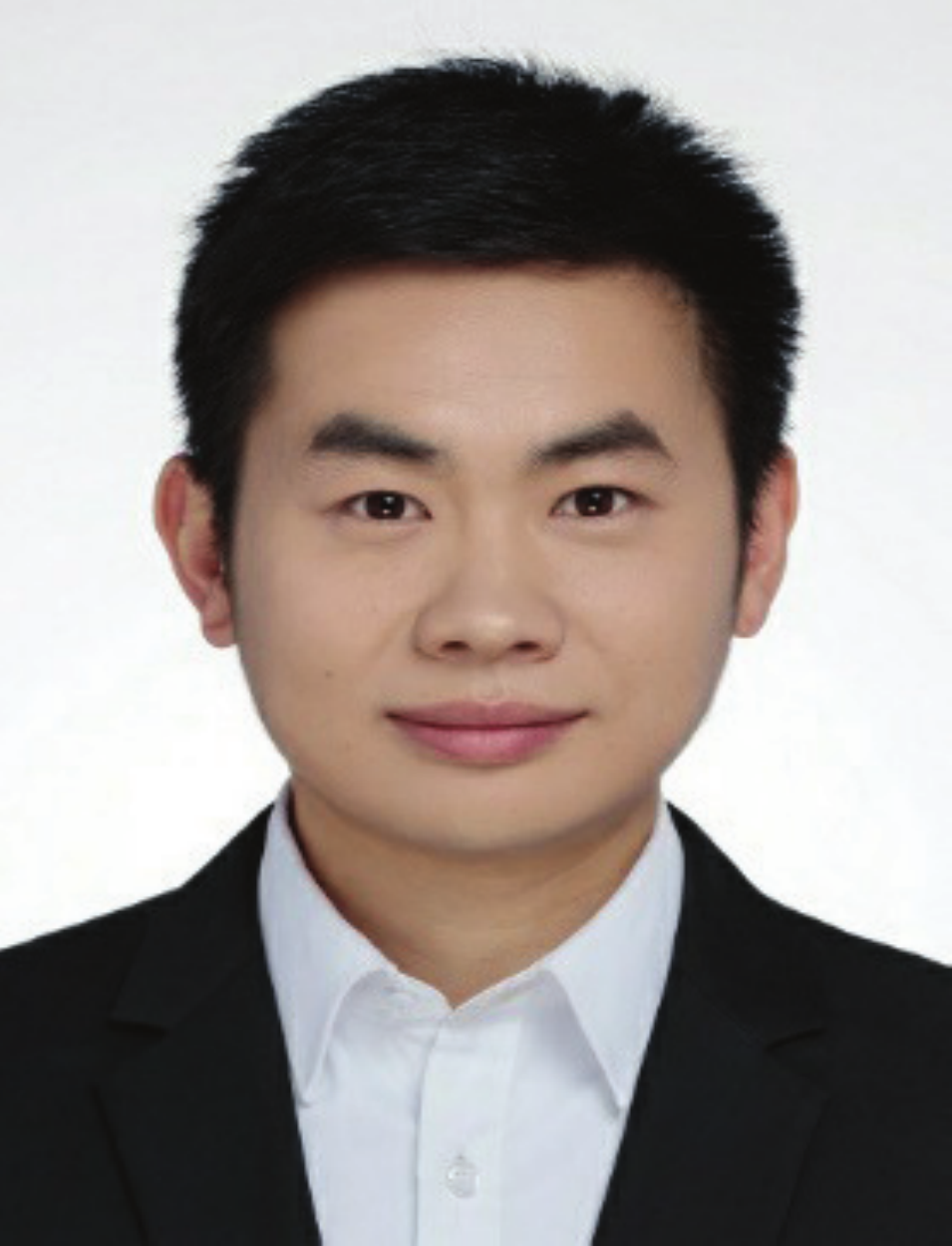}}] 
{Zheng Yang (M'11)} received the B.E. degree in computer science from Tsinghua University in 2006 and the Ph.D. degree in computer science from Hong Kong University of Science and Technology in 2010. He is currently an associate professor with Tsinghua University. His main research interests include wireless ad-hoc/sensor networks and mobile computing. He is a member of the IEEE and ACM.
\end{IEEEbiography}

\end{document}